\documentclass[table]{lmcs} 
\pdfoutput=1

\usepackage[utf8]{inputenc}

\usepackage{lastpage}
\lmcsdoi{20}{1}{17}
\lmcsheading{}{\pageref{LastPage}}{}{}%
{Mar.~09,~2023}{Feb.~29,~2024}{}


\usepackage{macros}

\usepackage{amssymb}
\usepackage{array}
\usepackage{booktabs}
\usepackage{makecell}
	
\begin{document}
		
\title[Stable species of structures]{Stabilized profunctors \\and stable
  species of structures}
\thanks{The first author was partially supported by the EPSRC grant
  EP/V002309/1.  The second author was supported by the PPS ANR project
  (ANR-19-CE48-0014), the EPSRC grant EP/V002309/1, and the PGRM Emergence
  project. The third author was supported by a Royal Society University
  Research Fellowship,  and a Marie Sk\l{}odowska-Curie grant co-funded by the
  Paris Region and the EU under Horizon 2020.}	
		
\author[M.~Fiore]{Marcelo Fiore\lmcsorcid{0000-0001-8558-3492}}[a]
\author[Z.~Galal]{Zeinab Galal\lmcsorcid{0009-0008-6402-3531}}[b]
\author[H.~Paquet]{Hugo Paquet\lmcsorcid{0000-0002-8192-0321}}[c]

\address{University of Cambridge}	
\email{marcelo.fiore@cl.cam.ac.uk}  

\address{Sorbonne University}	
\email{zeinab.galal@lip6.fr}  

\address{ University of Oxford}	
\email{hugo.paquet@cs.ox.ac.uk}  
		
\begin{abstract}
  We define a new bicategorical model of linear logic that refines the
  bicategory of groupoids, profunctors, and natural
  transformations. An object of this new model is a groupoid with
  additional structure, called a \emph{Boolean kit}, used to constrain 
  profunctors by stabilizing the groupoid action on their elements. 
  
  The theory of generalized species of structures, based on
  profunctors, is then refined to a new theory of \emph{stable
    species} of structures between groupoids with Boolean kits. While
  generalized species are in correspondence with analytic functors
  between presheaf categories, stable species are shown to correspond
  to \emph{stable functors} between full subcategories of presheaves
  determined by the kits. One motivation for stable functors is that,
  in the special case where kits enforce free actions, they correspond
  to finitary polynomial functors between categories of indexed sets,
  also known as normal functors. 

We show that the bicategory of groupoids with Boolean kits, stable species,
and natural transformations is cartesian closed.  This makes essential use of
the logical structure of Boolean kits and explains the well-known failure of
cartesian closure for the bicategory of finitary polynomial functors (between
categories of indexed sets) and cartesian natural transformations.  The
paper additionally develops the model of classical linear logic underlying the
cartesian closed structure and clarifies the connection to stable domain
theory.  
\end{abstract}

\maketitle

\section{Introduction}
\label{sec:intro}
\newcommand{\funto}{\to} 

Profunctors, sometimes called distributors or bimodules, are categorical
structures appearing in various areas of logic and computer science.  This
paper is concerned with profunctors between \emph{groupoids} and their
application to linear logic and structural combinatorics.  A 
main goal is to develop a formal connection with the theory of
\emph{polynomial functors}~\cite{PolynomialGpdKock,AbbottContainers}. 
Let us 
summarize the context and the main results of the paper.

\subsection*{Profunctors in logic and combinatorics}
A profunctor from $\gpd1$ to $\gpd2$, where $\gpd1$ and $\gpd2$ are
categories, is a functor 
$\gpd2^\op \times \gpd1 \to \Set$~\cite{Yoneda,BenabouBicat,BenabouDist,Lawvere}.  
It is common to define a bicategory whose objects are small categories,
morphisms are profunctors, and 2-cells are natural transformations; in this
paper we will call $\Prof$ the full sub-bicategory whose objects are
\emph{groupoids}. 

The bicategory $\Prof$ has well-known logical and combinatorial
features~\cite{TaylorGroupoidsLinearLogic,FioreCartesian2008} that we briefly
summarize now; more details are given in Section~\ref{subsec:profunctors}.
First, $\Prof$ is a model of linear logic~\cite{GIRARDLL}: the various logical
connectives are defined with sums, products, or opposites of groupoids, and
the exponential modality $\Sym \gpd1$ is constructed by considering symmetric
sequences of objects of $\gpd1$.  We can then explain the combinatorial nature
of $\Prof$ in this language: if $\One$ is the trivial one-object groupoid,
then profunctors in $\Prof(\Sym \One, \One)$ correspond to combinatorial
\emph{species of structures} in the sense introduced by
Joyal~\cite{JoyalSpecies}. The coKleisli bicategory $\Prof_\Sym$ is known as
the bicategory of 
\emph{generalized species of structures}~\cite{FioreCartesian2008}, since one
can consider combinatorial species between groupoids other than~$\One$.

The bicategory $\Prof_\Sym$ is cartesian closed~\cite{FioreCartesian2008} and
forms the basis for a number of models for typed and untyped
$\lambda$-calculus.  These models can be used, for instance, to track the
symmetries of resource usage~\cite{OlimpieriIntersection} or the combinatorics
of reduction paths~\cite{TsukadaRigid,TsukadaWeightedSpecies2018}.

\subsection*{Polynomial functors and analytic functors}
The starting point for this paper is a connection between Joyal's
combinatorial species and polynomial functors on sets, that we review briefly
(Section~\ref{sec:polynomialAnalytic} contains full details).

On one side, we consider \emph{finitary} polynomial functors $\Set \to \Set$,
which correspond to operators on sets of the form 
\begin{equation}
\label{eq:polyrep1}
X \longmapsto \sum_{n \in \N} A_n \times X^n
\end{equation}
where the coefficients $A_n$ are sets. On the other side, a combinatorial
species consists of a family of sets $F_n$, for $n\in \N$, equipped with a
left action of the symmetric group $\symgroup{n}$ on $n$ elements. This
determines a functor $\Set \to \Set$ of the form 
\begin{equation}
\label{eq:anarep}
X \longmapsto \sum_{n \in \N} F_n \timesquotient{\symgroup{n}} X^n
\end{equation}
where the operator
\raisebox{0.3em}{$\timesquotient{\symgroup{n}}$} performs a quotient of the
product under the action.  A functor of this form is known as an 
\emph{analytic functor}, and every analytic functor has a unique generating
species, up to isomorphism~\cite{JoyalAnalytic}. 

In special cases, when the actions on $F_n$ are 
\emph{free actions}~(\S\ref{eq:freeaction}), the quotient is equivalently a
set of the form $A_n \times X^n$, and so the analytic functor is also
polynomial.  Conversely, every finitary polynomial functor is analytic when
its coefficients are regarded as freely generated actions.

Thus, finitary polynomial functors determine a subcategory of 
$\Prof(\Sym \One, \One)$, and the language of species gives a basic
combinatorial perspective on polynomial functors. Our purpose here is to
generalize this connection, and to provide a new \emph{logical} perspective,
following the methodology offered by the logical structure of profunctors and
generalized species. 

\subsection*{Contributions of this paper}
We extend the correspondence between finitary polynomial functors and free
analytic functors to a wider setting: instead of functors over the category of
sets, we consider functors between full subcategories of presheaves over
groupoids.  Our first contribution is a logical device for constraining the
actions on the coefficients of analytic functors: in particular one may
require all actions to be free. We call this device a 
\emph{kit}~(Section~\ref{sec:kits}).

In Section \ref{sec:stprof}, we define a refinement of $\Prof$ called
$\SProf$, in which groupoids are equipped with kits in order to constrain the
profunctors between them.  Profunctors that respect the kit structure are
called \emph{stabilized}.  The bicategory $\SProf$ has the following key
properties: it is a model of linear logic, and the logical structure is
preserved by the forgetful functor 
$\SProf \to \Prof$~(Section~\ref{sec:logical-structure-sprof}).  The category
$\SProf(\Sym \One, \One)$  is the sub-category of $\Prof(\Sym \One, \One)$
that corresponds to finitary polynomial functors on sets. 

Thus, we can use the logical structure of $\SProf$ as a new language for
polynomial functors 
operating on groupoid actions. 
We recover well-known notions of polynomial functors operating on
sets~\cite{GirardNormal,AltenkirchIndexed} and we obtain a new notion of
polynomial functors of higher-order\footnote{Note that this `higher-order'
  structure refers to a bicategory whose \emph{morphisms} are polynomial
  functors. Polynomial functors over $\Set$ are also the \emph{objects} of a
  category which is known to be cartesian
  closed~\cite{DBLP:conf/cie/AltenkirchLS10}. 
} 
type. 


For example, the kits can be used to characterize, among the profunctors in
$\Prof(\Sym\One, \One)$, those corresponding to polynomial functors on sets.
We focus on this special case in Section~\ref{sec:kitsIntro}, and give the
general theory in Section~\ref{sec:kits}.  Kits are largely inspired by the
\emph{creeds} of Taylor~\cite{TaylorGroupoidsLinearLogic}, although we have
crucially taken a different approach based on double orthogonality.  As
explained in Section~\ref{sec:*autonomous}, this allows us to model the
involutive negation of classical linear logic. 

\subsection*{Intensional and extensional perspectives}
One important aspect of this paper is the dichotomy between \emph{intensional}
and \emph{extensional} presentations.  This is a recurrent theme in the
semantics of linear logic and in the theory of species, as illustrated by the
table of examples in Figure~\ref{fig:examples}.  Our model of stabilized
profunctors between groupoids with kits can also be presented in extensional
form.  We show that a kit $\kit1$ on a groupoid $\gpd1$ determines a full
subcategory $\StPSh(\gpd1,\kit1)$ of the presheaf category $\PSh(\gpd1)$, and
that a stabilized profunctor $P \in \SProf((\gpd1, \kit1), (\gpd2, \kit2))$ is
equivalently given by a \emph{linear} functor 
$\StPSh(\gpd1, \kit1) \to \StPSh(\gpd2, \kit2)$, where linear means that the
functor preserves filtered colimits, and has a left and a right adjoint on
every slice~(Definition~\ref{def:linearfunctor}). 
\renewcommand{\arraystretch}{2.6}
\begin{figure}
	\begin{tabular}{>{\itshape}ccc}
		{}  & \makecell{\textsc{Extensional} \quad\\ \textsc{presentation}} &
		\makecell{\textsc{Intensional}
			\quad
			\\ \textsc{presentation}} \\[1em]\toprule
		\makecell{Scott \cite{ScottDataTypes}}
		& \makecell{Continuous functions} & \makecell{Finitary graphs
			\\of continuous functions}\\ 
		\makecell{Girard \cite{GirardNormal}} & \makecell{Normal functors \\
			$\Set^I \to \Set$}
		& \makecell{Power series coefficients\\ $\mathcal{M}_\mathrm{fin}(I) \to
			\Set$} \\
		Gambino, Kock \cite{gambinokock2013}
		& \makecell{Polynomial functors\\ $\Set/I \to\Set/J$}
		& \makecell{Polynomial diagrams\\ $I \leftarrow E \to B \to J$} \\
		Joyal \cite{JoyalSpecies} & \makecell{Analytic functors \\
			$\Set\to \Set$ }& \makecell{Species of structures \\$\PP \to
			\Set$ } \\
		\makecell{Morita, Watts\\ Lawvere, B\'enabou}  & \makecell{Colimit-preserving functors \\$\PSh(\gpd1) \to \PSh(\gpd2)$}
		&\makecell{ Profunctors \\$\gpd2^\op \times \gpd1 \to \Set$ }\\
		Fiore et
		al. \cite{Fioreanalytic,FioreCartesian2008}
		& \makecell{Analytic functors\\ $\PSh(\gpd1) \to\PSh(\gpd2)$}
		& \makecell{Generalized species\\ $\gpd2^\op \times \Sym \gpd1 \to \Set$} \\[0.7em]
		\bottomrule
	\end{tabular}
	\label{fig:examples}
  \caption{Well-known categorical structures in extensional and intensional
    form. In the table, $\mathbb{A}$ and $\mathbb{B}$ are groupoids and
    $\mathcal M_{\mathrm{fin}}$ is the finite multiset construction.}
\end{figure}

From the combinatorial perspective, it is the coKleisli bicategory
$\SProf_\Sym$ for the exponential comonad which is of primary interest,
because of the connection to generalized species.  We show that the morphisms
in that bicategory can be presented extensionally as \emph{stable} functors
$\StPSh(\gpd1, \kit1) \to \StPSh(\gpd2, \kit2)$, where stability relaxes
linearity by not requiring a right adjoint on each slice and requiring the
preservation of epimorphisms~(Definition~\ref{def:stablefunctor}).  This
terminology comes from stable domain
theory~\cite{BerryStable,TaylorAlgebraicStableDomains,LamarcheThesis}, as we
discuss further in Section~\ref{sec:Stability}.

The conference version of this paper focused on the bicategory 
$\SEsp = \SProf_\Sym$, described directly, and on the correspondence with
stable functors~\cite{fiore2022combinatorial}.  This longer version provides
all details, including the underlying linear theory of stabilized profunctors. 

\subsection*{Structure of the paper}
Beyond the two background sections that follow, the paper is organized into
two main parts.  In the first part~(Sections~\ref{sec:kits-and-sprof}
to~\ref{sec:stable-species}), we regard profunctors as groupoids actions, and
introduce kits to control and stabilize these actions.  We study kits from a
logical perspective, and introduce an orthogonality relation for families of
subgroups of endomorphisms in a groupoid.  The induced Boolean algebra
structure gives rise to the notion of Boolean kit. We show that the bicategory
of Boolean kits and stabilized profunctors is a model of differential linear
logic. 

In the second part~(Sections~\ref{sec:stable-presheaves}
to~\ref{sec:Linearity}), we study the extensional aspects of stabilized
profunctors and stable species. We show that orthogonality for kits can be
translated to an orthogonality relation on presheaves.  We then prove that
stabilized profunctors and stable species can be characterized respectively as
linear and stable functors between subcategories of presheaves closed under
double-orthogonality.  This provides a 2-dimensional stable domain theory.

\section{Motivation: polynomial and analytic functors over sets}\label{sec:polynomialAnalytic}

At the heart of this paper is the relationship between two families of
functors from $\Set$ to $\Set$: finitary polynomial functors (also known as
\emph{normal} functors) and analytic functors. We start by giving the
relevant definitions.

\subsection*{Polynomial functors.} A 
function 
$
p : E \to B
$ 
between sets $E$ and $B$ determines an endofunctor on $\Set$ defined
as 
\[ 
X \longmapsto \sum_{b \in B} X^{E_b}
\]
where $E_b = p^{-1}\{b\}$ is the fibre of $p$ over $b \in B$, and $X^{E_b}$ is
the set of functions $E_b \to X$. A functor of this form is called a polynomial
functor. If we restrict to those $p$ such that
$E_b$ are finite sets for all $b \in B$, then the functor is naturally isomorphic to one of the
form 
\begin{equation}
\label{eq:polyrep2}
X \longmapsto \sum_{n \in \N} F_n \times X^n 
\end{equation}
where each set $F_n$ corresponds to the set of all $b \in B$ such that
$|E_b| = n$. The
analogy with traditional polynomials is manifest in this representation.
These polynomial functors are called \emph{finitary} and are determined by their action on finite input
sets. 

\subsection*{Analytic functors} 
We now turn to a strictly more general family of functors. An endofunctor on $\Set$ is an \emph{analytic functor} \cite{JoyalAnalytic} if
it is naturally isomorphic to one of the form 
\begin{equation}
\label{eq:analyticrep}
X \longmapsto \sum_{n \in \N} F(n) \timesquotient{\symgroup{n}} X^n , 
\end{equation}
where:
\begin{enumerate}
\item 
Each $F(n)$ is a set with a left action of the symmetric group $\symgroup{n}$
on $n$ elements. Concretely, this means that we have an assignment that to
every permutation $\sigma\in\symgroup{n}$ of the set 
$ [n] = \{ 1, \dots, n \} $ associates a permutation of the set $F(n)$
preserving identity and composition.  We write $\sigma \act p$ for the action
of $\sigma\in\symgroup{n}$ on an element $p \in F(n)$. 

\item 
The set $F(n) \timesquotient{\symgroup{n}} X^n$ is obtained by quotienting the
product $F(n) \times X^n$ under the equivalence relation $\sim$ containing the
pairs 
\[ 
(p, (x_{\sigma 1}, \dots, x_{\sigma n})) 
\enspace \sim \enspace 
(\sigma \act p, (x_1, \dots, x_n))
\]  
for all $\sigma \in \symgroup{n}$, $p \in F(n)$ and 
$(x_1, \dots, x_n) \in X^n$. 
\end{enumerate}

\noindent The notation $F(-)$ is justified, because the coefficients $F(n)$,
\emph{together with} the group actions, can be bundled into a functor 
$F : \PP \to \Set $ where $\PP$ is the category whose objects are
the natural numbers, and whose morphisms $m \to n$ are the bijections 
$[m] \to [n]$. 
This category is equivalent to the category of finite sets and bijections. The action of a permutation $\sigma \in \symgroup{n}$ on the
set $F(n)$ is then simply given by the functorial action 
$F(\sigma) : F(n) \to F(n)$.

The functor $F : \PP \to \Set$ is a \emph{species of structures} (or just a
\emph{species}, with its elements referred to as \emph{structures})
corresponding to the analytic functor \eqref{eq:analyticrep}.  Every analytic
functor has, up to isomorphism, a unique generating species, which may be
recovered using so-called \emph{weak generic elements}
(Section \ref{subsec:LRAgenericfact}).

For some more combinatorial intuition, 
each $F(n)$ is intended to model a type of structures (trees, partitions, $\dots$) with
$n$ labels (indexed $1$ to $n$) on which the symmetric group
acts by permutation. The induced analytic functor then constructs for each set $X$
the set of $X$-labelled $F$-structures. These are elements of $F(n)$
at any $n$ in which labels are elements of $X$.
This theory was developed by Joyal
\cite{JoyalSpecies,JoyalAnalytic} for application in
combinatorics. Joyal also noticed the connection to polynomial
functors that we explain next. 

\subsection*{Finitary polynomial functors are free analytic functors}
Finitary polynomial functors can be identified with  the sub-family of
\emph{free} analytic
functors.  The basic idea is as follows. Every set $A$ generates a
\emph{free action} of a group $G$, given by the product set $A \times G$ with
the action
\begin{equation}
\label{eq:freeaction}
\tau \act (a, \sigma) \   \eqdef \  (a, \tau \icomp \sigma)
\end{equation}
for every $\tau \in G$ and $(a, \sigma) \in A \times G$.
We extend this to polynomial functors.  Consider the finitart polynomial functor $X
\longmapsto \sum_{n \in \N} A_n \times X^n$.  Taking the free action generated
by $A_n$ of $\symgroup{n}$, for every $n \in \N$, we obtain a species 
$\PP \longto \Set$ given by 
\begin{align*}
n &\longmapsto A_n \times \symgroup{n} \\
(\tau : n \to n) &\longmapsto ((a, \sigma) \mapsto (a, \tau \icomp \sigma))
\end{align*}
which, via the construction \eqref{eq:analyticrep}, generates an analytic
functor.  For this species, when taking the quotient in
\eqref{eq:analyticrep}, we have 
\[ 
(A_n \times \symgroup{n}) \timesquotient{\symgroup{n}} X^n 
\ \cong \
A_n \times X^n
\] 
and therefore recover the polynomial functor.  Thus, every finitary
polynomial functor is, in particular, analytic.

\subsection*{Stabilizer subgroups and free species}

In this paper, ``finitary polynomial'' and
``analytic'' are viewed as the two extremes in a spectrum of classes of functors
$\Set \to \Set$. The intermediate classes are defined in terms of
the underlying species of structures, using a new notion we call a
kit. 

Recall that, if a group $G$ acts on a set $P$, the \emph{stabilizer} of an element $p \in P$, defined as $\Stab[G](p) \eqdef \{ \sigma \in G\mid \sigma \act p = p \}$, is a
subgroup of $G$.  The key observation is that free actions, as in
\eqref{eq:freeaction}, are exactly those for which every element has trivial
stabilizer.
Thus, if $F : \PP \to \Set$ is a species such that for every $n \in \N$, the
action of $\symgroup{n}$ on $F(n)$ is \emph{free} in the sense that every structure
in $F(n)$ has trivial stabilizer, then the associated analytic endofunctor on
$\Set$ is polynomial.
This characterizes the analytic functors that are polynomial in terms
of their generating species. 

\subsection*{Kits on $\PP$}\label{sec:kitsIntro}

One key idea of this paper is to use families of subgroups to specify the extent to which
a species may be free. Indeed, by specifying, for each $n \in \PP$, a set
$\mathcal{K}(n)$ of subgroups of $\PP(n,n)$ that are to be regarded as
\emph{permitted stabilizers}, we may restrict to species with structures
having only permitted stabilizers and thereby
identify a class of functors $\Set \to \Set$.  

As extreme special cases, one can take $\mathcal{K}(n)$ to contain all the
subgroups of $\PP(n,n)$ and recover the analytic functors; or, instead, take
$\mathcal{K}(n)$ to consist only of the trivial subgroup, forcing the species
to be free, and recover polynomial functors. Such 
families of subgroups $\mathcal K = \{\, \mathcal{K}(n) \,\}_{n \in
  \PP}$, subject to a compatibility condition, will be called kits
(Definition~\ref{def:kit}). 

There is no need that the choice of permitted stabilizers be uniform across
all $n \in \PP$ as in the two examples above.  But it is natural to require
that permitted stabilizers are closed under conjugation, since for every 
structure $p \in F(n)$ and permutation $\sigma \in \PP(n, n)$, the stabilizer
subgroups of $p$ and of $\sigma \act p$ are conjugate of each other:
$\Stab(\sigma \act p) = \sigma \,\icomp \Stab(p) \,\icomp
\inv\sigma$.

First we appropriately restrict species: 
\begin{defi}
\label{def:Sspecies}
For each $n \in \PP$, let $\mathcal{K}(n)$ be a family of subgroups of
$\PP(n, n)$ closed under conjugation, and write $\mathcal{K} = \{
\mathcal {K}(n) \}_{n \in \PP}$. A \textbf{$\mathcal{K}$-species} is a
functor $F : \PP \funto \Set$ such that every element of $F(n)$ has stabilizer in
$\mathcal{K}(n)$.  
\end{defi}

 This way, every family $\mathcal{K}$ determines a subclass of
 analytic endofunctors on $\Set$, namely, those whose generating species are
$\mathcal{K}$-species. We will call such a family a kit on the
groupoid $\PP$.

By varying the kit structure, one obtains different subclasses of
analytic functors on $\Set$. This paper is primarily
about the connection to polynomial functors, so our construction of \emph{stable}
species in terms of a linear exponential
comonad (Section~\ref{sec:stable-species}) corresponds to the kit
on $\PP$ that consists of only the trivial subgroups. On the other
hand, Joyal's species
are arbitrary functors and correspond to the maximal kit containing
all subgroups.


\section{Preliminaries on profunctors and generalized species}
\label{sec:preliminaries-prof-esp}

This paper extends the connection between polynomial functors and
analytic functors to a generalized setting based on profunctors.
In this section, we recall some important background: we describe the
construction of a bicategory of profunctors
(e.g.~\cite{Yoneda,BenabouBicat,BenabouDist,Lawvere}) and of the
induced bicategory of generalized species.

\subsection{Profunctors}
\label{subsec:profunctors} 

For groupoids $\gpd1$ and $\gpd2$, a profunctor $P : \gpd1 \profto
\gpd2$ is a functor $\gpd2^{\op}\times \gpd1 \to \Set$. This is a
family of sets indexed by pairs of objects $a \in \gpd1$ and $b \in \gpd2$,
and with an action of the two groupoids. We use the following
notation: if $(b,a)\in
\gpd2^{\op}\times \gpd1$, $p \in P(b,a)$, $\alpha \in \gpd1(a,a')$ and
$\beta \in \gpd2(b',b)$, we write 
\[\begin{aligned}
	\alpha \cdot p &:= P (\id[b], \alpha)(p)\in P(b,a') \\
	p \cdot \beta &:= P(\beta,\id[a])(p)\in P(b',a)
\end{aligned}\]
and by functoriality, the two actions commute. 

A profunctor $\gpd1 \profto \gpd2$ is equivalently a functor $P : \gpd1 \rightarrow
\PSh(\gpd2)$, and it is helpful to think of this as a Kleisli arrow for the
presheaf construction $\PSh$ (\emph{cf.}~\cite{RelativePseudomonads}).
Since $\gpd1$ is a dense subcategory of $\PSh(\gpd1)$,
we can define a functor $\PSh(\gpd1) \to \PSh(\gpd2)$ as the unique
colimit-preserving functor that extends $P$. Formally this is a left
Kan extension, as we use next.  

\subsubsection*{The bicategory of profunctors over groupoids}
Two profunctors $P: \gpd1 \profto \gpd2$ and $Q: \gpd2 \profto \gpd3$
compose to give $Q \circ P : \gpd1\profto \gpd3$, defined as the
composite functor $(\Lan_{\yon[\gpd2]}Q)P$ where $\Lan_{\yon[\gpd2]}Q$ denotes the left Kan extension of $Q$ along the Yoneda embedding $\yon[\gpd2]: \gpd2 \hookrightarrow \PSh{\gpd2}$:
\begin{center}
	\begin{tikzpicture}
		begin{tikzpicture}[thick]
		\node (A) at (0,0) {$\gpd1$};
		\node (B) at (2,1.5) {$\gpd2$};
		\node (PB) at (2,0) {$\PSh{\gpd2}$};
		\node (PC) at (4,1.5) {$\PSh{\gpd3}$};
		\node at (2.6,0.9) {$\Downarrow$};
		
		\draw [->] (A) -- node [above] {$P$} (PB);
		\draw [right hook->] (B) -- node [left] {$\yon[\gpd2]$} (PB);
		\draw [->] (B) -- node [above] {$Q$} (PC);
		\draw [->, dashed] (PB) -- node [below right] {$\Lan_{\yon[\gpd2]}Q$} (PC);
	\end{tikzpicture}
      \end{center}
There is a more concrete \emph{pointwise} formula for $Q \circ P$ in terms of a coend:
\[
(Q \circ P)(c,a) \quad \eqdef \quad \int^{b\in\gpd2} P(b,a) \times Q(c,b) \cong \left(\coprod\limits_{b\in \gpd2}P(b,a) \times Q(c,b)\right)/_{\sim}
\]
where $\sim$ is the least equivalence relation such that 
\[
(b, p \cdot \beta, q) \sim (b', p, \beta \cdot q)
\] for $p \in P(b',a)$, $q \in Q(c,b)$ and $\beta: b\rightarrow b' \in
\gpd2$.

The identity profunctor is the hom-functor $\gpd1^\op \times \gpd1
\to \Set$, defined as $(a, a') \mapsto \gpd1(a, a')$.
The composition of profunctors is only associative and unital up to
natural isomorphism, and thus we have the following \emph{bicategory}:
\begin{defi}
	We denote by $\Prof$ the bicategory given by:
	\begin{itemize}
		\item \textbf{objects:} small groupoids $\gpd1$, $\gpd2$;
		\item \textbf{$1$-cells:} profunctors $P : \gpd1 \profto \gpd2$;
		\item \textbf{$2$-cells:} natural transformations.
	\end{itemize}
\end{defi}

\subsubsection*{Profunctors in extensional form} As we mentioned, a
profunctor $P : \gpd1 \profto \gpd2$ induces a unique colimit-preserving
functor 
$\PSh(\gpd1) \to \PSh(\gpd2)$ between the presheaf
categories (Theorem~\ref{thm:BiequivalenceProfCocont}). In this representation, the composition of profunctors
reduces to a composition of functors, so we have a 2-category. 

\begin{defi}
	We denote by $\Cocont$ the 2-category given by:
	\begin{itemize}
		\item \textbf{objects:} small groupoids $\gpd1$, $\gpd2$;
		\item \textbf{$1$-cells:} cocontinuous functors $F : \PSh{\gpd1} \to \PSh{\gpd2}$;
		\item \textbf{$2$-cells:} natural transformations.
	\end{itemize}
\end{defi}

\begin{thm}\label{thm:BiequivalenceProfCocont}
	The bicategory $\Prof$ and the $2$-category $\Cocont$ are biequivalent.
\end{thm}

\subsection{Generalized species}
\label{subsec:generalized-species}

We now describe a bicategory of generalized species of structures indexed by arbitrary
groupoids. 

\subsubsection*{From species to generalized species}
For a species of structures $\B \to \Set$, consider the following two basic observations. First,
the groupoid $\B$ is the \emph{free symmetric strict monoidal completion} of
the terminal category $\One$. This completion is denoted $\Sym \One$
and described formally below. Second, $\Set$
is isomorphic to the category of presheaves $\PSh \One$ over $\One$. 
The  version, due to Fiore, Gambino, Hyland and
Winskel~\cite{FioreCartesian2008} extends the basic notion of a species
$\B \funto \Set$, corresponding to 
$\Sym \One \to \PSh \One$, to 
\begin{equation}
\label{eq:generalisedspecies}
\Sym \gpd1 
\longrightarrow 
\PSh \gpd2
\end{equation}
for $\gpd1$ and $\gpd2$ groupoids (in fact, they can be arbitrary small
categories, but we do not use this generality here).  Their main result is that
these assemble into a bicategory $\Esp$ of groupoids, generalized species, and
natural transformations, and that moreover this bicategory is cartesian closed. 

Since $\PSh(\gpd2) = [\gpd2^\op, \Set]$, a generalized species is
equivalently a profunctor $\gpd2^\op \times \Sym\gpd1 \to \Set$, and
the bicategory $\Esp$ can be obtained as a coKleisli bicategory for the
pseudo-comonad $\Sym$ over the bicategory of profunctors. This is the
approach we follow in this paper.

\subsubsection*{The symmetric strict monoidal completion}
 For a category $\gpd1$, define $\Sym \gpd1$ as the category whose objects are finite sequences $\seq{a_1, \dots, a_n}$ of objects of $\gpd1$ and a morphism $\alpha$ between two sequences $\seq{a_1, \dots, a_n}$ and $\seq{b_1, \dots b_n}$ consists of a pair $(\sigma, (\alpha_i)_{i \in \ints{n}})$ where $\sigma$ is a permutation of the set $\ints{n} = \{ 1, \dots, n\}$ and $(\alpha_i : a_i \rightarrow b_{\sigma(i)})_{i \in \ints{n}}$ is a sequence of morphisms in $\gpd1$. 
There are no morphisms in $\Sym \gpd1$ between sequences of different
lengths.

If $\gpd1$ is a groupoid, then $\Sym \gpd1$ is also a groupoid,
equipped with a symmetric monoidal structure given by the
concatenation of sequences $(u,v) \mapsto u \otimes v$. This
construction has a canonical $2$-monad structure on $\Cat$, that can
be lifted to a pseudo-comonad on $\Prof$, see
\emph{e.g.}~\cite{RelativePseudomonads}. It maps a profunctor $P : \gpd1 \profto \gpd2$ to the profunctor $\Sym P: \Sym \gpd1 \profto \Sym \gpd2$ given by:
\[
\Sym P(\seq{a_i}_{i \in \ints{n}}, \seq{b_j}_{j\in \ints{m}}) = \begin{cases} \coprod_{\varphi \in \symgroup{n}} \prod_{j \in [m]} P(b_j, a_{\varphi(j)}), & \text{if } n=m\\
	\varnothing, & \text{otherwise}
\end{cases}
\]
The counit and
comultiplication, denoted $\der$ and $\dig$ respectively (for
\emph{dereliction} and \emph{digging}, following the usual terminology of linear logic) are pseudo-natural transformations whose components are the profunctors 
 \[
 \begin{aligned}
 	\der[\gpd1] : \Sym \gpd1   &\profto   \gpd1\\
 	(a, u) &\mapsto \Sym\gpd1 (\seq{a}, u)\\
 	\dig[\gpd1]:  \Sym \gpd1 &\profto \Sym \Sym \gpd1\\
 	(\seq{u_1, \dots, u_n}, u)& \mapsto \Sym \gpd1 (u_1 \otimes \dots \otimes u_n, u).
\end{aligned}
\]

\subsubsection*{The bicategory of species}
Generalized species of structures are the
morphisms in the coKleisli bicategory $\Prof_\Sym$, which we call
$\Esp$. In other words, $\Esp(\gpd1, \gpd2)$ consists of profunctors of the form $\Sym \gpd1
\profto \gpd2$ and natural transformations between them; observe that the category $\Esp(\One, \One)$
is the category of combinatorial species in the sense of Joyal.

Just as a species $\B \to \Set$ induces an analytic functor $\Set \to \Set$,
a generalized species $F:\Sym\gpd1 \funto \PSh(\gpd2)$ induces an analytic
functor $\PSh(\gpd1) \funto \PSh(\gpd2)$ that transports a presheaf $X
\in \PSh(\gpd1)$ to the presheaf
\[
b \quad \mapsto\quad \coend^{\seq{a_i} \in \Sym \gpd1}  F(b, \seq{a_i}) \times
\PSh(\gpd1)( \coprod_{i=1}^n \yon(a_i), X) 
\]
for $b \in \gpd2$. This formula, and indeed the earlier one for Joyal species
\eqref{eq:analyticrep}, are obtained as left Kan extensions:
\begin{prop}
For groupoids $\gpd1, \gpd2$, and a generalized species 
$F: \Sym \gpd1 \profto  \gpd2$, the analytic functor described pointwise above
is a left Kan extension of $F$ along the functor 
$s
 : \Sym \gpd1 \funto \PSh(\gpd1) 
 : \seq{a_1, \dots, a_n} \longmapsto \coprod_{i=1}^n \yon(a_i)$, 
as on the left below: 
  \begin{center}
		\begin{tikzpicture}[
      scale=0.9]
		\node (A) at (0,1.25) {$\Sym\gpd1$};
		\node (B) at (3,1.25) {$\PSh(\gpd2)$};
		\node (C) at (1.5,0) {$\PSh(\gpd1)$};
		\node (D) at (1.5,0.75) {$\Downarrow$};
		
		\draw [->] (A) to node [above] {$F$} (B);
		\draw [<-, dotted] (B) to node [below right]  {$\Lan_{s} F$} (C);
		\draw [right hook->] (A) to node [below left] {$s$} (C);
		\end{tikzpicture}
		\qquad \qquad \qquad 
  	\begin{tikzpicture}[
      scale=0.9]
		\node (A) at (0,1.25) {$\B$};
		\node (B) at (3,1.25) {$\Set$};
		\node (C) at (1.5,0) {$\Set$};
		\node (D) at (1.5,0.75) {$\Downarrow$};
		
		\draw [->] (A) to node [above] {$F$} (B);
		\draw [<-, dotted] (B) to node [below right]  {$\Lan_{j} F$} (C);
		\draw [right hook->] (A) to node [below left] {$j$} (C);
		\end{tikzpicture}
	\end{center}
For $\gpd1 = \gpd2 = \One$ and $F$ viewed as a species $\B \funto \Set$, this
corresponds to the diagram on the right above, where $j$ is the inclusion
functor, and induces the formula \eqref{eq:analyticrep}.
\end{prop}

We note that, for a species $F$ in $\Esp(\gpd1, \gpd2)$ the
induced analytic functor is the same, under the equivalence
$\PSh(\gpd1) \simeq \Esp(\Zero,\gpd1)$ for $\Zero$ the empty groupoid, as the post-composition functor 
\[
F\circ - : \Esp(\Zero, \gpd1) \to \Esp(\Zero, \gpd2). 
\]

Analytic functors between presheaf categories over groupoids can be
characterized directly, thus giving an extensional presentation of
$\Esp$ as a strict 2-category \cite{Fioreanalytic}. We will not need
this characterization in the paper, so we omit it.

\section*{\Large Intensional Theory:\texorpdfstring{\\}{ }Stabilized Profunctors and Stable Species}

We proceed to develop our new model based on groupoids
with kits. We first study profunctors and species over groupoids
with kits; this is what we call the ``intensional theory''. This
provides a refinement of the bicategories $\Prof$ and $\Esp$. 

\section{Kits and stabilized profunctors}
\label{sec:kits-and-sprof}

\subsection{Groupoids and Kits}\label{sec:kits}
We first make some basic remarks about groupoids, and fix some notation. For every object $a$ in a groupoid $\gpd1$, the set of endomorphisms $\Endo[\gpd1](a) = \gpd1(a, a)$ forms a group under composition. The groups $\Endo[\gpd1](a)$ will play a central role in our development and we often simply write $\Endo(a)$. 

For a group $G$, we write $H \subgroup G$ to mean that $H$ is a subgroup of $G$; a canonical example is the \emph{trivial subgroup} $\{ \id[a] \} \subgroup \Endo[\gpd1](a)$ for an object $a$ in a groupoid $\gpd1$.  
Importantly, subgroups of endomorphisms in $\gpd1$ can be transferred across morphisms by conjugation: for a morphism $\alpha : a \to a'$ and $G \subgroup \Endo(a)$, the conjugate subgroup $\alpha\icomp G\icomp \inv\alpha = \{ \alpha\icomp g\icomp \inv\alpha \mid g \in G \}$ is a subgroup of $\Endo(a')$.

Our approach involves attaching to a groupoid $\gpd1$ some additional structure which we will call a \emph{kit}, comprising a family of subgroups of endomorphisms for each object. 

\begin{defi}
  \label{def:kit}
	A \defn{kit} on a groupoid $\gpd1$ is a family 
	$\kit1 = \setof{ \kit1(a) }_{a\in\gpd1}$ consisting of sets $\kit1(a)$ of
	subgroups of $\Endo[\gpd1](a)$ that is closed under conjugation; that is, such
	that, for all $\alpha: a'\to a$ in $\gpd1$, if $G\subgroup\Endo(a)$ is in
	$\kit1(a)$ then the conjugate subgroup 
	$\inv\alpha \icomp G \icomp \alpha\subgroup\Endo(a')$ is in $\kit1(a')$.  
\end{defi}
When $\kit1$ is a kit on $\gpd1$, we often refer to the pair $(\gpd1, \kit1)$ as a kit, which we also denote by $\kitstr{\gpd1}$ for convenience.
\begin{exa}
\label{ex:canonical_kits}
Every groupoid $\gpd1$ has the following canonical choices of kits: the trivial one $\Triv[\gpd1](a) = \left\{ \{ \id[a]\} \right\}$, and the maximal one $\Endo[\gpd1](a) = \{ G \mid G \subgroup \Endo[\gpd1](a) \}$. 
\end{exa}

As usual with linear logic, our model exhibits several aspects of
duality. Every groupoid $\gpd1$ has an \defn{opposite} or \defn{dual}
groupoid $\gpd1^\op$, to which it is isomorphic via the mapping
$\alpha \mapsto \inv\alpha$, and moreover $(\gpd1^\op)^\op=\gpd1$;
this holds in particular for any group $G$ seen as a one-object groupoid. Observe that for any object $a$ of a groupoid $\gpd1$ one has $\Endo[\gpd1^\op](a) = (\Endo[\gpd1^\op](a))^\op$.

To construct kits on dual groupoids we consider a notion of subgroup orthogonality. 
\begin{defi}
For a group $G$ and subgroups $H \subgroup G$ and $K \subgroup G^\op$, we say that $H$ and $K$ are \defn{orthogonal} if $H \cap K = \{ \id \}$, where $\id$ is the identity element in $G$ (and $G^\op$).
\end{defi} 

From this we get the following construction:
\begin{lem}
For a groupoid $\gpd1$ and a kit $\kit1$, $\gpd1^\op$ may be equipped with its \defn{orthogonal kit} $\kit1^\orth$ defined for $a \in \gpd1^\op$ by:
\[
\kit1^\orth(a) = \{ G \subgroup \Endo[\gpd1^\op](a) \mid \forall H \in \kit1(a), G \orth H \}.
\]
\end{lem}
\begin{proof}
For the closure under conjugation observe that for a morphism $\alpha : a' \to a$ and subgroups $G \subgroup {\Endo}_{\gpd1}(a)$ and $H \subgroup {\Endo}_{\gpd1^\op}(a')$,
	$\inv\alpha\icomp G\icomp\alpha \orth H$ if and only if
	$G \orth \alpha\icomp H\icomp\inv\alpha$.
\end{proof}

As a basic example observe that, for a groupoid $\gpd1$, the canonical kits from Example~\ref{ex:canonical_kits} are orthogonal:
\[ 
\Triv[\gpd1]^\orth = \Endo[\gpd1^\op]
 \qquad 
\Endo[\gpd1]^\orth = \Triv[\gpd1^\op]. 
\]

We define \emph{Boolean kits} to be those kits which are closed under \emph{double orthogonality}.  
\begin{defi}
\label{def:Kits}
	A \defn{Boolean kit} $\kit1$ on a groupoid $\gpd1$ is a kit $(\gpd1, \kit1)$ such
	that $\kit1=\kit1^\dorth$. 
\end{defi}

Writing $\Kit(\gpd1)$ for the poset of kits on $\gpd1$ under (component-wise) inclusion, the orthogonality relation on kits induces a Galois connection 
\vspace{-0.1cm}
	\begin{center}
	\begin{tikzpicture}[thick, line join=round,yscale=0.5]
	\node (A) at (0,0) {$\Kit(\gpd1)^\op $};
	\node (B) at (3.5, 0) {$\Kit(\gpd1^\op)$};
	\draw [->] (A) to [bend left =30]  node [above] {$(-)^\orth$} (B);
	\draw [->] (B) to [bend left =30]  node [below] {$(-)^\orth$} (A);
	\node (C) at (1.75,0) {$\bot$};
	\end{tikzpicture}
\end{center}
\vspace{-0.2cm}
whose fixed points $\kit1 = \kit1^\dorth$ are Boolean kits. The requirement that structures be closed under double orthogonality is common in the model theory of linear logic \cite{GlueingHylandShalk} and is necessary when one seeks a model for the \emph{classical} system, in which negation is involutive. More concretely, this condition ensures that certain basic properties always hold for a kit, as we will see shortly. 
	\begin{lem}\label{lem:BasicOrthogonalityProperties}
		Let $\kit1, \kit1'$ be kits on a groupoid $\gpd1$. Then: 
		\begin{enumerate}
			\item \label{prop:BasicOrthogonalityPropertiesTwo}
			$\kit1 \subseteq \kit1^\dorth$, where $\subseteq$ is component-wise inclusion.
			\item
			If $\kit1 \subseteq\kit1'$ then $\kit1'^\orth\subseteq\kit1^\orth$.
		\end{enumerate}
	\end{lem}
It follows that $\kit1^\orth = \kit1^{\dorth\perp}$; so, in particular, $\kit1^\orth$ is a Boolean kit. 
\begin{lem}\label{lem:kitDownClosed}
Let $(\gpd1, \kit1)$ be a Boolean kit, and let $a \in \gpd1$. Then, for each $a \in \gpd1$, the set $\kit1(a)$ is closed under subgroups (in particular, $\{\id[a]\} \in \kit1(a)$) and under unions of directed subsets.  
\end{lem}

\begin{proof}
	Let $H' \leq H \in \kit G$ and $G \in \kit G^\orth$, then $H' \cap G \subseteq H \cap G = \{ \id \}$ wich implies that $H' \in \kit G^\dorth = \kit G$.
	Let $\{ G_i \}_{i\in I}$ be a directed family in $\kit G$, then $ \bigcup_{i\in I} G_i$ is a subgroup of $G$ and for $H \in \kit G^\orth$, $\left( \bigcup_{i\in I} G_i  \right) \cap H = \bigcup_{i\in I} (G_i \cap H) = \{ \id \}$ so that $\bigcup_{i\in I} G_i$ is in $\kit G^\dorth = \kit G$ as desired.
\end{proof}

We consider as an example the possible Boolean kits on the group $C_6$, the cyclic group of order 6, seen as a groupoid. For an element $g$ of a group $G$, we denote by $\cyclic g$ the \defn{cyclic subgroup} generated by $g$; that is, $\cyclic g = \{ g^n \mid n \in \Z \}.$ Now, writing $\gamma$ for any generator of $C_6$, it can be verified that the only possible Boolean kits are the maximal and trivial ones, as well as the subsets $\{ \cyclic{\gamma^2 }, \{\id \}  \}$ and $\{ \cyclic{\gamma^3 }, \{\id \}  \}.$
These form a complete Boolean algebra, and in fact this is the case for the set of Boolean kits on any groupoid. 

It is sometimes useful to unfold the condition $\kit{A} = \kit{A}^\dorth$ by looking at the family of subsets $\bigcup \kit1(a) \subseteq \Endo(a)$, for $a \in \gpd1$. 
This brings us close to Taylor's \emph{creeds} \cite{TaylorGroupoidsLinearLogic}, which are subsets of endomorphisms
meeting certain explicit closure conditions. In practice, to show that a kit $\kit{A}$ is a Boolean kit, we show that for all $a\in \gpd{A}$, for all $\alpha \in \bigcup \kit{A}(a)^\dorth$, $\alpha$ is in $\bigcup \kit{A}(a)$ which corresponds to the \emph{saturation} property defined below:
\begin{defi}\label{def:saturatedKit}
	A kit $\kit{A}$ on a groupoid $\gpd{A}$ is \emph{saturated} if it satisfies the following condition: for $a \in \gpd{A}$ and $\alpha\in \Endo(a)$, if the formula $\Phi_{ \bigcup{\kit{A}}}(\alpha)$ given by
	\[ 
	\forall n\in\N, \alpha^n=\id \vee (\exists m\in\N , \alpha^{nm}\neq \id \wedge \alpha^{nm}\in \bigcup \kit{A}(a))
	\]
	holds, then $\alpha\in \bigcup \kit{A}(a)$.
\end{defi}

Saturation is however not enough to prove $\kit{A}= \kit{A}^\dorth$, it only shows that $\bigcup \kit{A}(a) = \bigcup \kit{A}^\dorth(a)$ for all $a \in \gpd{A}$. We need further that $\kit{A}(a) =\{ H \leq \gpd{A}(a,a)\mid H  \subseteq  \bigcup\kit{A}(a) \}$ for all $a\in \gpd{A}$.
\begin{lem}\label{lem:characterisationDoubleOrth}
	For a groupoid $\gpd{A}$ and a kit $\kit{A}$ on $\gpd{A}$, $\kit{A} = \kit{A}^\dorth$ if and only if 
	\begin{enumerate}
		\item for all $a \in \gpd{A}$, $\kit{A}(a) = \{ H \leq \gpd{A}(a,a)\mid H  \subseteq  \bigcup\kit{A}(a) \}$ 
		\item $\kit{A}$ is saturated. 
	\end{enumerate} 
\end{lem}

\begin{proof}
	Assume that $\kit{A} = \kit{A}^\dorth$ and let $H \leq  \gpd{A}(a,a)$ be such that $H  \subseteq  \bigcup\kit{A}(a)$ i.e. for all $\alpha \in H$, $\cyclic \alpha$ is in $\kit{A}(a)$ since $\kit{A}(a)$ is downclosed by Lemma \ref{lem:kitDownClosed}. Let $K$ be in $\kit{A}^\orth(a)$, for any $\alpha \in H \cap K$, we have $\cyclic \alpha \cap K = \{ \id \}$ which implies that $H$ is in $\kit{A}^\dorth(a) = \kit{A}(a)$ so that $\kit{A}(a) = \{ H \leq   \gpd{A}(a,a) \mid H  \subseteq  \bigcup\kit{A}(a) \}$ as desired. The saturation property is an immediate unfolding of the double orthogonality.
	For the other direction, if $\kit{A}$ is saturated, we have $\bigcup \kit{A}(a) = \bigcup \kit{A}^\dorth(a)$ for all $a \in \gpd{A}$ which implies that $\kit{A} = \kit{A}^\dorth$ by $(1).$ 
\end{proof}

For a kit $(\gpd1,\kit1)$ we will use the notation $(\gpd1,\kit1)^\orth=(\gpd1^\op,\kit1^\orth)$, and continue to use underline symbols $\kitstr{\gpd1}, \kitstr{\gpd2},$ for pairs $(\gpd1,\kit1), (\gpd2, \kit2)$.

\subsection{Stabilized Profunctors}\label{sec:stprof}
For groupoids equipped with kits, we will consider certain constrained profunctors which we call \emph{stabilized}. 

\begin{defi} 
\label{def:sprofunctor}
For kits $(\gpd1,\kit1)$ and $(\gpd2,\kit2)$, a \defn{stabilized profunctor} $P: (\gpd1,\kit1)\profto (\gpd2,\kit2)$ is a profunctor
	$P:\gpd1\profto\gpd2$ such that, for all $a\in\gpd1$, $b\in\gpd2$, 
	$p\in P(b,a)$, $\alpha\in\Endo(a)$, $\beta\in\Endo(b)$, if 
	$\alpha\act p\act\beta=p$ then 
	\[\textstyle
	\alpha \in \bigunion \kit1(a) \implies \beta \in \bigunion \kit2(b)
	\;\text{and}\;
	\beta \in \bigunion \kit2^\orth(b) \implies \alpha \in \bigunion \kit1^\orth(a).
	\]
\end{defi} 

It is instructive to consider some examples. 
\begin{exa}
\begin{enumerate}
\item Any profunctor $P : \gpd1 \profto \gpd2$ is a stabilized profunctor $P : (\gpd1, \Triv[\gpd1]) \profto (\gpd2, \Endo[\gpd2])$.
\item A stabilized profunctor $(\gpd1,\Endo[\gpd1])\profto(\gpd2,\Endo[\gpd2])$ is a profunctor $P : \gpd1 \profto \gpd2$ that acts freely on $\gpd1$. In other words, for all $a\in\gpd1$, $b\in\gpd2$, $p\in P(b,a)$, $\alpha\in\Endo(a)$, if $\alpha\act p=p$ then $\alpha=\id[a]$.
\item For any kit $\kitstr{\gpd1} = (\gpd1,\kit1)$, the  identity $\yon : \gpd1^\op \times \gpd1 \to \Set : (a, a') \mapsto \gpd1(a, a')$ is a stabilized profunctor $\kitstr{\gpd1} \profto \kitstr{\gpd1}$. This is because, 
for $a, a' \in \gpd1$, $\alpha \in \Endo(a)$, $\alpha' \in \Endo(a')$, and $\gamma \in \gpd1(a, a')$, $\alpha' \act \gamma \act \alpha = \alpha' \comp \gamma \comp \alpha = \gamma$ implies that $\alpha$ and $\alpha'$ are conjugate; and both $\kit1$ and $\kit1^\orth$ are closed under conjugation.   
\end{enumerate}
\end{exa}

Stabilized profunctors are closed under composition: 
\begin{lem}
	For kits $(\gpd1, \kit1)$, $(\gpd2, \kit2)$ and $(\gpd3, \kit3)$, if $P : (\gpd1, \kit1) \profto (\gpd2, \kit2)$ and $Q : (\gpd2, \kit2) \profto (\gpd3, \kit3)$ are stabilized profunctors, then their composite $Q \circ P$ is an stabilized profunctor $(\gpd1, \kit1) \profto (\gpd3, \kit3)$.
\end{lem}
\begin{proof}
Let $a \in \gpd1$, $c \in \gpd3$, $\alpha\in \gpd1(a,a)$, $\gamma \in \gpd3(c,c)$ and $t \in (Q \circ P)(c,a)$ be such that $\alpha \cdot t \cdot \gamma= t$. The element $t$ is of the form $p \coendtensor[b] q \in \int^b P(b,a) \times Q(c,b)$ for some $b \in \gpd2$, $p \in P(b, a)$ and $q \in Q(c,b)$ so $\alpha \cdot t \cdot \gamma= t$ is equivalent to $(\alpha \cdot p) \coendtensor[b] (q \cdot \gamma) = p \coendtensor[b] q$ which implies that there exists $\beta \in \gpd2(b,b)$ such that $\alpha \cdot p \cdot \beta = p$ and $\beta^{-1} \cdot q \cdot \gamma = q$. If $\alpha \in  \bigcup \kit1(a)$ then since $P$ is a stabilized profunctor and $\alpha\cdot p \cdot \beta = p$, we have $\beta \in \bigcup \kit2(b)$ which implies that $\beta^{-1} \in \bigcup \kit2(b)$ by closure under powers ($\beta$ is an element of a group $G \in \kit2(b)$). Likewise, since $Q$ is a stabilized profunctor and $\beta^{-1} \cdot q \cdot \gamma = q$, we have $\gamma \in \bigcup \kit3(c)$. The other implication $\gamma\in \bigcup \kit3^\perp(c) \Rightarrow \alpha\in \bigcup \kit1^\perp(a)$ is shown similarly.
\end{proof}

Boolean kits, stabilized profunctors, and natural transformations assemble into a bicategory which we denote $\SProf$. 

\begin{defi}
	The \emph{dual} of a profunctor $P:\gpd1\profto\gpd2$ is the profunctor 
	$P^\dual: \gpd2^\op\profto\gpd1^\op$ given by:
	\[
	(a,b) \mapsto P(b,a).
	\] 
	This induces a self-duality $(-)^\orth : \Prof^\op \to \Prof$.
\end{defi}
The following observation is immediate from the definition: if $P$ is a stabilized profunctor $\kitstr{\gpd1}\profto \kitstr{\gpd2}$ then $P^\dual$ is a stabilized profunctor $\kitstr{\gpd2}^\perp\profto \kitstr{\gpd1}^\perp$, and indeed the dualisation operation on kits $\kitstr{\gpd1} \mapsto  \kitstr{\gpd1}^\perp$ induces a self-duality on $\SProf$. 

The additional structure may be forgotten through a collapse:
\begin{prop}
	There is a forgetful pseudo-functor $\SProf \to \Prof$ which is strict, faithful and locally full and faithful; that is, for kits $(\gpd1, \kit1)$ and $(\gpd2, \kit2)$, the induced functor $\SProf((\gpd1, \kit1), (\gpd2, \kit2)) \to \Prof(\gpd1, \gpd2)$ is injective on objects and fully faithful.
\end{prop}

\section{Logical structure of stabilized profunctors}
\label{sec:logical-structure-sprof}

We show that $\SProf$ is a model of classical linear logic, and that
this structure extends and refines that of $\Prof$. In this
section we focus on multiplicative and additive structure. The
exponential modality will be discussed in the next section.

\subsection{Biproduct structure.}
\label{sec:biproducts}
 In the standard $1$-dimensional models of stability (coherence spaces, probabilistic coherence spaces, etc.), the linear categories do not have biproducts which implies in particular that they are not models of differential linear logic. In our setting, the biproduct structure in $\Prof$ can be transported to $\SProf$.
For a family of groupoids $(\gpd1_i)_{i\in I}$, we denote by $\with_{i \in I} \gpd1_i$ their coproduct in $\Cat$, whose objects are given by the disjoint union $\bigcup_{i\in I} \{i\} \times \gpd1_i$. In $\Prof$, this is a \emph{biproduct}, meaning that $\with_{i \in I} \gpd1_i$ is a product \emph{and} a coproduct for the family $(\gpd1_i)_{i\in I}$.

For a finite family of kits $\{ \kitstr{\gpd1}_i = (\gpd1_i, \kit1_i) \}_{i \in I}$, the catgorical coproduct $\with_{i \in I} \gpd1_i$ has a canonical kit structure $(\with_{i \in I} \gpd1_i, \coprod_{i \in I} \kit1_i)$ defined as $(\coprod_{i \in I} \kit1_i)(j, a) = \kit1_j(a)$ for all $j \in I, a \in \gpd1$. Since $(\coprod_{i \in I} \kit1_i^\orth)^\orth = \coprod_{i \in I} \kit1_i$, it is in fact a Boolean kit which implies that the product kit $\with_{i \in I} \kitstr{\gpd1}_i:= (\with_{i \in I} \gpd1_i, \coprod_{i \in I} \kit1_i)$ is equal to the coproduct kit $\oplus_{i \in I} \kitstr{\gpd1}_i:=(\with_{i \in I} \kitstr{\gpd1}_i^\orth)^\orth $.

The projections $\pi_j : \with_{i \in I} {\gpd1}_i \profto {\gpd1}_j$ in $\Prof$ determine stabilized profunctors $\with_{i \in I}\kitstr{{\gpd1}}_i \profto \kitstr{{\gpd1}}_j$, and one can show the existence of an adjoint equivalence
\[
\prod_{i \in I} \SProf(\kitstr{\gpd2}, \kitstr{\gpd1}_{i\in I})  \simeq \SProf(\kitstr{\gpd2}, \with_{i \in I} \kitstr{\gpd1}_i)
\]
making $\with_{i \in I} \kitstr{\gpd1}_i$ a bicategorical product. From the self-dual structure of $\SProf$, we obtain that the inclusions $\colimin_j : \kitstr{\gpd1}_j \profto  \with_{i \in I} \kitstr{\gpd1}_i$, obtained as the dual of the projection $\pi_j : \with_{i \in I} \kitstr{\gpd1}^\perp_i \profto \kitstr{\gpd1}^\perp_j$, are also stabilized profunctors and yield the desired adjoint equivalence in $\SProf$ for coproducts.
 We thus have finite biproducts, with the zero object defined by the kit $(\Zero, \varnothing)$ on the empty groupoid $\Zero$.

\subsection{$*$-Autonomous structure.}\label{sec:*autonomous}

The bicategory $\Prof$ can be equipped with a symmetric monoidal structure defined by the cartesian product of categories $\gpd1 \tensor \gpd2 = \gpd1 \times \gpd2$, with monoidal unit given by the terminal category $\One$. Additionally, the duality pseudo-functor $\gpd1 \mapsto \gpd1^\op$ makes $\Prof$ a compact closed bicategory \cite{StayCompactBicat}. From the point of view of linear logic, $\Prof$ is a degenerate model: $(\gpd1^\perp \tensor \gpd2^\perp)^\perp = \gpd1 \tensor \gpd2$, hence the multiplicative connectives $\parr$ and $\tensor$ are identified. The addition of kits to groupoids eliminates this degeneracy: $\SProf$ is not compact closed but merely $*$-autonomous, as we will see. 

For kits $(\gpd1, \kit1)$ and $(\gpd2, \kit2)$, we construct a tensor product as 
\[
	(\gpd1,\kit1) \tensor (\gpd2,\kit2) :=( \gpd1\times\gpd2 , (\kit1\times\kit2)^\dorth )
\]
where for $a \in \gpd1$ and $b\in \gpd2$, $(\kit1\times \kit2)(a,b):= \kit1(a)\times\kit2(b)$. 
The closure under double orthogonality is necessary because $\kit1\times \kit2$ is not a Boolean kit in general. The monoidal unit is the kit $(\One,\Triv[\One])$. To show that this induces a symmetric monoidal structure, it suffices to show that the symmetric monoidal structure of $\Prof$ lifts to $\SProf$, i.e. that associator, braiding, and left and right unitors are pseudo-natural transformations with components in $\SProf$, and that the tensor product is a pseudo-functor $\SProf \times \SProf \to \SProf$. All coherence axioms will then follow from the (faithful) forgetful pseudo-functor $\SProf \to \Prof$. 

This way, we avoid much of the work in defining a bicategorical monoidal structure. The proof is not automatic, however: our use of the closure operator $(-)^\dorth$ in the definition of $\tensor$ makes certain properties, e.g. associativity, difficult to verify. The following lemma is a key tool in this process.

\begin{lem}\label{lem:technicalAssociativity}
	For Boolean kits $\kitstr{\gpd1}= (\gpd1, \kit1)$ and $\kitstr{\gpd2} = (\gpd2, \kit2)$, define the kit $\kitstr{\gpd1} \multimap\kitstr{\gpd2}$ as $(\gpd1^{\op} \times \gpd2, \kit1 \multimap \kit2)$ where $(\kit1 \multimap \kit2)(a,b)$ is the set of all subgroups $H \leq \Endo(a,b)$ such that:
	\[
	\forall (\alpha,\beta) \in H, \alpha \in \bigcup \kit1(a) \implies \beta \in \bigcup \kit2(b)
	\text{  and  }
	\beta \in \bigcup \kit2^\perp(b) \implies \alpha \in \bigcup \kit1^\perp(a).
	\]
	The kit $\kitstr{\gpd1} \multimap\kitstr{\gpd2}$ is Boolean.
\end{lem}

\begin{proof}
	To prove that $\kitstr{\gpd1} \multimap\kitstr{\gpd2}$ forms a Boolean kit, we first show that for objects $a \in \gpd1, b\in \gpd2$ and subgroups $K \in \kit1(a)$ and $G \in \kit2^\orth(b)$, we have $K \times G \in (\kit1 \multimap \kit2)^\orth (a,b)$. Let $H$ be in $(\kit1 \multimap \kit2)(a,b)$ and $(\alpha,\beta) \in (K \times G) \cap H$. Since $\alpha \in K \in \kit1(a)$, we have $\beta\in \kit2(b)$ by definition of $\kit1\multimap \kit2$ which implies that $\beta=\id$ since $\beta\in G \in \kit2^\orth(b)$. Likewise, since $\beta \in G \in \kit2^\orth(b)$, we have $\alpha\in \kit1^\orth(a)$ which implies that $\alpha=\id$ as $\alpha\in K \in \kit1(a)$.
	
	Now, let $H$ be in $(\kit1 \multimap \kit2)^\dorth(a,b)$, we show that $H$ is in $(\kit1 \multimap \kit2)(a,b)$. Let $(\alpha,\beta)$ be in $H$ and assume that $\alpha\in \bigcup \kit1(a)$ and that there exists $n\in \N$ such that $\beta^n \in \bigcup \kit2^\perp(b)$. Hence, we have $\seq{\alpha^n} \times \seq{\beta^n} \in (\kit1(a) \times \kit2^\orth(b))$ which implies that $\seq{\alpha^n} \times \seq{\beta^n} \in (\kit1 \multimap \kit2)^\orth (a,b)$ so that $\alpha^n = \id $ and $\beta^n =\id$ as desired. We can show similarly that if $\beta \in \bigcup \kit2^\perp(b)$ then $\alpha \in \bigcup \kit1^\perp(a)$.
	
	For a subgroup $H \leq \Endo(a,b)$ such that $H \subseteq \bigcup (\kit1 \multimap \kit2)(a,b)$, it is immediate that $H$ is in $(\kit1 \multimap \kit2)(a,b)$. Hence, by Lemma \ref{lem:characterisationDoubleOrth}, we conclude that $\kit1 \multimap \kit2$ is a kit.
\end{proof}

As an immediate corollary, we obtain:

\begin{cor}\label{cor:IsoKitsLinear}
	For Boolean kits $\kitstr{\gpd1}$ and $\kitstr{\gpd2}$, the categories $\SProf(\kitstr{\gpd1}, \kitstr{\gpd2})$ and $\SProf(\kitstr{\One}, \kitstr{\gpd1} \multimap \kitstr{\gpd2})$ are isomorphic.
\end{cor}

\begin{lem}
	If $P_1: \kitstr{\gpd1}_1 \profto \kitstr{\gpd2}_1$ and $P_2: \kitstr{\gpd1}_2 \profto \kitstr{\gpd2}_2$ are stabilized profunctors then so is $P_1 \otimes P_2 : (\kitstr{\gpd1}_1 \otimes \kitstr{\gpd1}_2) \profto (\kitstr{\gpd2}_1 \otimes \kitstr{\gpd2}_2)$ defined as
	\[
	(P_1 \otimes P_2)((b_1,b_2), (a_1,a_2)) := P_1(b_1,a_1) \times P_2(b_2,a_2).
	\]
\end{lem}
\begin{proof}
	Assume that there is $(a_1,a_2) \in\gpd1_1\times \gpd1_2$, $(b_1, b_2) \in \gpd2_1\times \gpd2_2$, $(\alpha_1,\alpha_2) \in \Endo(a_1,a_2)$, $(\beta_1, \beta_2) \in \Endo(b_1,b_2)$ and $(p_1, p_2) \in P_1(b_1,a_1) \times P_2(b_2,a_2)$ such that 
	\[
	(\alpha_1, \alpha_2) \cdot (p_1, p_2) \cdot (\beta_1, \beta_2) = (p_1, p_2)
	\]
	i.e. $\alpha_1 \cdot p_1 \cdot \beta_1 = p_1$ and $\alpha_2 \cdot p_2 \cdot \beta_2 = p_2$.
	\begin{itemize}
		\item We first show that if $(\alpha_1, \alpha_2) \in \bigcup({\kit1}_1 \otimes {\kit1}_2)(a_1,a_2)$ then $(\beta_1, \beta_2) \in \bigcup({\kit2}_1 \otimes {\kit2}_2)(b_1,b_2)$. 
		
		Recall from Definition \ref{def:saturatedKit} that $(\alpha_1, \alpha_2) \in \bigcup({\kit1}_1 \otimes {\kit1}_2)(a_1,a_2) = \bigcup({\kit1}_1 \times {\kit1}_2)^\dorth(a_1,a_2)$ is equivalent to the formula $\Phi_{\bigcup {\kit1}_1 \otimes {\kit1}_2}(\alpha_1,\alpha_2)$ given by:
		\begin{gather*}
			\forall n \in \N, (\alpha_1, \alpha_2)^n = (\id, \id)  \:\: \vee \\
			\left( \exists m \in \N,  (\id,\id)  \neq (\alpha_1, \alpha_2)^{nm} \in \bigcup ({\kit1}_1 \times {\kit1}_2)(a_1,a_2)
			\right).
		\end{gather*}
		Let $n\in \N$ be such that $\beta_1^n \neq \id$ or $\beta_2^n \neq \id$. If $(\alpha_1, \alpha_2)^n = (\id, \id)$, then $(\alpha_1^n, \alpha_2^n) \in \bigcup({\kit1}_1 \times {\kit1}_2)(a_1,a_2)$ which implies that $(\beta_1^n, \beta_2^n) \in \bigcup({\kit2}_1 \times {\kit2}_2)(b_1,b_2)$ since $P_1$ and $P_2$ are stabilized profunctors. Taking $m=1$, the formula $\Phi_{\bigcup {\kit2}_1 \otimes {\kit2}_2}(\beta_1,\beta_2)$ holds.
		
		Assume now that there exists $m \in \N$ such that $(\alpha_1^{nm}, \alpha_2^{nm}) \neq (\id,\id)$ and $(\alpha_1^{nm}, \alpha_2^{nm}) \in \bigcup({\kit1}_1 \times {\kit1}_2)(a_1,a_2)$. It implies that $ (\beta_1^{nm}, \beta_2^{nm})$ is in $\bigcup({\kit2}_1 \times {\kit2}_2)(b_1,b_2)$ since $P_1$ and $P_2$ are stabilized profunctors. Assume that $\beta_1^{nm} = \id$ and $\beta_2^{nm} = \id$, then $(\beta_1^{nm}, \beta_2^{nm}) \in \bigcup({\kit2}_1^\orth \times {\kit2}_2^\orth)(b_1,b_2)$ which implies that $(\alpha_1^{nm}, \alpha_2^{nm}) \in \bigcup({\kit1}_1^\orth \times {\kit1}_2^\orth)(a_1,a_2)$. Hence, $(\alpha_1^{nm}, \alpha_2^{nm}) = (\id, \id)$ which gives us the desired contradiction.
		
		\item For the other direction, assume that $(\beta_1, \beta_2) \in \bigcup({\kit2}_1 \otimes {\kit2}_2)^\orth(b_1,b_2)= \bigcup({\kit2}_1 \times {\kit2}_2)^\orth(b_1,b_2)$, we want to show that  $(\alpha_1, \alpha_2) \in \bigcup({\kit1}_1 \times {\kit1}_2)^\orth(a_1,a_2)$. Let $H_1 \in {\kit1}_1(a_1)$, $H_2 \in {\kit1}_2(a_2)$ and $n \in \N$ such that $(\alpha_1^n, \alpha_2^n) \in H_1 \times H_2$. We need to show that $(\alpha_1^n, \alpha_2^n) = (\id, \id)$. Since $P_1$ is a stabilized profunctor, $\beta_1^n \in \bigcup {\kit2}_1(b_1)$, and therefore, by Corollary~\ref{cor:TensorLinearArrow}, we must have $\beta_2^n \in \bigcup {\kit2}_2(b_2)^\orth$. Since $P_2$ is a stabilized profunctor, we must have $\alpha_2^n \in {\kit1}_2(a_2)^\orth$, and therefore $\alpha_2^n = \id$. Likewise, we can show that $\alpha_1^n = \id$ so that $(\alpha_1^n, \alpha_2^n) =(\id, \id)$ as desired. \qedhere
	\end{itemize}
\end{proof}

From this, one derives that $\tensor$ defines a pseudo-functor of two arguments. Continuing with the monoidal structure, we turn to the unitors. For a kit $\kitstr{\gpd1} = (\gpd1, \kit1)$, the canonical isomorphisms $\One \times \gpd1 \cong \gpd1$ and $\gpd1 \times \One \cong \gpd1$ in $\Cat$ induce adjoint equivalences $\kitstr{\One}  \tensor \kitstr{\gpd1}  \profto \kitstr{\gpd1} $ and $\kitstr{\gpd1}  \tensor \kitstr{\One} \profto \kitstr{\gpd1} $ in $\SProf$ that are pseudo-natural in $\kitstr{\gpd1}$. Next, for kits $\kitstr{\gpd1}$ and $\kitstr{\gpd2}$, the component of the braiding $\sigma_{\gpd1, \gpd2} : \gpd1 \times \gpd2 \profto \gpd2 \times \gpd1$ in $\Prof$ is a stabilized profunctor $\kitstr{\gpd1} \tensor \kitstr{\gpd2} \profto \kitstr{\gpd2} \tensor \kitstr{\gpd1}$ and it induces a pseudo-natural family of adjoint equivalences in $\SProf$. 

\begin{lem}\label{lem:kitsTensorLinear}
	For Boolean kits $\kitstr{\gpd1}=(\gpd1, \kit1)$ and $\kitstr{\gpd2}=(\gpd2, \kit2)$, the Boolean kits $\kitstr{\gpd1} \multimap \kitstr{\gpd2}$ and $(\kitstr{\gpd1}\otimes \kitstr{\gpd2}^\orth)^\orth$ are isomorphic.
\end{lem}

\begin{proof}
	Let $a \in \gpd1$, $b \in \gpd2$ and $(\alpha,\beta)\in \bigcup (\kit1 \otimes \kit2^\orth)^\orth (a,b)= \bigcup (\kit1(a) \times \kit2(b)^\orth)^\orth$. Assume that $\alpha \in \bigcup \kit1(a)$, if there exists $n \in \N$ such that $\beta^n \in \bigcup \kit2^\orth(b)$, then $(\alpha,\beta)^n \in \bigcup (\kit1(a) \times \kit2(b)^\orth)$ which implies that $(\alpha,\beta)^n = (\id, \id)$ so that $\beta$ is in $\bigcup \kit2(b)$ as desired. Dually, assume that $\beta\in \bigcup \kit2^\orth(b)$, if there exists $n \in \N$ such that $\alpha^n \in \bigcup \kit1^\orth(a)$ then $(\alpha,\beta)^n=(\id,\id)$ which implies that $\alpha\in \bigcup \kit1^\dorth(a) =\bigcup \kit1(a)$ as desired.
	
	For the other inclusion, let $H\in (\kit1 \multimap \kit2)(a,b)$, $K \in \kit1(a)$, $G \in \kit2^\orth(b)$ and $(\alpha,\beta) \in H \cap (K \times G)$. Since $\alpha$ is in $H \in \kit1(a)$ and $(\alpha,\beta) \in \bigcup (\kit1 \multimap \kit2)(a,b)$, we have $\beta\in \bigcup \kit2(b)$ which implies that $\beta=\id$ as $\beta$ is in $G \in \kit2^\orth(b)$. We obtain similarly that $\alpha=\id$ which implies that $G$ is in $(\kit1(a) \times \kit2(b)^\orth)^\orth$.
\end{proof} 

\begin{cor}\label{cor:TensorLinearArrow}
	For Boolean kits $(\gpd1, \kit1)$ and $(\gpd2, \kit2)$, objects $a \in \gpd1$ and $b \in \gpd2$, if $(\alpha, \beta) \in \bigcup (\kit1(a) \times \kit2(b))^\orth$, then $\alpha \in \bigcup \kit1(a)$ implies $\beta \in \bigcup\kit2(b)^\orth$ and  $\beta\in \bigcup \kit2(b)$ implies $\alpha \in \bigcup\kit1(a)^\orth$.
\end{cor}

\begin{lem}\label{lem:AssocTensorStep}
	For Boolean kits $\kitstr{\gpd1}$, $\kitstr{\gpd2}$ and $\kitstr{\gpd3}$, the Boolean kits $(\kitstr{\gpd1} \otimes \kitstr{\gpd2}) \multimap \kitstr{\gpd3}$ and $\kitstr{\gpd1} \multimap (\kitstr{\gpd2} \multimap \kitstr{\gpd3})$ are isomorphic. 
\end{lem}

\begin{proof}
	We show that for objects $a \in \gpd1$, $b\in \gpd2$ and $c \in \gpd3$, we have:
	\[\begin{aligned}
		\left( (\kit1(a) \times \kit2(b))^\dorth \times \kit3^\orth(c)\right)^\orth &\cong \left((\kit1(a) \times \kit2(b)) \times \kit3^\orth(c) \right)^\orth \quad \text{and}\\
		\left(\kit1(a) \times (\kit2(b) \times \kit3^\orth(c))^\dorth\right)^\orth &\cong \left(\kit1(a) \times (\kit2(b) \times \kit3^\orth(c))\right)^\orth .
	\end{aligned}
	\]
	Since $\kit1(a) \times \kit2(b) \subseteq (\kit1(a) \times \kit2(b))^\dorth$, we have $(\kit1(a) \times \kit2(b)) \times \kit3^\orth(c) \subseteq (\kit1(a) \times \kit2(b))^\dorth \times \kit3^\orth(c)$ which implies that 
	\[\left( (\kit1(a) \times \kit2(b))^\dorth \times \kit3^\orth(c)\right)^\orth \subseteq \left((\kit1(a) \times \kit2(b)) \times \kit3^\orth(c) \right)^\orth.\]
	
	For the reverse inclusion, let $G$ be in $\left((\kit1(a) \times \kit2(b)) \times \kit3^\orth(c) \right)^\orth$ and $H$ in $(\kit1(a) \times \kit2(b))^\dorth \times \kit3^\orth(c)$ i.e. $H = H_1 \times H_2$ with $H_1 \in (\kit1(a) \times \kit2(b))^\dorth $ and $H_2 \in \kit3^\orth(c)$. To show that $G \perp H$, assume that there exists $((\alpha_1, \alpha_2), \alpha_3) \in G \cap H$. Then, $\alpha_3 \in H_2 \in \kit3^\orth(c)$ so by Corollary \ref{cor:TensorLinearArrow}, we have $(\alpha_1,\alpha_2) \in (\kit1(a) \times \kit2(b))^\orth$. Since $(\alpha_1,\alpha_2) \in H_1 \in (\kit1(a) \times \kit2(b))^\dorth$, we must have $\alpha_1 = \id$ and $\alpha_2 = \id$. Hence, $(\alpha_1, \alpha_2) \in  \kit1(a) \times \kit2(b)$ so by Corollary \ref{cor:TensorLinearArrow} again, we have $\alpha_3 \in \kit3^\dorth(c)$ which implies that $\alpha_3 = \id$ as desired.
	
	The isomorphism $\left(\kit1(a) \times (\kit2(b) \times \kit3^\orth(c))^\dorth\right)^\orth \cong \left(\kit1(a) \times (\kit2(b) \times \kit3^\orth(c))\right)^\orth$ is derived using a similar argument.
\end{proof}

\begin{lem}
	For Boolean kits $(\gpd1, \kit1), (\gpd2, \kit2)$ and $(\gpd3, \kit3)$, the component of the associator $\alpha_{\gpd1, \gpd2, \gpd3}$ in $\Prof$ is a stabilized profunctor
	\[
	((\gpd1, \kit1) \tensor (\gpd2, \kit2))\tensor (\gpd3, \kit3) \profto (\gpd1, \kit1) \tensor ((\gpd2, \kit2)\tensor (\gpd3, \kit3))
	\] 
	inducing a family of adjoint equivalences pseudo-natural in $(\gpd1, \kit1), (\gpd2, \kit2),$ and $(\gpd3, \kit3)$. 
\end{lem}

\begin{proof}
	It suffices to show that the Boolean kits $(\kitstr{\gpd1} \otimes \kitstr{\gpd2}) \otimes \kitstr{\gpd3}$ and $\kitstr{\gpd1} \otimes (\kitstr{\gpd2} \otimes \kitstr{\gpd3})$ are isomorphic. By Lemmas \ref{lem:kitsTensorLinear} and \ref{lem:AssocTensorStep}, we have:
	\begin{align*}
		\qquad(\kitstr{\gpd1} \otimes \kitstr{\gpd2}) \otimes \kitstr{\gpd3} & \cong \left((\kitstr{\gpd1} \otimes \kitstr{\gpd2}) \multimap \kitstr{\gpd3}^\orth \right)^\orth \cong \left( \kitstr{\gpd1} \multimap (\kitstr{\gpd2} \multimap \kitstr{\gpd3}^\orth) \right)^\orth\\
		& \cong \left( \kitstr{\gpd1} \multimap (\kitstr{\gpd2} \otimes \kitstr{\gpd3})^\orth \right)^\orth \cong \kitstr{\gpd1} \otimes (\kitstr{\gpd2} \otimes \kitstr{\gpd3}).\qquad\qquad\qquad\quad \qedhere
	\end{align*}
\end{proof}

\begin{prop}
The bicategory $\SProf$ is symmetric monoidal.  
\end{prop}

Finally, $\SProf$ is \defn{$*$-autonomous} \cite{barr2006autonomous}: 
\begin{prop}
	For Boolean kits $\kitstr{\gpd1}, \kitstr{\gpd2},$ and $\kitstr{\gpd3}$ there is a pseudo-natural adjoint equivalence 
	\[
	\SProf(\kitstr{\gpd1} \tensor \kitstr{\gpd2},\kitstr{\gpd3}^\perp) \simeq \SProf(\kitstr{\gpd1}, (\kitstr{\gpd2} \tensor \kitstr{\gpd3})^\perp).
	\]
\end{prop}

\begin{proof}
By Corollary \ref{cor:IsoKitsLinear}, it suffices to show that $(\kitstr{\gpd1} \otimes \kitstr{\gpd2}) \multimap \kitstr{\gpd3}^\orth \cong\kitstr{\gpd2}\multimap (\kitstr{\gpd2}\otimes \kitstr{\gpd3})^\orth$ which we obtain from Lemmas \ref{lem:kitsTensorLinear} and \ref{lem:AssocTensorStep}.
\end{proof}
It follows that the symmetric monoidal bicategory $\SProf$ has closed structure given by $\kitstr{\gpd1} \multimap \kitstr{\gpd2} = (\kitstr{\gpd1} \tensor \kitstr{\gpd2}^\perp)^\perp $, and that the object $\bot = \One ^\perp \cong \One$ is dualising: the canonical stabilized profunctor $\kitstr{\gpd1} \profto (\kitstr{\gpd1} \multimap \bot) \multimap \bot $ is part of an equivalence representing the 
involutive negation of classical linear logic. 

Moreover, the operation $\kitstr{\gpd1} \parr
\kitstr{\gpd2}:=(\kitstr{\gpd1}^\perp \tensor
\kitstr{\gpd2}^\perp)^\perp$ gives rise to a second monoidal structure
on $\SProf$ distinct from the tensor $\otimes$; in the compact closed
$\Prof$ these degenerate to the same structure. Finally,  we have the
additional property that $\kit1 \parr \kit2 = (\kit1^\orth \times
\kit2^\orth)^\orth$ is a (component-wise) subset of $\kit1 \otimes
\kit2 = (\kit1\times \kit2)^\dorth$, making $\SProf$ a model for the
\emph{mix rule}~\cite{CockettMix}.

\section{Stable Species}
\label{sec:stable-species}
As described in Section~\ref{subsec:generalized-species}, the bicategory $\Esp$ of generalized species of structures
arises as a coKleisli construction over the bicategory $\Prof$, for
the pseudo-comonad $\Sym$, where $\Sym \gpd1$ is the symmetric strict
monoidal completion of $\gpd1$. 

\subsection{Exponential structure.}
We proceed to extend the $\Sym$ construction to $\SProf$. We show that
this determines a linear exponential pseudo-comonad, whose coKleisli
bicategory is that of \emph{stable species of structures}. With an eye to linear logic, we call this extension $\oc$, so that for every $(\gpd1, \kit1) \in \SProf$ one has  
\[
\oc (\gpd1, \kit1) := (\Sym \gpd1, \oc \kit1)
\]
and we will define the Boolean kit $\oc \kit1$ shortly. 

Consider an object $u = \seq{a_1, \dots, a_n} \in \Sym\gpd1$ and an
element $\alpha=(\sigma, (\alpha_i : a_i \to a_{\sigma(i)})_{i \in
  \ints{n}})$ of $\Endo[\Sym \gpd1](u)$. Unless the permutation
$\sigma$ is the identity, the morphisms $\alpha_i$ are not
endomorphisms. Thus, in order to define $\oc \kit1(u)$ from the sets
$\kit1(a_i)$, we generate an endomorphism of $a_i$ for every $i \in
\ints{n}$, using the following steps. 

\begin{defi}
	For $\sigma \in \symgroup{n}$ and $i \in \ints{n}$, let $o(\sigma, i)$ be  the smallest strictly positive integer such that $\sigma^{o(\sigma, i)}(i) = i$. Equivalently, $o(\sigma,i)$ is the length of the cycle containing $i$ in the disjoint cycle decomposition of the permutation $\sigma$. If there is no ambiguity on the permutation, we just write $o(i)$ for $o(\sigma, i)$.
\end{defi}	

\begin{lem}\label{lem:OrderConjugatePermutations}
	Let $n$ be in $\N$ and $\sigma, \tau, \phi$ be permutations in $\symgroup{n}$. If $\sigma = \varphi^{-1} \tau \varphi$, then for all $i \in \ints{n}$, $o(\sigma, i) = o(\tau, \varphi(i))$.
\end{lem}

\begin{proof}
	We have $\sigma^{o(\tau, \varphi(i))}(i) = \varphi^{-1} \tau^{o(\tau, \varphi(i))} \varphi(i)= \varphi^{-1}\varphi(i)=i$. Assume that there exists $0< j < o(\tau, \varphi(i))$ such that $\sigma^j(i) =i$, then $\tau^j\varphi(i) = \varphi \sigma^j(i) = \varphi(i)$ which contradicts the minimality of $o(\tau, \varphi(i))$.
\end{proof}

\begin{defi}
	For a sequence $u = \seq{a_1, \dots, a_n}$ in $\Sym \gpd1$ and a morphism $\alpha \in \Sym\gpd1(u,u)$, we define for each $i \in \ints{n}$, an endomorphism $\bangloop{\alpha }{i} : a_i \to a_i$ as the composite 
	\begin{center}
		\begin{tikzpicture}[line join=round]
			\node (A) at (0,0) {$a_i$} ;
			\node (B) at (2.5, 0) {$a_{\sigma(i)}$};
			\node (C) at (5.5, 0) {$a_{\sigma^2(i)}\dots$} ;
			\node (D) at (8.5, 0) {$a_i$} ;
			
			\draw [->] (A) to   node [above] {$\alpha _i$} (B);
			\draw [->] (B) to   node [above] {$\alpha _{\sigma(i)}$} (C);
			\draw [->] (C) to   node [above] {$\alpha _{\sigma^{o(i) -1}(i)}$} (D);
		\end{tikzpicture}
	\end{center}
\end{defi}

\begin{defi}\label{def:expBooleanKit}
	For a Boolean kit $(\gpd1, \kit1)$,  $\oc (\gpd1, \kit1)$ is defined as $(\Sym \gpd1,(\kit1^\Sym)^\dorth)$, where for an object $u = \seq{a_1, \dots, a_n} \in \Sym \gpd1$, $\kit1^\Sym(u)$ is given by:
	\[
	\{ H \subgroup \Endo(u) \mid \Forall{\alpha \in H}  \forall i \in \ints{n}, \bangloop{ \alpha}{i} \in \kit1(a_i)  \}
	\]
\end{defi}

 Since the kit $\kit1^\Sym$ is not Boolean in general, the closure under double orthogonality ensures that $\oc \kit1 = (\kit1^\Sym)^\dorth$ is indeed a Boolean kit
We show that the operation $(\gpd1, \kit1) \mapsto \oc (\gpd1, \kit1) $ extends to a pseudo-comonad on $\SProf$. The proof is a technical but straightforward verification involving the characterisation described in Lemma~\ref{lem:characterisationDoubleOrth}.

\begin{prop}
	\label{prop:bang_preserves_qprofs}
	For Boolean kits $(\gpd1, \kit1)$ and $(\gpd2, \kit2)$ and $P:(\gpd1, \kit1) \profto (\gpd2, \kit2)$ a stabilized profunctor, the profunctor $\Sym P : \Sym \gpd1 \profto \Sym \gpd2$ is a stabilized profunctor $\oc(\gpd1, \kit1) \profto \oc(\gpd2, \kit2)$. 
\end{prop}

\begin{proof}
	Recall that for objects $u = \seq{a_1, \dots, a_n} \in \Sym \gpd1$ and $v = \seq{b_1, \dots, b_n} \in \Sym \gpd2$, the profunctor $\Sym P: \Sym \gpd1 \profto \Sym \gpd2$ is given by:
	\[
	\Sym P(v,u) = \begin{cases} \coprod_{\varphi \in \symgroup{n}} \prod_{j \in [m]} P(b_j, a_{\varphi(j)}), & \text{if } n=m\\
		\varnothing, & \text{otherwise}
		\end{cases}
	\]
	with the following functorial action: for an element $p = (\varphi, \seq{p_j}_{j \in [m]}) \in \Sym P(v,u)$ and morphisms $\alpha = (\sigma, \seq{\alpha_i}_{i \in [n]}) : u \to u'$ and $\beta = (\tau, \seq{\beta_j}_{j \in [m]}) : v'\to v$, 
	\[
	\alpha \act p \act \beta = (\sigma \comp \varphi \comp \tau, \seq{\alpha_{\varphi(\tau(j))}\act p_{\tau(j)} \act \beta_j }_{j \in [m]}). 
	\]
	Assume that there exists morphisms $\alpha = (\sigma, (\alpha_i)_{i \in \ints{n}}) \in \Endo(u)$, $\beta= (\tau, (\beta_i)_{i \in \ints{n}}) \in \Endo(v)$, and $p = (\varphi, (p_i)_{i \in \ints{n}})$ be in $\Sym P(v,u)$ such that $\alpha \cdot p \cdot \beta = p$, i.e. 
	\[
	\sigma \circ \varphi \circ \tau = \varphi 
	\] 
	and for all $i \in \ints{n}$, 
	\[
	\alpha_{\varphi(\tau(i))} \cdot p_{\tau(i)} \cdot \beta_i = p_{i}.
	\]
	We first show that it implies that for all $i \in \ints{n}$, $\bangloop{\alpha}{\varphi(i)} \cdot p_i \cdot \bangloop{\beta}{i} = p_i$.
	By Lemma \ref{lem:OrderConjugatePermutations}, we have:
	\[\begin{aligned}
		\bangloop{\alpha}{\varphi(i)} \cdot p_i \cdot \bangloop{g}{i}  &=\alpha_{\sigma^{o(\sigma, \varphi) -1}(\varphi(i))} \dots \alpha_{\sigma \varphi(i)}  \alpha_{\varphi(i)} \cdot p_i \cdot \beta_{\tau^{o(\tau,i)-1}(i)}  \dots  \beta_{\tau(i)} \cdot \beta_i\\
		&= \alpha_{\sigma^{o(\tau, i) -1}(\varphi(i))} \dots  \alpha_{\sigma \varphi(i)}  \alpha_{\varphi(i)} \cdot p_i \cdot \beta_{\tau^{o(\tau,i)-1}(i)}  \dots  \beta_{\tau(i)} \cdot \beta_i
	\end{aligned}\]
	Note that since $\tau^{o(\tau, i)}(i)= i$, $\tau^{o(\tau, i)-1}(i)= \tau^{-1}(i)$ we have $\alpha_{\varphi(i)} \cdot p_i \cdot\beta_{\tau^{o(\tau,i)-1}(i)} = \alpha_{\varphi \tau(\tau^{-1}(i))} \cdot p_{\tau(\tau^{-1}(i))} \cdot \beta_{\tau^{-1}(i)} = p_{\tau^{-1}(i)}$. Repeating this process, we obtain the desired result.
	We now show that the following two implications hold:
	\[
	\alpha \in \bigcup \oc\kit1(u) \Rightarrow \beta\in \bigcup \oc\kit2(v) \quad \text{and} \quad \beta \in \bigcup \oc\kit2^\orth(v) \Rightarrow \alpha\in \bigcup \oc\kit1^\orth(u).
	\]
	
	\begin{itemize}
		\item Assume that $\alpha\in \bigcup \oc\kit1(u)$. By Lemma~\ref{lem:characterisationDoubleOrth}, it is equivalent to the following formula:
		\[
		\forall m \in \N, \alpha^m = \id \vee \left(\exists k \in \N, \alpha^{mk} \neq \id \wedge(\forall i \in \ints{n},  \bangloop{\alpha^{mk}}{a_i} \in \kit1(a_i)) \right).
		\] 
		Let $m \in \N$ be such that $\beta^m \neq \id$, if $\alpha^m = \id$, then $\sigma^m =\id$ and for all $i \in \ints{n}$, $\alpha_i^m =\id \in \kit1(a_i)$. It implies that $\tau^m = \id$ since $\sigma^m \circ \varphi \circ \tau^m = \varphi$ and for all $i \in \ints{n}$, $\bangloop{\beta^m}{i}=\beta^m_i \in \kit2(b_i)$ since $\alpha_i \cdot p_i \cdot \beta_i= p_i$ and $P$ is a stabilized profunctor. Hence $\beta \in \oc \kit2(v)$ by Lemma \ref{lem:BasicOrthogonalityProperties}$.1$.s
		If $\alpha^m \neq \id$, then there exists $k \in \N$ such that $\alpha^{mk} \neq \id$ and for all $i \in \ints{n}$,  $\bangloop{\alpha^{mk}}{i} \in \kit1(a_i)$. Since $P$ is a stabilized profunctor and $\bangloop{\alpha^{mk}}{i}\cdot p_i \cdot \bangloop{\beta^{mk}}{i} =p_i$, we have $\bangloop{\beta^{mk}}{i} \in \kit2(b_i)$ for all $i \in \ints{n}$. It remains to show that $\beta^{mk} \neq \id$. If $\beta^{mk} = \id$, then $\tau^{mk} = \id$ and $\bangloop{\beta^{mk}}{i}= \beta^{mk}_i = \id \in \kit2^\orth(b_i)$ for all $i$. It implies that $\sigma^{mk} =\id$ and $\bangloop{\alpha^{mk}}{i} = \alpha^{mk}_i \in \kit1^\orth(a_i)$ for all $i \in \ints{n}$ as $P$ is a stabilized profunctor. Hence, we must have $\alpha^{mk}_i = \id$ for all $i$ which implies that $\alpha^{mk} =\id$ since $\alpha^{mk}_i \in \kit1(a_i)$ by assumption. 
		\item Assume that $\beta \in \bigcup \oc\kit2^\orth(v)$ i.e. for all $m \in \N$, if for all $i \in \ints{n}$, $\bangloop{\beta^m}{i} \in \kit2(b_i)$ then $\beta^m = \id$. Assume that there exists $m \in \N$ such that for all $i \in \ints{n}$, $\bangloop{\alpha^m}{i} \in \kit1(a_i)$. Since $\bangloop{\alpha^m}{\varphi(i)} \cdot p_i \cdot \bangloop{\beta^m}{i} = p_i$ for all $i \in \ints{n}$ and $P$ is a stabilized profunctor, we have $\bangloop{\beta^m}{i} \in \kit2(b_i)$ for all $i \in \ints{n}$ which implies that $\beta^m = \id$. Hence, $\tau^m = \id$ and for all $i \in \ints{n}$, $\bangloop{\beta^m}{i} = \beta^m_i = \id \in \kit2^\orth(b_i)$ which entails that $\sigma^m =\id$ and $\bangloop{\alpha^m}{i} = \alpha^m_i  \in \kit1^\orth(a_i)$ since $P$ is a stabilized profunctor. We conclude that $\alpha^m_i = \id$ for all $i\in \ints{n}$ so that $\alpha^m = \id$ as desired. \qedhere
	\end{itemize}
\end{proof}

As with the symmetric monoidal structure, 
the many coherence axioms required for a  pseudo-comonad are immediately verified given the situation in $\Prof$. All that is needed is that components of the counit and comultiplication, given above, are stabilized profunctors. 
\begin{lem}\label{lem:derelectionStable}
	For a Boolean kit $\kitstr{\gpd1}=(\gpd1,\kit1)$, the profunctor $\der[\gpd1] : \Sym \gpd1\profto \gpd1$ is a stabilized profunctor $\kitstr{\oc \gpd1} \profto \kitstr{\gpd1}$.
\end{lem}
\begin{proof}
	This follows directly from observing that for a sequence $\seq{a}$ of length one in $\Sym \gpd1$, 
	\[
	\oc \kit1(\seq{a}) =\{ (\id, \seq{\alpha}) \suchthat \alpha \in \kit1(a)\} \cong \kit1(a). \qedhere
	\]
\end{proof}

The next lemma is another technical verification based on the characterisation of Lemma~\ref{lem:characterisationDoubleOrth}. 
\begin{lem}
	For a Boolean kit $(\gpd1,\kit1)$ and a sequence $\seq{u_1, \dots, u_n} \in \Sym \Sym \gpd1$, there is a mapping 
	\[
	\overline{(-)}: \Endo(\seq{u_1, \dots, u_n}) \to \Endo (u_1 \otimes \dots \otimes u_n)
	\]
	such that for $\beta \in \Endo(\seq{u_1, \dots, u_n})$, we have:
	\[\begin{aligned}
		\beta \in \bigcup \oc \oc \kit1(\seq{u_1, \dots, u_n})\quad &\Leftrightarrow \quad \bar{\beta}\in \bigcup \oc \kit1(u_1 \otimes  \dots \otimes u_n) \qquad\text{ and}\\
		\beta \in \bigcup (\oc \oc \kit1)^\orth(\seq{u_1, \dots, u_n})\quad &\Leftrightarrow \quad \bar{\beta}\in \bigcup (\oc \kit1)^\orth(u_1 \otimes  \dots \otimes u_n) 
	\end{aligned}
	\]
\end{lem}

\begin{proof}
	Let $\beta = (\tau,\seq{\beta_i}_{i \in \ints{n}} )$ be in $\Endo(\seq{u_1, \dots, u_n})$ and let $n_i$ be the length of the sequence $u_i= \seq{a_1^i, \dots, a_{n_i}^i}$ for $1 \leq i \leq n$ and define $m$ to be $\sum_{i\in \ints{n}} n_i$.
	For each $i$, $\beta_i$ consists of a permutation $\phi_i : \ints{n_i} \to \ints{n_i}$ and an $n_i$-tuple of morphisms $\seq{\delta_j^i : a_j^i \to a_{\phi_i(j)}^{\tau(i)}}_{j \in \ints{n_i}}$ where $u_i =\seq{a^i_1, \dots, a^i_{n_i}}$.
	Define $\bar{\beta} =(\bar{\tau}, \seq{\bar{\beta}_j}_{j \in \ints{m}}) \in \Endo(u_1 \otimes\dots\otimes  u_n)$ 
	as follows: the permutation $\bar{\tau}: \ints{m} \xrightarrow{\sim}\ints{m}$ is defined by:
	\[
	\bar{\tau} : k \mapsto \phi_{l+1}(k - \sum\limits_{i=1}^l n_i)+ \sum\limits_{i=1}^{\tau(l+1)-1}n_i \quad \text{if} \quad \sum\limits_{i=1}^l n_i < k \leq \sum\limits_{i=1}^{l+1} n_i
	\]
	and for $1\leq k \leq m$, we define $\bar{\beta}_k: a_k \rightarrow a_{\bar{\tau}(k)}$ as $\bar{\beta}_k := \delta_j^i$ where $i:=l$ and $j:=k - \sum\limits_{i=1}^l n_i$ if $\sum\limits_{i=1}^l n_i < k \leq \sum\limits_{i=1}^{l+1} n_i$.
	The idea is that if $\sum\limits_{i=1}^l n_i < k \leq \sum\limits_{i=1}^{l+1} n_i$, then $k$ is in the subsequence $u_{l+1}$. 
	
	Note that for $a^i_j$ in $u_1 \otimes \dots \otimes u_n$, we have $\bangloop{\bar{\beta}}{a^i_j} = \bangloop{\bangloop{\beta}{u_i}}{a^i_j}$, therefore $\beta$ is in $\bigcup\kit1^{\Sym \Sym}(\seq{u_1, \dots, u_n})$ if and only if $\bar{\beta}$ is in $\bigcup\kit1^\Sym(u_1 \otimes \dots \otimes u_n)$. Since we further have $\overline{\beta^n} = (\bar{\beta})^n$, we have  $\beta \in \bigcup (\oc\oc\kit1)^\dorth(\seq{u_1, \dots, u_n})$ if an only if $\bar{\beta}\in \bigcup (\oc\kit1)^\orth(u_1 \otimes  \dots \otimes u_n)$. It also implies that the following formulae are also equivalent:
	\[\begin{aligned}
		\forall n\in\N, \beta^n=\id \vee (\exists m\in\N , \beta^{nm}\neq \id &\wedge \beta^{nm}\in \bigcup \kit1^{\Sym\Sym}(\seq{u_1, \dots, u_n}))\\
		\forall n\in\N, \bar{\beta}^n=\id \vee (\exists m\in\N , \bar{\beta}^{nm}\neq \id &\wedge \bar{\beta}^{nm}\in \bigcup \kit1^\Sym(u_1 \otimes \dots \otimes u_n)
	\end{aligned}
	\]
	which implies that $\beta \in \bigcup \oc \oc \kit1(\seq{u_1, \dots, u_n})$ if an only if $\bar{\beta}\in \bigcup \oc \kit1(u_1 \otimes  \dots \otimes u_n)$ by Definition \ref{def:saturatedKit}.
\end{proof}

\begin{cor}\label{lem:diggingStable}
	For a Boolean kit $\kitstr{\gpd1}=(\gpd1,\kit1)$, the profunctor $\dig[\gpd1] : \Sym \gpd1\profto \Sym \Sym \gpd1$ is a stabilized profunctor $\kitstr{\oc \gpd1} \profto \kitstr{\oc \oc\gpd1}$.
\end{cor}

\begin{proof}
	
	Let $\gamma : u_1 \otimes\dots\otimes  u_n \to v$ be in $\dig[\gpd{A}](\seq{u_1, \dots, u_n},v)$, it is of the form $(\rho, \seq{\gamma_k}_{k \in \ints{m}})$ where $\rho : \ints{m} \xrightarrow{\sim}\ints{m} $ and
	$\gamma_k : (u_1 \otimes\dots\otimes  u_n)_k \to v_{\rho(k)}$ for $1\leq k \leq m$. Assume that there exists $\alpha = (\sigma, \seq{\alpha_k}_{k \in \ints{m}}) \in \Endo(v)$ and $\beta= (\tau,\seq{\beta_i}_{i \in \ints{m}} ) \in \Endo(\seq{u_1, \dots, u_n})$ such that $\alpha \cdot\gamma \cdot \beta= \gamma$. The morphism $ \bar{\beta} : u_1 \otimes \dots \otimes u_n \to u_1 \otimes \dots \otimes u_n$ defined in Lemma \ref{lem:diggingStable} verifies $\gamma \cdot \beta=\gamma \circ \bar{\beta}$ which implies that $\alpha\cdot \gamma \cdot \beta= \gamma$ is equivalent to $\alpha\circ \gamma\circ \bar{\beta} = \gamma$ i.e.
	\[
	\sigma \rho \bar{\tau} = \rho \qquad \text{and} \qquad \alpha_{\rho\bar{\tau}(k)}\circ \gamma_{\bar{\tau}(k)} \circ \bar{\beta}_k = \gamma_k \quad \text{for all } 1\leq k \leq m.
	\]
	\begin{itemize}
		\item Assume that $\alpha$ is in $\bigcup \oc\kit1(v)$, then $\bar{\beta}$ is in $\bigcup \oc\kit1(u_1 \otimes \dots \otimes u_n)$ by closure under conjugation, which implies that $\beta$ is in $\bigcup \oc \oc\kit1(\seq{u_1, \dots, u_n})$ by Lemma \ref{lem:diggingStable}.
		\item Assume that $\beta$ is in $\bigcup (\oc \oc\kit1)^\orth(\seq{u_1, \dots, u_n})$, then $\bar{\beta}$ is in $\bigcup (\oc \kit1)^\orth(u_1 \otimes \dots \otimes  u_n)$ by Lemma \ref{lem:diggingStable} which implies that $\alpha$ is in $\bigcup (\oc\kit1)^\orth(v)$ by closure under conjugation.\qedhere
	\end{itemize}
\end{proof}

\begin{prop}
The pseudo-functor $\oc : \SProf \to \SProf$ extends to a pseudo-comonad.  
\end{prop}

As usual for orthogonality models of linear logic, the construction
$\wn (\gpd1, \kit1) := (\Sym \gpd1, (\oc \kit1^\perp)^\perp )$ induces
a pseudo-monad structure on $\SProf$ that interprets the ``why not"
modality of linear logic. Note that, unlike in $\Prof$, the
connectives $\oc$ and $\wn$ are distinguished. 

The coKleisli bicategory $\SProf_\oc$ is denoted $\SEsp$, and its
morphisms are called \defn{stable species of structures}. Besides the
connection to stabilized profunctors,
this choice of terminology is justified by their extensional
characterisation as stable functors, which we provide in the second
part of this paper.

\subsection{Cartesian closed structure.}
In this section, we show that the cartesian closed structure of $\Esp = \Prof_{\Sym}$, first outlined in \cite{FioreCartesian2008}, extends to  $\SEsp$.

The coKleisli bicategory $\SProf_{\oc}$ has binary products inherited from the linear category $\SProf$, i.e. given by the $\with$ construction of Sec.~\ref{sec:biproducts}, and terminal object $(\Zero, \varnothing)$. From here the path to cartesian closure is relatively standard, though not immediate. The structure is derived from the monoidal closed structure in $\SProf$ via a fundamental property of $\oc$ often referred to as the \emph{Seely equivalence}:
\[
	\oc (\kitstr{\gpd1} \with \kitstr{\gpd2} ) \simeq \oc  \kitstr{\gpd1}\otimes\oc \kitstr{\gpd2}.
\]
This property is derived from an  adjoint equivalence 
\[
\begin{tikzcd}
\Sym (\gpd1 \with \gpd2) \arrow[bend left=13]{r}{S_{\gpd1, \gpd2}} \arrow[phantom]{r}[description]{\simeq} & \Sym \gpd1 \tensor \Sym \gpd2 
\arrow[bend left=13]{l}{T_{\gpd1, \gpd2}}
\end{tikzcd}
\]
 in the bicategory $\Prof$, which  extends to $\SProf$.

\begin{lemC}[\cite{FioreCartesian2008}]\label{lem:SeelyEquivSymSpecies}
	For categories $\gpd1$ and $\gpd2$, consider the functor $ \gpd1 \with \gpd2 \to \Sym \gpd1 \otimes \Sym \gpd2$ mapping $(1, a)$ to $(\seq{a}, \seq{})$ and $(2,b)$ to $(\seq{}, \seq{b})$, by the universal property of the completion $\Sym$ it induces a functor $\Sym (\gpd1 \with \gpd2) \to \Sym \gpd1 \otimes \Sym \gpd2$ that we denote by $S_{\gpd1, \gpd2}$. We use the same notation as in~\cite{FioreCartesian2008} and write $(w.1, w.2)$ for the image of $w \in \Sym(\gpd1 \with \gpd2)$ by $S_{\gpd1, \gpd2}$.
	Let $T_{\gpd1, \gpd2} : \Sym \gpd1 \otimes \Sym \gpd2 \to \Sym (\gpd1 \with \gpd2)$ be the functor defined by
	\[
	(u,v) \mapsto (\Sym \colimin_1 u) \otimes (\Sym \colimin_2 v).
	\]
	The functors $T_{\gpd1, \gpd2}$ and $S_{\gpd1, \gpd2}$ form an equivalence of categories $\Sym (\gpd1 \with \gpd2) \simeq \Sym \gpd1 \otimes \Sym \gpd2$. We also have $\Sym \Zero \simeq \One$.
\end{lemC}

\begin{lem}\label{lem:inclusionLiftTensorPar}
	For Boolean kits $\kitstr{\gpd1}=(\gpd1,\kit1)$ and $\kitstr{\gpd1}=(\gpd1,\kit1)$, the mapping $(\alpha,\beta) \mapsto \alpha \otimes \beta$ induces an inclusion 
	\[
	\kit1^\Sym(u) \times \kit2^\Sym (v) \subseteq (\kit1 \with \kit2)^\Sym(u \otimes v)
	\]
\end{lem}

\begin{proof}
	The inclusion follows from the following observation: for $a\in \gpd1$, $\bangloop{\alpha\otimes \beta}{(1,a)} = (1, \bangloop{\alpha}{a})$ and for $b\in \gpd2$, $\bangloop{\alpha \otimes \beta}{(2,b)} = (2, \bangloop{\beta}{b})$.
\end{proof}

\begin{lem}
For Boolean kits $\kitstr{\gpd1}=(\gpd1,\kit1)$ and $\kitstr{\gpd2}=(\gpd2,\kit2)$, the Seely profunctors $S_{\gpd1, \gpd2}$ and $T_{\gpd1, \gpd2}$ are stabilized profunctors, so that there is an adjoint equivalence $\oc (\kitstr{\gpd1} \with \kitstr{\gpd2} ) \simeq \oc  \kitstr{\gpd1}\otimes\oc \kitstr{\gpd2}$ in $\SProf$. We also have $\oc \kitstr{\Zero} \simeq \kitstr{\One}$.
\end{lem}

\begin{proof}
	We show that for $(\alpha,\beta) \in (\Endo(u) \times \Endo(v))$, the mapping $(\alpha,\beta) \mapsto \alpha \otimes \beta$ induces an inclusion 
	\[
	\oc \kit1(u) \times \oc \kit2 (v) \hookrightarrow \oc (\kit1 \with \kit2)(u \otimes v).
	\]
	Assume that $(\alpha,\beta)$ is in $\oc \kit1(u) \times \oc \kit2 (v)$ and that there exists $n \in \N$ such that $(\alpha\otimes \beta)^n \neq \id$. It implies that $(\alpha,\beta)^n \neq (\id, \id)$, we suppose without loss of generality that $\alpha^n \neq \id$. Since $\alpha$ is in $\oc \kit1(u)$, there exists $m \in \N$ such that $\alpha^{nm} \neq \id$ and $\alpha^{nm} \in \kit1^\Sym(u)$. If $\beta^{nm} = \id$ then $\beta^{nm} \in \kit2^\Sym (v)$ so that $(\alpha,\beta)^{nm} \in \kit1^\Sym(u) \times \kit2^\Sym (v)$ which implies that $(\alpha\otimes \beta)^{nm} \in (\kit1 \with \kit2)^\Sym(u \otimes v)$ by Lemma \ref{lem:inclusionLiftTensorPar} and $\alpha^{nm} \neq \id$ entails that $(\alpha\otimes \beta)^{nm} \neq \id$.
	
	If $\beta^{nm} \neq \id$ then since $\beta^{nm} \in \oc \kit2 (v)$, there exists $l \in \N$ such that $\beta^{nml}\neq \id$ and $\beta^{nml} \in \kit2^\Sym (v)$. Since $\kit1^\Sym(u)$ is closed under powers, $\alpha^{nml}$ is in $\kit1^\Sym(u)$ as well. Hence, by Lemma \ref{lem:inclusionLiftTensorPar}, $(\alpha\otimes \beta)^{nml} \in (\kit1 \with \kit2)^\Sym(u \otimes v)$ and $\beta^{nml} \neq \id$ implies $ (\alpha\otimes \beta)^{nml} \neq \id$. Hence, $(\alpha\otimes \beta)$ is in $\oc (\kit1 \with \kit2)(u \otimes v)$ as desired. 
	
	We now obtain the desired inclusion.
	\[
	(\oc \kit1 \otimes \oc \kit2)(u,v) = (\oc \kit1(u) \times \oc \kit2 (v))^\dorth \hookrightarrow \oc (\kit1 \with \kit2)^\dorth(u \otimes v) = \oc (\kit1 \with \kit2)(u \otimes v).
	\]

	We show that for $\gamma \in \Endo(w)$, the mapping $\gamma \mapsto (\gamma.1, \gamma.2)$ induces an inclusion 
	\[
	\oc (\kit1 \with \kit2)(w) \hookrightarrow (\oc \kit1 \otimes \oc \kit2)(w.1,w.2).
	\]
	To do so, we show that $(\kit1 \with \kit2)^\Sym (w)\hookrightarrow \oc \kit1(w.1) \times \oc \kit2(w.2)$ and the desired inclusion will follow by applying $(-)^\dorth$ on both sides. Assume that $\gamma$ is in $(\kit1 \with \kit2)^\Sym (w)$, then since for elements $a$ in $w.1$ and $b$ in $w.2$
	\[
	\bangloop{\gamma}{(1,a)} = (1, \bangloop{\gamma.1}{a}) \quad \text{and} \quad \bangloop{\gamma}{(2,b)} = (2, \bangloop{\gamma.2}{b})
	\]
	we have $(\gamma.1,\gamma.2) \in  (\kit1^\Sym (w.1) \times  \kit2^\Sym (w.2)) \subseteq  (\oc \kit1(w.1) \times \oc \kit2 (w.2))$.
\end{proof}

For Boolean kits $(\gpd1,\kit1)$ and $(\gpd2,\kit2)$, the internal hom $\kitstr{\gpd1} \multimap \kitstr{\gpd2} = (\kitstr{\gpd1} \otimes \kitstr{\gpd2}^\orth)^\orth$ in $\SProf$  comes equipped with a linear evaluation morphism $\ev_{\gpd1, \gpd2} : \kitstr{\gpd1} \otimes (\kitstr{\gpd1} \multimap \kitstr{\gpd2}) \profto \kitstr{\gpd2}$. We define the function space in $\SProf_{\oc}$ as 
\[(\gpd1,\kit1) \Rightarrow (\gpd2,\kit2) = \oc(\gpd1,\kit1) \multimap(\gpd2,\kit2)\enspace \] and the cartesian evaluation $\Ev_{\gpd1, \gpd2} : (\Sym(\gpd1 \Rightarrow \gpd2) \with \gpd1) \profto \gpd2$ as the morphism 
\[
\ev_{\Sym\gpd1, \gpd2} \circ (\der[\gpd1 \Rightarrow \gpd2] \otimes \id) \circ S_{\gpd1 \Rightarrow \gpd2, \gpd1}. 
\]
As a composite of stabilized profunctors, $\Ev_{\gpd1, \gpd2}$ is in $\SProf_{\oc}((\kitstr{\gpd1} \Rightarrow \kitstr{\gpd2}) \with \kitstr{\gpd1}, \kitstr{\gpd2})$. For a stable species $P$ in $\SProf_{\oc}(\kitstr{\gpd1} \with \kitstr{\gpd2}, \kitstr{\gpd3})$, its \emph{currying} $\Lambda(P) \in \SProf_{\oc}(\kitstr{\gpd1},  \kitstr{\gpd2} \Rightarrow \kitstr{\gpd3})$ is given by $\lambda (P \circ T_{\gpd1, \gpd2})$, where 
\[
	\lambda : \SProf(\oc \kitstr{\gpd1} \otimes \oc \kitstr{\gpd2}, \kitstr{\gpd3}) \to \SProf(\oc \kitstr{\gpd1}, \oc \kitstr{\gpd2} \multimap \kitstr{\gpd3})
\]
is provided by the monoidal closed structure on $\SProf$.

\begin{thm}
	The bicategory $\SEsp$ is cartesian closed.
\end{thm}

 \begin{proof}
 	The adjoint equivalence 
 	\begin{center}
 		\begin{tikzpicture}[line join=round,yscale=0.5]
 		\node (A) at (0,0) {$\SProf_{\oc}(\kitstr{\gpd1}, (\kitstr{\gpd2} \Rightarrow \kitstr{\gpd3}))$};
 		\node (B) at (4.5, 0) {$\SProf_{\oc}(\kitstr{\gpd1} \with \kitstr{\gpd2},\kitstr{\gpd3})$};
 		\draw [->] (A) to [bend left =30]  node [above] {$\Ev_{\kitstr{\gpd2}, \kitstr{\gpd3}} \circ ((-) \with \kitstr{\gpd2})$} (B);
 		\draw [->] (B) to [bend left =30]  node [below] {$\Lambda$} (A);
 		\node (C) at (2.25,0) {$\bot$};
 		\end{tikzpicture}
 	\end{center}	
 	is obtained directly from the Seely adjoint equivalence.
 \end{proof}

 \subsection{Differential structure.}
 To a combinatorial species $F : \PP \to \Set$, one can associate an exponential generating series of the form $x \mapsto \sum_{n\geq 0} f_n \frac{x^n}{n!}$. Here, $f_n$ is identified with the cardinality of the set $F(\ints{n})$, interpreted as the number of $F$-structures of size $n$. This formal power series has a derivative:
\[
	x \mapsto \sum_{n\geq 1} f_n n \frac{x^{n-1}}{n!} = \sum_{m\geq 0} f_{m+1}  \frac{x^{m}}{m!}
\]
and the corresponding species is the functor $\PP \to \Set : \ints{n} \mapsto F(\ints{n+1})$. 

This combinatorial notion extends to the wider universe of generalized
species \cite{fiore2005mathematical}. A generalized species $F : \Sym\gpd1 \profto \gpd2$ 
has a differential $\mathbf{D}F : \oc \gpd1 \otimes \gpd1 \profto
\gpd2$ with action
\[
	(b, (u, a)) \mapsto F(b, u \otimes \seq{a}), 
\]
which coincides with the above construction when $\gpd1 = \gpd2 =
\One$. 

The combinatorial notion of differentiation has a logical counterpart in the framework of \emph{differential linear logic} put forward by Ehrhard et al. \cite{ehrhard2003differential,EhrhardDiLL16}, which extends linear logic with operators for differentiation. The bicategory $\Prof$ provides a model for differential linear logic, in which logical derivatives are formally assigned a combinatorial interpretation \cite{fiore2005mathematical}. 

The differential operators of $\Prof$ extend to $\SProf$, which we
show is also a model for differential linear logic. Indeed, given a
model of linear logic with biproducts, the structure required to model
differential operators is relatively minimal
\cite{Fiore2007DifferentialSI}. The fact that $\SProf$ is a model for
differential linear logic follows from the forgetful functor to $\Prof$, which ensures that coherence axioms are satisfied, together with the following direct observation: 
\begin{lem}
	For a kit $\kitstr{\gpd1}=(\gpd1,\kit1)$, the profunctor  $\coder[\gpd1] :  \gpd1\profto \Sym \gpd1$ given by $(u, a) \mapsto \Sym\gpd1(u, \seq{a})$ is a stabilized profunctor $\kitstr{ \gpd1} \profto \kitstr{ \oc\gpd1}$.
\end{lem}

We finally note that, although species $\PP \to \Set$ can be differentiated, 
they do not support an adequate notion of \emph{integration} as needed for instance in the resolution of differential equations \cite{LerouxViennot}. 
It is known however that ``$\mathbf{L}$-species'', which
correspond to $\N$-indexed families of sets without any symmetric group actions, enjoy uniqueness and
existence theorems for solutions of differential equations
\cite{SpeciesBLR}. As $\mathbf{L}$-species correspond to our stable
species $\oc \kitstr{\One} \profto \kitstr{\One}$, the bicategory
$\SEsp$ appears as a candidate framework for the study of formal
differential equations, providing a bridge to similar investigations on the logical side \cite{DBLP:conf/lics/Kerjean18,DBLP:conf/fossacs/KerjeanL19}.

\section*{\Large Extensional Theory: Linear and Stable Functors}

In the rest of the paper, we study the extensional aspects of
stabilized profunctors and stable species. We will see, in
Sections~\ref{sec:Stability} and \ref{sec:Linearity}, that both
stabilized profunctors and stable species can be characterized
extensionally as functors between \emph{subcategories} of presheaves,
which one can view as generalized Scott domains. We give a 
biequivalence between the bicategory of stable species and a
$2$-category of \emph{stable} functors, and as a corollary we
obtain that the latter is cartesian closed.

\section{Categories of stabilized presheaves}
\label{sec:stable-presheaves}
First we describe and study these subcategories of presheaves, by considering presheaves that arise as global elements $1 \profto (\gpd1, \kit1)$ in $\SProf$, which we call \emph{stabilized presheaves}.  
m
For a kit $(\gpd1, \kit1)$, a presheaf $X : \gpd1^\op \to \Set$ defines a stabilized profunctor $\One \profto (\gpd1, \kit1)$ if and only if, for all $a \in \gpd1$, $x \in X(a)$, and $\alpha : a \to a$, 
\[
x \act \alpha = x \implies \alpha \in \bigcup \kit1(a). 
\]
This property is more succinctly expressed using the group-theoretic notion of stabilizer: for $a \in \gpd1$ and $x \in X(a)$, we denote by $\Stab(x)$ the subgroup of $\Endo(a)$ consisting of those endomorphisms $\alpha$ such that $x\act\alpha = x$. We now define ($\kit1$-)stabilized presheaves as those presheaves whose elements all have stabilizer in $\kit1$.  

\begin{defi}
\label{def:StSh}
	For a kit $(\gpd1,\kit1)$ we let $\StPSh(\gpd1,\kit1)$ be the full subcategory
	of $\PSh(\gpd1)$ consisting of presheaves $X$ such that for every $a \in \gpd1$ and $x \in X(a)$, 
	$\Stab(x)\in \kit1(a)$; we call these \defn{stabilized presheaves}. 
\end{defi}

When $\kitstr{\gpd1} = (\gpd1, \kit1)$ is a Boolean kit, there is an  equivalence of categories 
$
\SProf(\One, \kitstr{\gpd1}) \simeq \StPSh(\kitstr{\gpd1})
$. In this section, we focus on properties of the subcategory $\StPSh(\kitstr {\gpd1})$ for an arbitrary kit $\kitstr{\gpd1}$. First we observe that for a groupoid $\gpd1$, the operation sending a kit $\kit1$ to the category $\StPSh(\gpd1, \kit1)$ has a right inverse. 
\begin{lem}
For a groupoid $\gpd1$ and a full subcategory $\asubcat$ of $\PSh(\gpd1),$
the family $\Stab(\asubcat)$ defined by 
\[ 
\Stab(\asubcat)(a) = \{ G  \mid \exists X \in \asubcat, \exists x \in X(a), G = \Stab(x) \} 
\] 
is a kit on $\gpd1$. Additionally,  $\Stab(\StPSh(\gpd1, \kit1)) = (\gpd1, \kit1)$ and $\asubcat \subseteq \StPSh(\gpd1, \Stab(\asubcat))$. 
\end{lem}

\subsection{Stabilized quotients of representables.}

For an object $a$ of a groupoid $\gpd1$, consider the representable presheaf $\yon(a) : \gpd1^\op \to \Set$. For $a' \in \gpd1$, every element $\gamma \in \yon(a)(a') = \gpd1(a', a)$ has trivial stabilizer, since in a groupoid $\gamma \comp \alpha = \gamma$ implies $\alpha = \id$. To generate stabilized presheaves, we consider representables \emph{quotiented} by kit subgroups. 
\begin{defi}
For an object $a$ of a groupoid $\gpd1$ and $G \subgroup \Endo(a)$, let
$\extyon a G \in\PSh(\gpd1)$ be the quotient of $\yon(a)$ under $G$; that is,
the colimit of the composite $G\to\gpd1\rightembedding\PSh(\gpd1)$.  
\end{defi}

In elementary terms, the presheaf $\extyon a G$ maps an object $a'$ to the
	quotient of $\yon(a)(a') = \gpd1(a', a)$ under the equivalence relation
	$\qeqrel[G]$ given by $\gamma' \qeqrel[G] \gamma$ if and only if
	$\gamma' \icomp \inv\gamma \in G$. 
For example, $\extyon a{\Triv(a)} \iso \yon(a)$, and $\extyon a{\Endo(a)}(a')$ is a singleton when $a\cong a'$, and $\varnothing$ otherwise. 

\begin{lem}
For $(\gpd1, \kit1)$ a kit, $a \in \gpd1$ and $G \in \kit1(a)$, $\extyon a G \in \StPSh(\gpd1, \kit1)$. 
\end{lem}
\begin{proof}[Proof sketch]
The stabilizer of an element $[\alpha]_{\qeqrel[G]} \in \extyon a G (a')$ is the conjugate subgroup $\alpha^{-1} G \alpha$, which is in $\kit1(a')$ since kits are closed under conjugation.  
\end{proof}

\subsection{Stabilized presheaves as coproduct completions.}

We show that in fact $\StPSh(\gpd1,\kit1)$ is equivalent to the free coproduct completion of the full subcategory of $\PSh(\gpd1)$ spanned by quotients of representables of the form $\extyon a G$, for $a \in \gpd1$ and $G \in \kit1(a)$. 
The first step in showing this is to observe that $\StPSh(\gpd1, \kit1)$ has all coproducts, calculated as in $\PSh(\gpd1)$. This holds because for presheaves $X$ and $Y$, the stabilizer of an element of $X + Y$ is equal to the stabilizer of the corresponding element in $X$ or $Y$.

Next, we show a representation theorem stating that every object of $\StPSh(\gpd1,\kit1)$ can be obtained as a sum of quotients of representables of the form $\extyon a G$, with $a \in \gpd1$ and $G \in \kit1(a)$.

\begin{lem}
\label{lem:StShRepresentation}
	For a kit $(\gpd1,\kit1)$, using the axiom of choice, every presheaf 
	$X \in \StPSh(\gpd1,\kit1)$ has a representation 
 \[ X \cong \coprod_{i\in I}\extyon{a_i}{G_i} \] where each $a_i \in \gpd1$ and $G_i \in \kit1(a_i)$.  
\end{lem}

\begin{proof}
	Using the axiom of choice, we choose a representative $a_c \in \gpd1$ for each connected component $c\in \pi_0 (\gpd1)$. For each $x \in X(a)$, we denote by $o(x) :=\{ x \cdot \alpha \mid \alpha \in \gpd1(a,a) \}$, the orbit of $x$. For $a \in \gpd1$, we write $\Orbits(X(a)):=\{ o(x) \mid x \in X(a)\}$ for the set of orbits of $X(a)$. Using the axiom of choice again, we chose a representative $x_i$ for each orbit $i \in \Orbits(X(a))$.
	We define the presheaf $Y$ as:
	\[
	\coprod_{c \in \pi_0(\gpd1)} \:\coprod_{i \in \Orbits(X(a_c))} \extyon{a_c}{\Stab(x_i)}
	\]
	Let $f : Y \Rightarrow X$ be the natural transformation whose components $f_a$ are given by:
	\[
	(c, i, [\alpha: a \to a_c]) \quad \mapsto \quad x_i \cdot \alpha
	\]
	We first start by showing that $f_a$ is well-defined: assume that $[\alpha]=[\beta]$ i.e. $\alpha \beta^{-1}$ is in $\Stab(x_i)$, then $x_i \cdot (\alpha \beta^{-1}) = x_i$ which implies that $x_i \cdot \alpha = x_i \cdot \beta$ as desired.
	
	For surjectivity, let $x$ be in $X(a)$ and let $c$ be the connected component containing $a$ so there exists a morphism $\gamma : a_c \to a$. Let $x_i$ be the representative of the orbit $o (x\cdot h) \in \Orbits(X(a_c))$ so there exists $\beta : a_c \to a_c$ such that $x \cdot \gamma = x_i \cdot \beta$. The element $(c, o(x \cdot \gamma), [\beta \gamma^{-1}])$ in $Y(a)$ is then a preimage of $x$ as $x_i \cdot \beta \gamma^{-1} =x$.
	
	For injectivity, let $(c, i, [\alpha: a \to a_c])$ and $(d, j, [\beta: a \to a_d])$ be two elements of $Y(a)$ such that $x_i \cdot \alpha= x_j \cdot \beta$. Since $\beta \alpha^{-1} \in \gpd1(a_c, a_d)$, $a_c$ and $a_d$ are in the same connected component which implies that $c=d$ and therefore $a_c = a_d$ since there is one representative chosen for each component. Hence, we have $x_i = x_j \cdot (\beta \alpha^{-1}) \in X(a_c)$ which implies that $x_i$ and $x_j$ are in the same orbit so $i=j$ and $x_i =x_j$ since there is one representative chosen for each orbit. We now have $x_i = x_j \cdot (\beta \alpha^{-1})$ which implies that $[\alpha]=[\beta]$.
	
	Naturality of $f$ is immediate so we conclude that it is a natural isomorphism.
\end{proof}

\begin{cor}\label{cor:RepresentationQuantitativePresheaves}
	For a kit $(\gpd1, \kit1)$, using the axiom of choice, every presheaf in
	$\StPSh(\gpd1, \kit1)$ has a representation as a filtered colimit of finitely
	presentable presheaves of the form $\coprod_{i\in I}\extyon{a_i}{G_i}$ with
	the set $I$ finite and each group $G_i \in \kit1(a_i)$ is finitely generated for $i\in I$.
\end{cor}

From this we derive the following: 
\begin{prop}
For $(\gpd1, \kit1)$ a kit, the category $\StPSh(\gpd1, \kit1)$ is equivalent to the free coproduct completion of the full subcategory of $\PSh(\gpd1)$ spanned by the $\extyon a G$, for $a \in \gpd1$ and $G \in \kit1(a)$.    
\end{prop}
\begin{proof}[Proof sketch]
Write $\coprod_{\kit1}$ for this completion. Since $\StPSh(\gpd1, \kit1)$ contains the $\extyon a G$ and has all coproducts there is a canonical functor $\coprod_{\kit1} \to \StPSh(\gpd1, \kit1)$. It is full and faithful by a straightforward argument and essentially surjective on objects by Lemma~\ref{lem:StShRepresentation}.
\end{proof}

\subsection{Presheaves and orthogonality.}
\newcommand{\FS}{\mathbf{FS}}

In this section, we show that the duality enjoyed by the kits on a groupoid $\gpd1$ through the $(-)^\perp$ construction translates to a duality between presheaves and co-presheaves over $\gpd1$. Presheaves over $\gpd1$ (equivalently profunctors $\One \profto \gpd1$) are thought of as modelling closed programs of type $\gpd1$ and co-presheaves (or profunctors $\gpd1 \profto \One$) provide a notion of program environment\footnote{This duality has many  alternative descriptions as proofs/counterproofs, Player/Opponent, etc.}.
In this viewpoint, orthogonality provides a refined setting in which one explicitly controls the possible interactions between programs and environments. Here we ensure that the stabilizer groups on either side of the interaction have no elements in common but the identity.

\begin{defi}
For a groupoid $\gpd1$, a presheaf $X : \gpd1^\op \to \Set$ and a co-presheaf $Y : \gpd1 \to \Set$, we say that $X$ and $Y$ are \emph{orthogonal}, written $X \orth Y$, if the presheaf on $\gpd1$ with object mapping 
$
a \longmapsto X(a) \times Y(a) 
$
and functorial action defined by 
\[
(x, y) \act \alpha = (x \act \alpha, \alpha^{-1} \act y)
\]
is \defn{free}: all its elements have trivial stabilizer. Explicitly, a presheaf $Z\in\PSh(\gpd1)$ is free whenever, for
all $\alpha, \beta: a'\to a$ in $\gpd1$ and $z\in Z(a)$,
\[
z\act\alpha = z\act \beta \implies \alpha = \beta.
\]
\end{defi}

Orthogonality in this sense is the counterpart, in our model, of the orthogonality relations underlying a number of well-known models of classical linear logic \cite{GirardCoherence, LoaderTotality, EhrhardFiniteness, DanosProbaCoherence}. All of these models were given a unified, axiomatic treatment based on \emph{double-glueing} categories along hom-functors, in a framework developed by Hyland and Shalk \cite{GlueingHylandShalk}. The bicategorical version of this framework, which should encompass the model of this paper, has yet to be worked out. 

The following construction may then be performed on subcategories of presheaves: for any full subcategory $\asubcat$ of $\PSh(\gpd1)$, the category $\asubcat^\orth$ is the full subcategory of $\PSh(\gpd1^\op)$ with objects 
\[
\{ Y : \gpd1 \to \Set \mid \forall X \in \asubcat, X \perp Y \}.
\]
The two notions of duality (on kits and full subcategories) are closely related. Let $\FS(\gpd1)$ be the (large) poset of full subcategories of $\PSh(\gpd1)$ under inclusion. As with kits, the orthogonality relation induces a Galois connection
	\begin{center}
	\begin{tikzpicture}[thick, line join=round,yscale=0.5]
	\node (A) at (0,0) {$\FS(\gpd1)^\op$};
	\node (B) at (3.5, 0) {$\FS(\gpd1^\op)$};
	\draw [->] (A) to [bend left =30]  node [above] {$(-)^\orth$} (B);
	\draw [->] (B) to [bend left =30]  node [below] {$(-)^\orth$} (A);
	\node (C) at (1.75,0) {$\bot$};
	\end{tikzpicture}
\end{center}
whose fixed points are those $\asubcat$ verifiying $\asubcat^\dorth = \asubcat$. 
Thus we have two notions of duality, respectively intensional (based on subgroups) and extensional (based on presheaves). These are fundamentally connected via the equations 
\[
\StPSh(\gpd1, \kit1)^\orth = \StPSh(\gpd1^\op, \kit1^\orth) \quad \text{and} \quad \Stab(\asubcat)^\orth = \Stab(\asubcat^\orth)
\]
from which we derive another characterization of Boolean kits as those kits $(\gpd1, \kit1)$ satisfying $\StPSh(\gpd1, \kit1)^\dorth = \StPSh(\gpd1, \kit1)$. Moreover, our earlier definition of stabilized profunctor (Definition \ref{def:sprofunctor}) may be rephrased in extensional style: 
\begin{lem}
	\label{lem:qprofunctor_characterisation}
	A profunctor $P : \gpd1 \profto \gpd2$ induces two functors $P\lanyon:= \Lan_{\yon[\gpd1]}P : \PSh(\gpd1) \to \PSh(\gpd2)$ and $P\lanyonop := (P^\dual)\lanyon: \PSh(\gpd2^\op)\to\PSh(\gpd1^\op)$. For Boolean kits $\kitstr{\gpd1}=(\gpd1,\kit1)$ and $\kitstr{\gpd2}=(\gpd2,\kit2)$, $P$ is a stabilized profunctor from
	$\kitstr{\gpd1}$ to $\kitstr{\gpd2}$ if and
	only if 
	\begin{center}
		\begin{minipage}{3cm}
			\begin{tikzpicture}[scale=.75]
			\node (A) at (0, 2) {$\PSh(\gpd1)$};
			\node (B) at (2.5, 2) {$\PSh(\gpd2)$};
			\node (C) at (0, 0) {$\StPSh(\kitstr{\gpd1})$};
			\node (D) at (2.5, 0) {$\StPSh(\kitstr{\gpd2})$};
			
			\draw [->] (A) -- node [above] {$P\lanyon$} (B);
			\draw [right hook->] (D) to node {} (B);
			\draw [right hook->] (C) to node {} (A);
			\draw [->, dashed] (C) to node {} (D);
			\end{tikzpicture}
		\end{minipage}
		\qquad and \qquad
		\begin{minipage}{3cm}
			\begin{tikzpicture}[scale=.75]
			\node (A) at (0, 2) {$\PSh(\gpd2^\op)$};
			\node (B) at (2.5, 2) {$\PSh(\gpd1^\op)$};
			\node (C) at (0, 0) {$\StPSh(\kitstr{\gpd2}^\star)$};
			\node (D) at (2.5, 0) {$\StPSh(\kitstr{\gpd1}^\star)$};
			
			\draw [->] (A) -- node [above] {$P\lanyonop$} (B);
			\draw [right hook->] (D) to node {} (B);
			\draw [right hook->] (C) to node {} (A);
			\draw [->, dashed] (C) to node {} (D);
			\end{tikzpicture}
		\end{minipage}
	\end{center}
	That is, $P\lanyon: \PSh(\gpd1) \rightarrow \PSh(\gpd2)$ restricts to a
	functor $\StPSh(\kitstr{\gpd1}) \to \StPSh(\kitstr{\gpd2})$ and 
	$P\lanyonop: \PSh(\gpd1^\op) \rightarrow \PSh(\gpd2^\op)$ restricts to a
	functor $\StPSh(\kitstr{\gpd2}^\star) \to \StPSh(\kitstr{\gpd1}^\star)$.
\end{lem}

\begin{proof}
	\begin{itemize}
		\item[] 
		\item[$(\Rightarrow)$] Assume that $P:\gpd1 \profto \gpd2$ is a stabilized profunctor. For $X\in \StPSh(\kitstr{\gpd1})$ and $b\in\gpd2$ we want to show that the stabilizer of every element in $P^{\#} X(b) = \int^{a\in\gpd1} P(b,a) \times X(a)$ is in $\kit2(b)$. Let $t = p \coendtensor[a] x$ be in $P^{\#} X(b)$.  For every $\beta: b\to b$ in $\Stab(t)$, we have $(p \cdot \beta) \coendtensor[a] x 
		= (p \coendtensor[a] x) \cdot \beta	= p \coendtensor[a] x$, i.e. there exists $\alpha : a \rightarrow a$ in $\gpd1$ such that $\alpha \cdot p \cdot \beta = p$ and $x \cdot \alpha = x$. Since $X$ is in $\StPSh(\kitstr{\gpd1})$, we must have $\alpha \in \bigcup \kit1(a)$ and since $P$ is a stabilized profunctor, we must also have $g \in \bigcup \kit2(b)$ as desired. One shows that $P_{\#} Y$ is in $\StPSh(\kitstr{\gpd1}^\orth)$ for $Y \in \StPSh(\kitstr{\gpd2}^\orth)$ analogously.

		\item[$(\Leftarrow)$] Assume that $P^{\#}(\StPSh(\kitstr{\gpd1})) \hookrightarrow \StPSh(\kitstr{\gpd2})$. Let $a\in\gpd1$, $b\in\gpd2$, $p\in P(b,a)$, $\alpha\in\gpd1(a,a)$, and $\beta\in\gpd2(b,b)$ be such that $\alpha\cdot p\cdot \beta=p$. Suppose that $\alpha\in \bigcup \kit1(a)$ and let $X = \extyon a {\langle \alpha \rangle} \in \StPSh(\kitstr{\gpd1})$. We then have $P^{\#}(X) \in \StPSh(\kitstr{\gpd2})$ by hypothesis. Hence, for the element $p \coendtensor[a] {\langle \alpha \rangle} = p \coendtensor[a] {\langle \alpha \rangle} \alpha
		= (\alpha \cdot p) \coendtensor[a] {\langle \alpha \rangle}$ in $P^{\#} X(b)$, it follows that $(p \coendtensor[a] {\langle \alpha \rangle}) \cdot \beta
		= (\alpha \cdot p \cdot \beta) \coendtensor[a] {\langle \alpha \rangle}
		=  p \coendtensor[a] {\langle \alpha \rangle}$ which implies $\beta \in \bigcup \kit2(b)$ as desired. 
		
		Assuming $P_{\#}(\StPSh(\kitstr{\gpd2}^\orth)) \hookrightarrow \StPSh(\kitstr{\gpd1}^\orth)$,
		for $\beta \in \bigcup \kit2^\perp(b)$, one considers 
		$\extyon b {\langle \beta \rangle} \in \StPSh(\kitstr{\gpd2}^\orth)$ and reasons
		analogously to show that ${\alpha\in \bigcup \kit1^\perp(a)}$.\qedhere
	\end{itemize}
\end{proof}

\subsection{Stabilized presheaves on Boolean kits}

We now focus on categories of stabilized presheaves associated with Boolean kits and discuss the additional properties which hold there. A major difference with the kit case is the existence in $\StPSh(\gpd1, \kit1)$ of all non-empty limits and filtered colimits. 
For the latter, we use that kits are closed under lubs of directed sets of subgroups.
\begin{lem}
For a Boolean kit $(\gpd1, \kit1)$, $\StPSh(\gpd1, \kit1)$ has filtered colimits.  
\end{lem}

The case of non-empty limits is immediate given the following:
\begin{lem}\label{lem:StPShconnectedLimits}
For a Boolean kit $(\gpd1, \kit1)$ and a morphism $X \xrightarrow{f} Y$ in $\PSh(\gpd1)$, if $Y \in \StPSh(\gpd1, \kit1)$ then $X \in \StPSh(\gpd1, \kit1)$.  
\end{lem}
\begin{proof}[Proof idea]
For $a \in \gpd1$ and $x \in X(a)$, $\Stab(x) \subseteq \Stab(f_a x) \in \kit1(a)$, and Boolean kits are closed under subgroups. 
\end{proof}

\begin{cor}\label{cor:QuantitativePresheavesCreation}
	For a Boolean kit $(\gpd1,\kit1)$, the embedding 
	$\StPSh(\gpd1,\kit1)\rightembedding\PSh(\gpd1)$ creates isomorphisms,
	coproducts, filtered colimits, epimorphisms and non-empty limits.  
\end{cor}
\begin{proof}
	Let $\{X_i\}_{i \in I}$ be a finite set of presheaves in $\StPSh(\kitstr{\gpd1})$. We show that their coproduct in $\PSh{\gpd1}$ is an element of $\StPSh(\kitstr{\gpd1})$. For $(i, x) \in  (\coprod_{i \in I} X_i)(a) = \coprod_{i \in I} X_i(a)$, ${\Stab}_{\coprod_{i \in I} X_i}(i, x)= {\Stab}_{X_i}(x)$ which is in $\kit{A}(a)$, so we are done.
	
	Let $D : \filt{I} \to \StPSh(\kitstr{\gpd1})$ be a filtered diagram. Denote by $X_i$ the presheaf $D(i)$ for $i \in \filt{I}$ and write $X$ for the colimit of $D$. Let $a$ be in $\gpd1$, for every $x \in X(a)$, there exists $i \in \filt{I}$ and $y \in X_i(a)$ such that $(\colimin_i)_a(y) = x$, where $\colimin_i : X_i \to X$ denotes the cocone component. By naturality of $\colimin_i$, if $\alpha: a \to a$ is in ${\Stab}_X(x)$ then $(\colimin_i)_a (y \alpha ) = (\colimin_i)_a(y)$. Since $\filt{I}$ is filtered, there exists $j \in \filt{I}$ and $f : i \to j$ such that $D(f)_a(y \cdot \alpha) = D(f)_a( y)$. Hence, $\alpha\in {\Stab}_{X_j}(D(f)_a(y))$ so $\alpha$ is in $\bigcup \kit{A}(a)$ as desired.

	We show that a morphism $e : X \Rightarrow Y$ is an epimorphism in $\StPSh(\kitstr{\gpd1})$ if and only if its components $e_a : X(a) \to Y(a)$ are surjective for all $a \in \gpd1$. Assume that there exists $a \in \gpd1$ and $y \in Y(a)$ such that for all $x\in X(a)$, $e_{a}(x) \neq y$. By Lemma \ref{lem:StShRepresentation}, $Y \cong \coprod_{i\in I} \extyon {a_{i}} {G_{i}}$ where each $G_{i}$ is a group in $\kit{A}(a_{i})$. Let $j \in I$ be such that $y$ is in the component $\extyon {a_{j}} {G_{j}}(a)$. We show that for all $b \in \gpd1$ and $z \in \extyon {a_{j}} {G_{j}}(b)$, $z$ is not in the image of $e_b$. Assume that there exists $x \in X(b)$ such that $e_b (x) = z$, then for any representative $\alpha \in \gpd1(a,a_j)$ of the equivalence class $y$ and $\beta \in \gpd1(b, a_j)$ of the equivalence class $z$, we obtain by naturality of $e$ that
	\[
	e_a(x \cdot (\beta^{-1}\alpha)) = (e_b (x))\cdot (\beta^{-1}\alpha) = [\beta] \cdot (\beta^{-1}\alpha) = [\beta \beta^{-1}\alpha] =[\alpha]=y.
	\]
	Hence, we can factor $e$ as 
	\begin{center}
		\begin{tikzpicture}[thick]
			\node (A) at (0, 0) {$X$};
			\node (B) at (2.5, 0) {$\sum_{\substack{i \in I \\ i \neq j}}\extyon {a_{i}} {G_{i}} $};
			\node (C) at (6.5, 0) {$\sum_{i\in I} \extyon {a_{i}} {G_{i}} \cong Y$};
			
			\draw [->] (A) -- node [above] {} (B);
			\draw [right hook->] (B) to node [above] {}  (C);
		\end{tikzpicture}
	\end{center}
	Let $\colimin_Y : Y \to Y + \One$ be the left hand coproduct inclusion and let $f : Y \to Y+\One$ be the natural transformation whose components $f_a$ are given by 
	\[ \sum_{i\in I} \extyon {a_{i}} {G_{i}}(a) \ni (i, y)  \mapsto \begin{cases}
		\colimin_a (i,y) \in \{1\} \times Y(a) & \text{ if } i \neq j\\
		(2, \star) \in \{2\} \times \One(a)& \text { if } i =j
	\end{cases}\]
	It is immediate that $\colimin_Y e = f e$ but we do not have $\colimin_Y = f$ contradicting the assumption of $e$ being an epimorphism. Hence, we must have that $e$ is pointwise surjective as desired.
\end{proof}

On the other hand, the category $\StPSh(\gpd1, \kit1)$ has a terminal object (an empty limit) only if $\kit1$ is maximal:  
\begin{prop}
	For a Boolean kit $(\gpd1,\kit1)$, the following are equivalent.
	\begin{enumerate}
				\item 
		The terminal presheaf is in $\StPSh(\gpd1,\kit1)$.
		\item 
		For all $a\in\gpd1$, 
		$\kit1(a) = \{ G \mid G\subgroup\Endo(a) \}$.
					\item 
		$\StPSh(\gpd1,\kit1) = \PSh(\gpd1)$. 
			\end{enumerate}
\end{prop}
\begin{proof}
	\begin{itemize}
		\item[]
		\item $(1 \Rightarrow 2)$ Let $a$ be in $\gpd1$ and $G$ be a subgoup of $\gpd1(a,a)$. Since $\extyon{a}{G}$ is in $\StPSh(\kitstr{\gpd1})$ by assumption and $\Stab(\extyon{a}{G})= \{ K \mid K\leq G \}$, we obtain that $G\in \kit{A}(a)$ as desired.
		\item $(2 \Rightarrow 3)$ Let $1_{\gpd1}: \gpd1^{op} \rightarrow \Set$ denote the terminal presheaf. For all $a \in\gpd1$, $\Stab(1_{\gpd1})(a) =   \{ G\leq \gpd1(a,a) \mid a\in\gpd1 \}$ which implies that $1_{\gpd1}$ is in $\StPSh(\kitstr{\gpd1})$ by definition.
		\item $(3 \Rightarrow 1)$ Direct consequence of Lemma \ref{lem:StPShconnectedLimits}.\qedhere
	\end{itemize}
\end{proof}
	
We additionally note that for a Boolean kit $(\gpd1, \kit1)$, the Yoneda embedding $\gpd1 \to \PSh(\gpd1)$ factors through the inclusion $\incl[\kit1]: \StPSh(\gpd1, \kit1) \hookrightarrow \PSh(\gpd1)$ via an embedding $\gpd1 \hookrightarrow \StPSh(\gpd1, \kit1)$ which we denote $\yy[\kit1]$.
 This holds because Boolean kits contain all trivial subgroups, and for an object $a$ of a groupoid $\gpd1$ we have $\extyon a {\{ \id[a] \}} \cong \yon (a)$.

We continue our development in the next section and introduce \emph{stable functors}. We will show that for a stable species $P$, the mapping $P \mapsto \rsLan{P}$ where $\rsLan{P}= (\Lan_{s}P)\incl[\kit1]$ induces an equivalence between the hom-category of stable species and that of stable functors. Among the stable functors we will then identify the \emph{linear} ones, and show they correspond to stabilized profunctors.

\section{Stable functors}
\label{sec:Stability}
\subsection{Local right adjoints and generic factorizations}\label{subsec:LRAgenericfact}

The notion of stability in domain theory is often stated by means of a \emph{minimal data property} first put forward by Berry~\cite{BerryStable}: 
\begin{defi}[Berry]
\label{def:stabilitydomains}
For cpos $(A, \leq_A)$ and $(B, \leq_B)$, a Scott-continuous function $f : A \to B$ is \defn{stable} if for every $x \in A$ and $y \in B$ such that $y \leq_B f x$, there exists a unique minimal $x_0 \in A$ such that $x_0 \leq_A x$ and $y \leq_Bf x_0$.   
\end{defi}
For $x$ and $y$ as in the definition, $x_0$ has an operational meaning  as the unique minimal amount of information used from $x$ to produce the output $y$. We denote this minimal element $x_0$ by $m(f,x,y)$. It is less well-known that the points $m(f,x,y)$ determine a family of \emph{local} left adjoints to the restriction $f_x: A/x \to B/f(x)$ between the slices $A/x =\{ x' \in A \mid x' \leq_A x\}$ and $B/f(x) = \{ y \in B \mid y\leq_B f(x)\}$:
\begin{center}
	\begin{tikzpicture}[thick, line join=round,yscale=0.5, xscale=0.8]
		\node (A) at (0,0) {$A/x$};
		\node (B) at (3.5, 0) {$B/f(x)$};
		\draw [->] (A) to [bend left =30]  node [above] {$f_x$} (B);
		\draw [->] (B) to [bend left =30]  node [below] {$m(f,a,-)$} (A);
		\node (C) at (1.75,0) {$\top$};
	\end{tikzpicture}
\end{center}
This condition is equivalent to stability.
\begin{lem}
For cpos $(A, \leq_A)$ and $(B, \leq_B)$, a Scott-continuous function $f : A \to B$ is \defn{stable} if and only if for every $x \in A$, the restriction $f_x : A/x \to B/fx$ has a left adjoint. 
\end{lem}

For our purposes, this characterization has the advantage that the existence of local adjoints may be considered in full generality: 
\begin{defi}\label{def:localAdjoint}
For categories $\gpd1$ and $\gpd2$, 
a functor $T : \gpd1 \rightarrow \gpd2$ is said to be a 
\emph{local right (resp.~left) adjoint} if for every object 
$a \in \gpd1$, the induced functor $ T_a : \gpd1/_a \rightarrow \gpd2/_{T(a)}$ is a right (resp.~left) adjoint.
\end{defi}

This notion is related to the \emph{multiadjoint functors} initially developed by Diers \cite{diers1977categories} (see also \cite{carboni_johnstone_1995,Weber2004generic, Weber2007familial,Osmond2020diers,gambinokock2013}). Finitary local right adjoints are precisely the normal functors considered by Girard \cite{GirardNormal} and Hasegawa \cite{HasegawaAnalytic}, and the relevance of this notion to stable domain theory was further emphasized by Lamarche \cite{LamarcheThesis} and Taylor \cite{TaylorGroupoidsLinearLogic}. The extensional theory of stable species of structures fits into this line of research:

\begin{defi}\label{def:stablefunctor}
For kits $\kitstr{\gpd1}$ and $\kitstr{\gpd2}$, a functor $\StPSh(\kitstr{\gpd1}) \to \StPSh(\kitstr{\gpd2})$ is called \defn{stable} when it is a finitary, local right adjoint functor that preserves epimorphisms.  
\end{defi}

The requirement of preserving epimorphisms will become clear as we construct a pseudo-inverse to the mapping $P \mapsto  \rsLan{P}$. Epi-preserving functors are of known significance in this area since, as shown by Fiore \cite{Fioreanalytic}, analytic functors between presheaf categories over groupoids may be characterized as finitary functors preserving \emph{wide quasi-pullbacks}; this forces the preservation of epis. 

We note, additionally, that the local right adjoint property ensures that stable functors preserve \emph{wide pullbacks}, connecting with yet another presentation of stable functions in Berry's theory, as Scott continuous functions preserving bounded meets. 

\newcommand{\Bool}{\mathbf{Bool}}
\newcommand{\tru}{\mathbf{t}}
\newcommand{\fal}{\mathbf{f}}

\begin{exa} \label{ex:parallelor}
	In a cartesian closed and extensive category, such as a topos, the finite
	product and finite coproduct functors are stable.  In connection to this, we
	discuss the prototypical non-stable, and hence non-sequential, function from
	Berry's stable domain theory~\cite{BerryStable}.  To this end, let $S$ be the
	Sierpinski space $(0\subset1)$ and, for $\mathrm{Bool} = \setof{\fal,\tru}$,
	consider the \mbox{\em parallel-or} function
	$\textit{por}:S^\mathrm{Bool}\times S^\mathrm{Bool}\to S^\mathrm{Bool}$
	defined as the least monotone function such that:
	$\textit{por}\big(\setof{\fal},\setof{\fal}\big)=\setof{\fal}$ and 
	$\textit{por}\big(\setof{\tru},\setof{\,}\big)
	=\textit{por}\big(\setof{\,},\setof{\tru}\big)
	=\setof{\tru}$. 

	Parallel-or is not realizable as a stable functor at the categorical level in
	the strong sense that there is no stable functor
	$F:\Set^{\mathrm{Bool}}\times\Set^{\mathrm{Bool}}\to\Set^{\mathrm{Bool}}$ 
	such that $F\big(0,0)=0$ and 
	$F(\yon(\tru),0\big) = F(0,\yon(\tru)) = F(\yon(\tru),\yon(\tru)) 
	= \yon(\tru)$. 
	Indeed, such functors do not preserve the pullback
	\[
	\begin{tikzcd}[row sep = 0em, column sep = 2em]
		& \big(\yon(\tru),\yon(\tru)\big) & \\
		\big(\yon(\tru),0\big) \arrow{ur}{} & & \big(0,\yon(\tru)\big) \arrow{ul}{} \\ 
		& \big(0,0\big)\arrow{ur}{} \arrow{ul}{} & 
	\end{tikzcd}
	\]
	and thus induce local functors
	$F/(\yon(\tru),\yon(\tru))
	: \Set^{\mathrm{Bool}}\times\Set^{\mathrm{Bool}}/(\yon(\tru),\yon(\tru))
	\to
	\Set^{\mathrm{Bool}}/\yon(\tru)$
	that fail to be right adjoints.
	
	However, the generalization from domains to categories allows for an
	intensional quantitative interpretation of parallel or.  Indeed, for 
	$K:\Set\to S$ the \emph{collapse} functor mapping a set to $0$ if it is
	empty and to $1$ otherwise, we have a stable functor 
	\[ 
	P(X,Y) = (X_\fal\times Y_\fal,X_\tru + Y_\tru)
	\]
	lifting 
	$\textit{por}$ as follows: 
	\[\begin{tikzcd} 
		\ar[d, "K^{\mathrm{Bool}}\times K^{\mathrm{Bool}}", swap]
		\Set^{\mathrm{Bool}} \times \Set^{\mathrm{Bool}} \ar[r, "P"] & 
		\Set^{\mathrm{Bool}} \ar[d, "K^\mathrm{Bool}"]
		\\
		S^{\mathrm{Bool}} \times S^{\mathrm{Bool}} \ar[r, "\textit{por}", swap] &
		S^{\mathrm{Bool}}
	\end{tikzcd}\]
\end{exa}

For kits $\kitstr{\gpd1}$ and $\kitstr{\gpd2}$ and a local right adjoint functor $T : \StPSh(\kitstr{\gpd1}) \to \StPSh(\kitstr{\gpd2})$, one may unfold for each $X \in \StPSh(\kitstr{\gpd1})$ the right adjoint property for $T_X$: for every $h : Y \to TX$, there exist $f : X_0 \to X$ and $g: Y \to TX_0$ such that $h$ factorizes as follows:
\begin{center}
	\begin{tikzpicture}[scale=0.9]
		\node (A) at (0, 0) {$Y$};
		\node (B) at (2, 0) {$T(X_0)$};
		\node (C) at (4, 0) {$T(X)$};
		
		\draw [->] (A) -- node [below] {$g$} (B);
		\draw [->] (B) to node [below] {$T(f)$} (C);
		\draw [->](A) to [bend left =30]  node [above] {$h$} (C);
	\end{tikzpicture}
\end{center}

Note the analogy with Definition~\ref{def:stabilitydomains}.  
The minimality requirement for $X_0$ is represented here by a universal property known as \emph{genericity}, which $g$ must additionally satisfy. The idea is to exhibit $T(f) \circ g$ as initial in the category of factorizations of $h$ of the form $T(-) \circ -$. These correspond to the normal forms studied by Girard and Hasegawa \cite{GirardNormal, HasegawaAnalytic}.
\begin{defi}
	\label{def:generic}
	Suppose that $T : \gpd1 \rightarrow \gpd2$ is a functor between categories $\gpd1$ and $\gpd2$. A morphism $g : b \rightarrow T(a)$ in $\gpd2$ is said to be \emph{generic} if for
	every commuting square as on the left below, there exists a unique morphism $k : a \rightarrow a'$ in $\gpd1$ making the two triangles on the right commute. 
	\[
	\begin{tikzcd}[column sep=0em, row sep=1em]
		& T(a'') & \\
		T(a) \arrow{ur}{T(h)} & & \arrow[swap]{ul}{T(f)} T(a') \\ 
		& b \arrow{ul}{g} \arrow[swap]{ur}{v} &
	\end{tikzcd}\qquad
	\begin{tikzcd}[column sep=0em, row sep=1em]
		& a'' & \\
		a \arrow{ur}{h} \arrow[dashed]{rr}{k} & & \arrow[swap]{ul}{f} a' \\[-7pt]
		T(a) \arrow{rr}{T(k)} &[+8pt] & T(a') \\ 
		& b \arrow{ul}{g} \arrow[swap]{ur}{v} &
	\end{tikzcd}
	\]
\end{defi}
We can then derive a standard result:
\begin{thm}
	A functor $T : \gpd1 \rightarrow \gpd2$ is a local right adjoint if and only if it \defn{admits generic factorizations}, in the sense that every morphism $h : b \to T(a)$ has a factorisation $h = T(f) \circ g$ with $g$ generic. 
\end{thm}

Generic morphisms (elsewhere known as \emph{strict generic}~\cite{Weber2004generic} or \emph{candidates}~\cite{TaylorGroupoidsLinearLogic}) 
provide the elements in the construction of a stable species from a stable functor.  What we obtain is analogous to the \emph{trace} of a stable function in Berry's work~\cite{BerryStable} and closely related to Taylor's trace \cite{tracefactorization}. It is appropriate here to use the same terminology.
In the case of combinatorial or generalized species, the map $k$ in Definition \ref{def:generic} is not required to be unique and one works with the corresponding notion of \emph{generic elements} introduced by Joyal \cite{JoyalAnalytic} and generalized by Fiore \cite{Fioreanalytic} in order to construct a species from an analytic functor.

To show that stable species induce stable functors via the mapping $P \mapsto  \rsLan{P}$, we reduce the proof of admitting generic factorizations to only considering generic factorizations relative to representables:

\newcommand{\cat}[1]{\mathbb{#1}}
\begin{lem}\label{lem:ColimGeneric}
	Let $T:\gpd1\to\gpd2$ be a functor.  
	\begin{enumerate}[beginpenalty=99] 
		\item For diagrams $A: \cat{I}\to\gpd1$ and $B: \cat{I}\to\gpd2$, and a natural transformation $g: B \Rightarrow TA: \cat{I}\to \gpd2$, if $g_i: B_i\to T(A_i)$ is generic for all $i\in\cat{I}$ then so is the induced composite 
		\[\begin{tikzcd}[column sep = large]
			k 
			=
			\colim (B) \arrow[r,"\colim(g)"] & \colim (TA)
			\arrow[r,"{[T\colimin_i]_{i\in\cat{I}}}"] & T(\colim A).
		\end{tikzcd}\]
		\item For $b\in\gpd2$, a diagram $A: \cat{I}\to\gpd1$, and a cone $g: b \Rightarrow TA: \cat{I}\to \gpd2$, if $g_i: b \to T(A_i)$ is generic for all $i\in\cat{I}$ then so are the composites
		\[
		b \xlongto{g_i} T(A_i)\xto{T(\colimin_i)} T(\colim(A))
		\]
		for every $i\in\cat{I}$.
	\end{enumerate}
\end{lem}

\begin{proof}
	\begin{enumerate}
		\item[]
		\item Let $(\colimin_i: B_i\to \colim(B))_{i\in\cat{I}}$ be a colimiting cocone. Consider $l : \colim(B) \rightarrow T(d)$ in $\gpd2$ and $h : \colim(A) \rightarrow c$, $f: d\rightarrow c$ in $\gpd1$ such that $T(h) \circ k = T(f) \circ l$; equivalently, $T(h) \circ T(\colimin_i) \circ g_i = T(f) \circ l \circ \colimin_i$ for all $i\in \cat{I}$. Since each $g_i$ is generic, there exists, for every $i \in \cat{I}$, a unique morphism $e_i : A_i \rightarrow d$ such that $h \circ \colimin_i = f \circ e_i$ and $T(e_i) \circ g_i = l \circ \colimin_i$.  
		\begin{center}
			\begin{tikzpicture}[thick]
				\node (A) at (0,0) {$T(\colim A)$};
				\node (B) at (3,0) {$T(c)$};
				\node (C) at (0,2.5) {$\colim(B)$};
				\node (D) at (3,2.5) {$T(d)$};
				\node (F) at (-3,2.5) {$B_i$};
				\node (G) at (-3,0) {$T(A_i)$};
				
				\draw [->] (C) -- node [left,pos=0.25] {$k$} (A);
				\draw [->] (A) -- node [below] {$T(h)$} (B);
				\draw [->] (C) -- node [above] {$l$} (D);
				\draw [->] (D) -- node [right] {$T(f)$} (B);
				\draw [->] (F) -- node [left] {$g_i$} (G);
				\draw [->] (F) -- node [above] {$\colimin_i$} (C);
				\draw [->] (G) -- node [below] {$T(\colimin_i)$} (A);
				\draw [->,dotted,/tikz/commutative diagrams/crossing over] (G) -- node [above,pos=0.2,xshift=-3pt] {$T(e_i)$} (D);
				\draw [->,dotted] (A) -- node [below right] {$T(e)$} (D);
			\end{tikzpicture}
		\end{center}
		One can show that $(e_i : A_i \rightarrow d)_{i \in\cat{I}}$ is a cocone of $A$.  So there exists a unique morphism $e:\colim(A) \to d$ such that $e\circ\colimin_i=e_i$; from which it follows that $f\circ e=h$. Moreover, $T(e)\circ k = l$ because $T(e)\circ k\circ\colimin_i = l\circ\colimin_i$ for all $i\in\cat{I}$.  This establishes existence. As for uniqueness, if $e':\colim(A)\to d$ has the required property of $k$ with respect to $h$, $f$, $l$ then each $e'\circ\colimin_i$ has the generic property of $g_i$ with respect to $h\circ\colimin_i$, $f$, $l\circ\colimin_i$ making $e'\circ\colimin_i=e_i$ and therefore $e'=e$.
		
		\item Let $i \in \cat{I}$ and assume that there exists $l :b \rightarrow T(d)$, $k : \colim(A) \rightarrow c$ and $f: d\rightarrow c$ such that $T(k) \circ T(\colimin_i) \circ g_i = T(f) \circ l$. Since $g_i$ is generic, there exists a unique $e_i : A_i \rightarrow d$ such that $ k \circ \colimin_i = f \circ e_i$ and $T(e_i) \circ g_i = l$. 
		\begin{center}
			\begin{tikzpicture}[thick]
				\node (A) at (0,0) {$T(\colim(A))$};
				\node (B) at (3,0) {$T(c)$};
				\node (C) at (0,3) {$b$};
				\node (D) at (3,3) {$T(d)$};
				\node (E) at (0,1.5) {$T(A_i)$};
				
				\draw [->] (C) -- node [left] {$g_i$} (E);
				\draw [->] (E) -- node [left] {$T(\colimin_i)$} (A);
				\draw [->] (A) -- node [below] {$T(k)$} (B);
				\draw [->] (C) -- node [above] {$l$} (D);
				\draw [->] (D) -- node [right] {$T(f)$} (B);
				\draw [->,dotted] (E) -- node [below right] {$T(e_i)$} (D);
			\end{tikzpicture}
		\end{center}
		
		One can check that $(e_i : A_i \rightarrow d)_{i \in \cat{I}}$ is a cocone of $A$. Hence, there exists a unique morphism $e : \colim(A) \rightarrow d$ such that $e \circ \colimin_i = e_i$ for all $i \in \cat{I}$. Since $(f \circ e_i : a_i \rightarrow c)_{i \in \cat{I}}$ is a cocone of $A$, there exists a unique $q : \colim(A) \rightarrow c$ such that $q \circ \colimin_i = f \circ e_i$ for all $i \in \cat{I}$.
		Both $k$ and $f \circ e$ verify this condition which implies that $k= f\circ e$. We also have that $T(e) \circ T(\colimin_i) \circ g_i = T(e_i) \circ g_i = l$.
		To show that $T(\colimin_i) \circ g_i$ is generic, it remains to show that $e : \colim(A) \rightarrow d$ is the unique morphism that makes the two diagrams below commute.  
		\begin{center}
			\begin{tikzpicture}[scale =0.6]
				\node (A) at (0,0) {$T(\colim(A))$};
				\node (A1) at (5,0) {$\colim(A)$};
				\node (B) at (7.5,0) {$c$};
				\node (C) at (0,2.5) {$b$};
				\node (D) at (3,2.5) {$T(d)$};
				\node (D1) at (7.5,2.5) {$d$};
				
				\draw [->] (C) -- node [left] {$T(\colimin_i)\circ g_i$} (A);
				\draw [->] (A1) -- node [below] {$k$} (B);
				\draw [->] (C) -- node [above] {$l$} (D);
				\draw [->] (D1) -- node [right] {$f$} (B);
				\draw [->] (A) -- node [below right] {$T(e)$} (D);
				\draw [->,dotted] (A1) -- node [above left] {$e$} (D1);
			\end{tikzpicture}
		\end{center}
		Assume that there exists $e': \colim(A) \rightarrow d$ such that $k= f\circ e'$ and $T(e') \circ T(\colimin_i) \circ g_i = l$. We then have that $e' \circ \colimin_i= e_i$ for all $i \in \cat{I}$ by genericity of the $g_i$'s which implies that $e' = e$ by the universal property of the colimit. \qedhere
	\end{enumerate}
\end{proof}

\begin{prop}
	\label{prop:ColimGeneric}
	Let $T : \gpd1 \rightarrow \gpd2$ be a functor.  For a diagram 
	$B: \cat{I}\to\gpd2$, if $T$ admits generic factorizations relative to $b_i$
	for all $i\in \cat{I}$ then it admits generic factorizations relative to
	$\colim(B)$ whenever this exists. 
\end{prop}
\begin{proof}
	Assume that $\colim(B)$ exists in $\gpd2$ and let $f : \colim(B) \rightarrow T(a)$ be a morphism in $\gpd2$. By hypothesis, each $f_i := f \circ \colimin_i : b_i \rightarrow T(a)$ can be factored as:
	\begin{center} 
		\begin{tikzpicture}
			\node (A) at (0, 0) {$b_i$};
			\node (B) at (2, 0) {$T(c_i)$};
			\node (C) at (4, 0) {$T(a)$};
			
			\draw [->] (A) -- node [above] {$g_i$} (B);
			\draw [->] (B) to node [above] {$T(h_i)$}  (C);
		\end{tikzpicture}
	\end{center}
	where $g_i$ is generic. Define the diagram $C: \cat{I} \rightarrow \gpd1$ as follows: $C(i) := c_i$ for $i \in \cat{I}$ and for a morphism $e : i \rightarrow j \in \cat{I}$, $C(e)$ is the unique morphism $c_i \rightarrow c_j$ obtained from the genericity of $g_i$ in the square below: 
	\begin{center}
		\begin{tikzpicture}[scale =0.8]
			\node (A) at (0,0) {$T(c_i)$};
			\node (B) at (3,0) {$T(a)$};
			\node (C) at (0,2.5) {$b_i$};
			\node (D) at (3,2.5) {$T(c_j)$};
			
			\draw [->] (C) -- node [left] {$g_i$} (A);
			\draw [->] (A) -- node [below] {$T(h_i)$} (B);
			\draw [->] (C) -- node [above] {$g_j \circ B(e)$} (D);
			\draw [->] (D) -- node [right] {$T(h_j)$} (B);
		\end{tikzpicture}
	\end{center}
	By Lemma \ref{lem:ColimGeneric}$.1$, $[T(\colimin_i)]_{i\in \cat{I}} \circ \colim(g) : \colim(B) \to \colim(C)$ is generic. Let $[h] : \colim(C) \rightarrow a$ be the mediating morphism induced by the cocone $(h_i : c_i \rightarrow a)_{i \in \cat{I}}$. Both $f$ and $T([h]) \circ [T(\colimin_i)]_{i\in \cat{I}} \circ \colim(g)$ verify the conditions of the mediating morphism $\colim(B) \rightarrow T(a)$ induced by the cocone $(T(h_i) \circ T(\colimin_i) \circ g_i : b_i \rightarrow T(a))_{i \in \cat{I}}$ which implies that $f = T([h]) \circ [T(\colimin_i)]_{i\in \cat{I}} \circ \colim(g)$ as desired.
\end{proof}

Recall that a generalized species $P : \Sym \gpd1 \profto \gpd2$ induces an analytic functor $\Lan_{s} P: \PSh{\gpd1} \to \PSh{\gpd2}$ mapping a presheaf $X \in \PSh \gpd1$ and an object $b\in \gpd2$ to $(\Lan_{s} P)(X)(b)$ given by the following coend formula:
\[
\coend^{u=\seq{a_1, \dots, a_n}\in \Sym \gpd1} P(b,u) \times \PSh\gpd1(s(u), X) \cong \coend^{u=\seq{a_1, \dots, a_n}\in \Sym \gpd1} P(b,u) \times \prod_{i=1}^nX(a_i)
\]
which explicitly consists of equivalences classes of triples $u \in \Sym \gpd1$, $p\in P(b,u)$, $\bar{x} \in \PSh\gpd1(s(u), X)\cong X^\Sym(u)$ that we denote by $p\coendtensor[u]\overline{x}$. Thus, $p \coendtensor[u] \bar{x}=q \coendtensor[v] \bar{y}$ if and only if there exists $\alpha: u \to v$ in $\Sym \gpd1$ such that $\bar{y} \cdot \alpha = \bar{x}$ and $\alpha \cdot p = q$.

The Boolean kit $(\Sym \gpd1,\oc  \kit1)$ in Definition \ref{def:expBooleanKit} has an extensional characterization in terms of stabilized presheaf which we will use in the proofs:

\begin{lem}\label{lem:bangKitsPresheaves}
	For a kit $\kitstr{\gpd1}= (\gpd1, \kit1)$, the following equality holds for all $u \in \Sym \gpd1$:
	\[
	\oc \kit1(u) = \{ G \subgroup \Endo(u) \mid \exists X \in \StPSh(\kitstr{\gpd1}), G \in \Stab(X^\Sym(u)) \}^\dorth.
	\]
\end{lem}

\begin{lem}\label{lem:FreenessStable}
	For Boolean kits $\kitstr{\gpd1}, \kitstr{\gpd2}$, a presheaf $X\in\StPSh(\kitstr{\gpd1})$ and a stable species $P:\oc\kitstr{\gpd1} \profto \kitstr{\gpd2}$, if two elements $p \coendtensor[u] \bar{x}$ and $ q \coendtensor[v] \bar{y}$ in $(\Lan_{s} P)(X)(b)$ are equal for some $b \in \gpd2$, then there exists a unique $\alpha : u \to v$ in $\Sym \gpd1$ such that $\bar{y} \cdot \alpha = \bar{x}$ and $\alpha \cdot p = q$. 
\end{lem}

\begin{proof}
	Assume that for $p \in P(b,u)$, $q\in P(b,v)$, $\bar{x} \in X^\Sym(u)$ and $\bar{y}\in X^\Sym(v)$, we have the equality $p \coendtensor[u] \bar{x}= q \coendtensor[v] \bar{y}$. The existence of $\alpha:u\to v$ in $\Sym \gpd1$ such that $\bar{y} \cdot \alpha = \bar{x}$ and $\alpha \cdot p = q$ follows from the pointwise definition of the coend. For uniqueness, assume that there exists $\beta : u \to v$ such that $\bar{y} \cdot \beta = \bar{x}$ and $\beta \cdot p = q$ as well.
	The equality $\alpha \cdot p =\beta\cdot p$ implies that $\beta^{-1} \alpha \cdot p = p$ and since $\id[b]$ is in $\bigcup \kit2^\orth(b)$, we have $\beta^{-1} \alpha \in \bigcup (\oc \kit1)^\orth(u)$ as $P$ is a stable species.
	On the other hand, $\bar{y} \cdot \alpha = \bar{y} \cdot \beta$ implies that $\beta^{-1}\alpha$ is in $\Stab(X^\Sym(u))$ which entails that $\beta^{-1}\alpha\in \bigcup (\oc \kit1)(u)$ by Lemma \ref{lem:bangKitsPresheaves}. Hence, $\beta^{-1}\alpha =\id$ as desired.
\end{proof}

\begin{prop}\label{prop:rsLanStable}
	For a stable species $P: \oc\kitstr{\gpd1} \profto \kitstr{\gpd2}$, the
	restricted left Kan extension 
	$\rsLan{P}=(\Lan_{s} P)\incl[\kit1]: \StPSh(\kitstr{\gpd1}) \to \PSh{\gpd2}$ 
	factors through the inclusion
	$\StPSh(\kitstr{\gpd2} ) \hookrightarrow \PSh(\gpd2)$.  Furthermore, the
	resulting functor $\rsLan P:\StPSh(\kitstr{\gpd1})\to \StPSh(\kitstr{\gpd2})$
	is stable.
\end{prop}

\begin{proof}
	To prove that $(\Lan_{s} P)\incl[\kit1]$ restricts to a functor $\StPSh(\kitstr{\gpd1}) \to \StPSh(\kitstr{\gpd2})$, we show that for all $X \in \StPSh(\kitstr{\gpd1})$ and $b \in \gpd2$, the stabilizer of every element in $\Lan_{s} P  X(b) = \int^{u\in \Sym \gpd1} P(b,u) \times X^\Sym(u)$ is in $\kit2(b)$. Let $t = p \coendtensor[u] \bar{x}$ in $\Lan_{s} P  X(b)$. For every $\beta: b\to b$ in $\Stab(t)$, we have $(p \cdot \beta) \coendtensor[u] \bar{x} = p \coendtensor[u] \bar{x}$, i.e. there exists $\alpha : u \to u$ in $\Sym \gpd1$ such that $\alpha \cdot p \cdot \beta = p$ and $\bar{x} \cdot \alpha = \bar{x}$. Since $X$ is in $\StPSh(\kitstr{\gpd1})$ and $\alpha$ is in $\Stab(X^\Sym(u))$, we must have $\alpha \in \bigcup \oc \kit1(u)$ and since $P$ is a stable species, it implies that $\beta \in \bigcup \kit2(b)$ as desired. It remains to show that the functor $\rsLan{P} : \StPSh(\kitstr{\gpd1}) \to \StPSh(\kitstr{\gpd2})$ is stable.
	
	The functor $\Lan_{s} P:\PSh \gpd1 \to \PSh \gpd2$ is analytic and therefore is finitary and preserves quasi-pullbacks (and hence epimorphisms) \cite{Fioreanalytic}.
	Since the embeddings $\incl[\kit1] : \StPSh\kitstr{\gpd1} \to \PSh\gpd1$ and $\incl[\kit2]: \StPSh\kitstr{\gpd2} \to \PSh\gpd2$ create epimorphisms and filtered colimits, $\rsLan{P}$ is epi-preserving and finitary.
	
	It remains to show that $\rsLan{P}$ admits generic factorizations. By Proposition \ref{prop:ColimGeneric}, it suffices to show that $\rsLan{P}$ admits generic factorizations relative to representables. Consider a morphism $\yon b \to \rsLan{P}(X)$ in $\StPSh(\kitstr{\gpd2})$, it is of the form $p \coendtensor[u]\bar{x}$ with $p \in P(b,u)$ and $\bar{x}: su \to X $ for some $u\in \Sym \gpd1$. We show that
	\begin{center}
		\begin{tikzpicture}[thick]
			\node (A) at (0, 0) {$ \yon b$};
			\node (B) at (2.5, 0) {$\rsLan{P} (su)$};
			\node (C) at (5, 0) {$\rsLan{P}(X)$};
			
			\draw [->] (A) -- node [above] {$p \coendtensor[u] \id$} (B);
			\draw [->] (B) to node [above] {$\rsLan{P}(\bar{x})$}  (C);
		\end{tikzpicture}
	\end{center}
	is a generic factorization for $p \coendtensor[u]\bar{x}$. The equality $p \coendtensor[u]\bar{x} = \rsLan{P}(\bar{x})(p \coendtensor[u] \id)$ is immediate, it remains to prove that $p \coendtensor[a] \id$ is generic.
	
	Assume that there exist $Y,Z \in \StPSh(\kitstr{\gpd1})$ and morphisms $p' \coendtensor[v] \bar{y} : \yon b \to \rsLan{P}(Y)$ and $f: Y \to Z$ and $\bar{z} : su \to Z$ such that the following diagram commutes in $\StPSh(\kitstr{\gpd2})$:

	\begin{center}
		\begin{tikzpicture}[scale=1, thick]
			
			\node (A) at (0, 2) {$\yon b$};
			\node (B) at (3,2) {$\rsLan{P} (Y)$};
			\node (C) at (0,0) {$\rsLan{P} (su)$};
			\node (D) at (3,0) {$\rsLan{P} (Z)$};
			
			\draw [->] (A) to node [above] {$p' \coendtensor[v] \bar{y}$} (B);
			\draw [->] (A) to node [left] {$p \coendtensor[u] \id$} (C);
			\draw [->] (B) to node [right] {$\rsLan{P}(f)$} (D);
			\draw [->] (C) to node [below] {$\rsLan{P}(\bar{z})$} (D);
			
		\end{tikzpicture}
	\end{center}
	
	The equality $\rsLan{P}(f) (p' \coendtensor[v] \bar{y}) = p \coendtensor[u] \bar{z}$ is equivalent to $p' \coendtensor[v] f \bar{y} = p \coendtensor[u] \bar{z}$ which implies by Lemma \ref{lem:FreenessStable} that there exists a unique $\alpha : u \to v$ in $\Sym \gpd1$ such that $(f \bar{y}) \cdot \alpha = \bar{z}$ and $\alpha \cdot p = p'$. We then have $\rsLan{P}(\bar{y} \cdot \alpha) (p \coendtensor[u] \id) = p \coendtensor[u] \bar{y} \cdot \alpha = \alpha \cdot p \coendtensor[v]\bar{y} = p' \coendtensor[v] \bar{y}$. It remains to show uniqueness, assume that there exists $\bar{t} : su \to Y$ such that $f \bar{t} = \bar{z}$ and $\rsLan{P}(\bar{t})(p \coendtensor[u] \id) = p' \coendtensor[v] \bar{y}$ i.e. $p \coendtensor[u] \bar{t} = p' \coendtensor[v] \bar{y}$. 
	\begin{center}
		\begin{tikzpicture}[scale=1, thick]
			
			\node (A) at (0, 2) {$\yon b$};
			\node (B) at (3,2) {$\rsLan{P} (Y)$};
			\node (C) at (0,0) {$\rsLan{P} (su)$};
			\node (D) at (3,0) {$\rsLan{P} (Z)$};
			
			\draw [->] (A) to node [above] {$p' \coendtensor[v] \bar{y}$} (B);
			\draw [->] (A) to node [left] {$p \coendtensor[u] \id$} (C);
			\draw [->] (B) to node [right] {$\rsLan{P}(f)$} (D);
			\draw [->] (C) to node [below] {$\rsLan{P}(\bar{z})$} (D);
			\draw [->, dotted] (C) to node [fill=white] {$\rsLan{P}(\bar{t})$} (B);
		\end{tikzpicture}
	\end{center}
	It implies that there exists $\beta: u \to v$ such that $\bar{y} \cdot \beta= \bar{t}$ and $\beta \cdot p = p'$. Hence, $(f \bar{y}) \cdot \alpha = f\bar{t} = \bar{z}$ which implies that $\alpha=\beta$ so that $\bar{t} = \bar{y} \cdot \alpha$ as desired.
\end{proof}

Analogously to the trace operator for stable functions, we now define a trace functor providing an inverse to the operation $P \mapsto \widetilde{P}$ by using the generic factorization property.

\begin{defi}
Consider Boolean kits $\kitstr{\gpd1}$ and $\kitstr{\gpd2}$ and a stable functor $T : \StPSh(\kitstr{\gpd1}) \to \StPSh(\kitstr{\gpd2})$. The \defn{trace} of $T$ is the profunctor $\trace(T) : \gpd2^\op \times \Sym \gpd1 \to \Set$ with object mapping   
\[
\trace(T)(b, \seq{a_i}_{i \in [n]}) = \left\{ g 
 : \yon b \to T\left(\textstyle \coprod_{i} \yon a_i\right) \mid 
 g \text{ is generic} \right\}
\]
and functorial action given by composition.
\end{defi}

The trace of a stable functor is a stable species of structures: 
\begin{prop}\label{prop:StableQProf}
For Boolean kits $\kitstr{\gpd1}$ and $\kitstr{\gpd2}$ and a stable functor $T : \StPSh(\kitstr{\gpd1}) \to \StPSh(\kitstr{\gpd2})$, $\trace(T) \in \SEsp(\kitstr{\gpd1}, \kitstr{\gpd2})$.  
 \end{prop}

\begin{proof}
	We first need to show that $\trace(T)$ is indeed a species, i.e. for $t : b \to T(s(u))$ generic and morphisms $\alpha \in \Sym\gpd1(u,u)$, $\beta \in \gpd2(b,b)$, $ \alpha\cdot t \cdot \beta = T(s(\alpha)) \circ t \circ \beta$ is also generic which holds since generic elements are stable under precomposition or postcomposition by isomorphisms. We now prove that $\trace(T)$ is a stable species: assume that $\alpha \cdot t  \cdot \beta =t$ i.e. the following square commutes in $\StPSh(\kitstr{\gpd2})$:
	\begin{center}
		\begin{tikzpicture}[scale=0.75]
			\node (A) at (0, 2.3) {$b$};
			\node (B) at (3,2.3) {$T(s(u)) $};
			\node (C) at (0,0) {$b$};
			\node (D) at (3,0) {$T(s(u)) $};
			
			\draw [->] (A) to node [above] {$t$} (B);
			\draw [->] (C) to node [left] {$\beta$} (A);
			\draw [->] (B) to node [right] {$T(s\alpha)$} (D);
			\draw [->] (C) to node [below] {$t$} (D);
		\end{tikzpicture}
	\end{center}
	\begin{itemize}
		\item Assume that $\alpha\in \bigcup\oc \kit1(u)$, using Lemmas \ref{lem:characterisationDoubleOrth} and \ref{lem:bangKitsPresheaves}, it implies that for all $n \in \N$, $\alpha^n =\id$ or there exists $m \in \N$ and $\id \neq \alpha^{nm} \in \Stab(X^\Sym(u))$ for some $X \in \StPSh(\kitstr{\gpd1})$.
		We want to show that $\beta$ is in $\bigcup \kit2(b) =\bigcup \kit2^\dorth(b)$ i.e. for all $n \in \N$, $\beta^n = \id$ or there exists $m \in \N$ such that $\id \neq \beta^{nm} \in\bigcup \kit2(b)$.  Equivalently, we show that for all $n \in \N$, $\beta^n = \id$ or there exist $m \in \N$ and $Y \in \StPSh(\kitstr{\gpd2})$ with $\id \neq \beta^{nm} \in \Stab(Y(b))$.
		
		Let $n \in \N$ and assume that $\beta^n \neq \id$,  if $\alpha^n = \id$, then $t \cdot \beta^n = t$ which implies that $\beta^n$ is in $\Stab(T(s(u))$. Since $T(s(u))$ is in $\StPSh(\kitstr{\gpd2})$, we can take $m:=1$ and obtain the desired result. If there exist $m \in \N$ and $X \in \StPSh(\kitstr{\gpd1})$ such that $\id \neq \alpha^{nm} \in \Stab(X^\Sym(u))$, then there exists $\bar{x} : s(u) \to X$ in $\StPSh(\kitstr{\gpd1})$ (or equivalently $\bar{x} \in X^\Sym(u)$) such that $\bar{x} \cdot \alpha^{nm} = \bar{x}$ which implies that the following diagram commutes in $\StPSh(\kitstr{\gpd2})$:
		\begin{center}
			\begin{tikzpicture}[scale=0.8]
				\node (A) at (0, 2.3) {$b$};
				\node (B) at (3,2.3) {$T(s(u)) $};
				\node (C) at (0,0) {$b$};
				\node (D) at (3,0) {$T(s(u)) $};
				\node (E) at (7,1.1) {$T(X)$};
				
				\draw [->] (A) to node [above] {$t$} (B);
				\draw [->] (C) to node [left] {$\beta^{nm}$} (A);
				\draw [->] (B) to node [right] {$T(s\alpha^{nm})$} (D);
				\draw [->] (C) to node [below] {$t$} (D);
				\draw [->] (B) to node [above] {$T(\bar{x})$} (E);
				\draw [->] (D) to node [below] {$T(\bar{x})$} (E);
			\end{tikzpicture}
		\end{center}
		Hence, $\beta^{nm}$ is in $\Stab(T(X)(b))$ and since $T(X)$ is in $\StPSh(\kitstr{\gpd2})$, it only remains to show that $\beta^{nm}\neq \id$. Assume that $\beta^{nm} =\id$, then $\alpha^{nm} \cdot t =t$, so the following diagram commutes in $\StPSh(\kitstr{\gpd2})$:
		\begin{center}
			\begin{tikzpicture}[scale =0.8]
				\node (A) at (0,0) {$T(s(u)) $};
				\node (B) at (4,0) {$T(X)$};
				\node (C) at (0,2.5) {$b$};
				\node (D) at (4,2.5) {$T(s(u)) $};
				
				\draw [->] (C) -- node [left] {$t$} (A);
				\draw [->] (A) -- node [below] {$T(\bar{x})$} (B);
				\draw [->] (C) -- node [above] {$t$} (D);
				\draw [->] (D) -- node [right] {$T(\bar{x})$} (B);
				\draw [->] (A) -- node [fill=white] {$T(s\alpha^{nm})$} (D);
			\end{tikzpicture}
		\end{center}
		since $t$ is generic, it implies that $\alpha^{nm}=\id$ so $\beta^{nm}\neq \id$ as desired.
		\item 	Assume now that $\beta \in \bigcup \kit2^\orth(b)$, we want to show that $\alpha \in\bigcup (\oc \kit1)^\orth(u)$ i.e. for all $n$, if there exists $X' \in \StPSh(\kitstr{\gpd1})$ such that $\alpha^n$ is in $\Stab(X'^\Sym(u))$, then $\alpha^n = \id$. Assume that there exists $X' \in \StPSh(\kitstr{\gpd1})$ such that $\alpha^n$ is in $\Stab(X'^\Sym(u))$, i.e. there exists $\bar{x}' : s(u) \to X'$ in $\StPSh(\kitstr{\gpd1})$ such that $\bar{x}' \cdot \alpha^n = \bar{x}'$ which implies that the following diagram commutes in $\StPSh(\kitstr{\gpd2})$:
		\begin{center}
			\begin{tikzpicture}[scale=0.8]
				\node (A) at (0, 2.2) {$b$};
				\node (B) at (3,2.2) {$T(s(u) )$};
				\node (C) at (0,0) {$b$};
				\node (D) at (3,0) {$T(s(u) )$};
				\node (E) at (7,1.1) {$T( X')$};
				
				\draw [->] (A) to node [above] {$t$} (B);
				\draw [->] (C) to node [left] {$\beta^n$} (A);
				\draw [->] (B) to node [left] {$T(s\alpha^{n})$} (D);
				\draw [->] (C) to node [below] {$t$} (D);
				\draw [->] (B) to node [above] {$T(\bar{x}')$} (E);
				\draw [->] (D) to node [below] {$T(\bar{x}')$} (E);
			\end{tikzpicture}
		\end{center}
		Hence, $\beta^n \in \Stab(T(X')(b))$ which implies that $\beta^n =\id$ so that the following diagram commutes:
		\begin{center}
			\begin{tikzpicture}[scale =0.8]
				\node (A) at (0,0) {$T(s(u) )$};
				\node (B) at (4,0) {$T(X')$};
				\node (C) at (0,2.5) {$b$};
				\node (D) at (4,2.5) {$T(s(u) )$};
				
				\draw [->] (C) -- node [left] {$t$} (A);
				\draw [->] (A) -- node [below] {$T(\bar{x}')$} (B);
				\draw [->] (C) -- node [above] {$t$} (D);
				\draw [->] (D) -- node [right] {$T(\bar{x}')$} (B);
				\draw [->] (A) -- node [fill=white]  {$T( s(\alpha^n ))$} (D);
			\end{tikzpicture}
		\end{center}
		Since $t$ is generic, we obtain that $\alpha^n = \id$ as desired.\qedhere
	\end{itemize}
\end{proof}

We proceed to show that for a stable species $P \in \SEsp(\kitstr{\gpd1}, \kitstr{\gpd2})$, $\trace(\widetilde{P}) \cong P$, and for a stable functor $T : \StPSh(\kitstr{\gpd1}) \to \StPSh(\kitstr{\gpd2})$, $\widetilde{\trace(T)} \cong T$.

\begin{lem}\label{lem:StableUnitIso}
	For a stable species $P: \oc \kitstr{\gpd1} \profto\kitstr{\gpd2}$, the transformation $\eta : P \Rightarrow \trace(\rsLan{P})$ whose components $\eta_{b,u}: P(b,u) \to \trace(\rsLan{P})(b,u)$ are given by
	\[
	p  \mapsto (p \coendtensor[u] \id) : b \to \rsLan{P}(s(u))
	\]
	is a natural isomorphism.
\end{lem}

\begin{proof}
	
	We first need to show that $(p \coendtensor[u] \id) : b \to \rsLan{P}(s(u))$ is generic for the map $\eta_{b,u}$ to be well-defined. Assume that the diagram below commutes in $\StPSh(\kitstr{\gpd2})$ for some morphisms $\bar{x}:s(u) \to X$, $f : Y \to X$ and $h : b \to \rsLan{P}(Y)$:
	\begin{center}
		\begin{tikzpicture}[scale=0.75]
			
			\node (A) at (0, 2) {$b$};
			\node (B) at (3,2) {$\rsLan{P}(Y)$};
			\node (C) at (0,0) {$\rsLan{P}(s(u))$};
			\node (D) at (3,0) {$\rsLan{P}(X)$};
			
			\draw [->] (A) to node [above] {$h$} (B);
			\draw [->] (A) to node [left] {$(p \coendtensor[u] \id)$} (C);
			\draw [->] (B) to node [right] {$\rsLan{P}(f)$} (D);
			\draw [->] (C) to node [below] {$\rsLan{P}(\bar{x})$} (D);
		\end{tikzpicture}
	\end{center}
	
	The morphism $h : b \to \rsLan{P}(Y)$ corresponds to an element $(p' \coendtensor[u'] \bar{y})$ where $p'\in P(b,u')$ and $\bar{y} : s(u') \to Y$. Since the diagram commutes, we have $p \coendtensor[u] \bar{x} = p' \coendtensor[u'] f \bar{y}$ i.e. there exists $\alpha : u \to u' \in \Sym \gpd1$ such that $\alpha \cdot p = p'$ and $f \bar{y} s(\alpha)= \bar{x}$. Define $k : s(u) \to Y$ to be $\bar{y} s(\alpha)$. We then have $f k =\bar{x}$ and $\rsLan{P}(k)(p \coendtensor[u] \id) = p \coendtensor[u] \bar{y} s(\alpha) = \alpha \cdot p \coendtensor[u'] \bar{y} =p' \coendtensor[u'] \bar{y}$ as desired.  
	
	For uniqueness, assume that there exists $k' :  s(u) \to Y$ such that $f k' = \bar{x}$ and $\rsLan{P}(k')(p \coendtensor[u] \id) = p \coendtensor[u]k' =p' \coendtensor[u'] \bar{y}$. It implies that there exists $\beta : u \to u'$  in $\Sym \gpd1$ such that $\beta\cdot p = p'$ and $\bar{y} s(\beta) = k'$. Since $\id[b] \in \bigcup\kit2(b)^\orth$ and $\beta \alpha^{-1} \cdot p' =p'$, we have $\beta\alpha^{-1}  \in \bigcup \oc \kit1(u')^\orth$. Now, $f \bar{y}  s(\beta) = f \bar{y}  s(\alpha)$ implies that $\beta \alpha^{-1}$ is in $\Stab(X^\Sym(u'))$. Since $X \in \StPSh(\kitstr{\gpd1})$, we obtain that $\alpha = \beta$ as desired. Hence, $\eta_{u,b}$ is well-defined, it remains to show that it is injective and surjective.

	For injectivity, if there are $p$ and $p'$ in $P(b,u)$ such that $p \coendtensor[u] \id= p' \coendtensor[u] \id$ then there exists $\alpha : u \to u'$ such that $\alpha\cdot p =p'$ and $\alpha\circ \id =\id$ which implies that $p=p'$. For surjectivity, let $p \coendtensor[v] \bar{z} : b \to \rsLan{P}(s(u))$ be a generic map. Since the diagram below commutes, there exists a unique $k : s(u) \to s(v)$ in $\StPSh(\kitstr{\gpd1})$ such that $\bar{z} k = \id$ and $p \coendtensor[v] k \bar{z} = p \coendtensor[v] \id$.  
	\begin{center}
		\begin{tikzpicture}[scale=0.75]
			
			\node (A) at (0, 2) {$b$};
			\node (B) at (3,2) {$\rsLan{P}(s(v))$};
			\node (C) at (0,0) {$\rsLan{P}(s(u))$};
			\node (D) at (3,0) {$\rsLan{P}(s(u))$};
			
			\draw [->] (A) to node [above] {$(p \coendtensor[v] \id)$} (B);
			\draw [->] (A) to node [left] {$(p \coendtensor[v] \bar{z})$} (C);
			\draw [->] (B) to node [right] {$\rsLan{P}(\bar{z})$} (D);
			\draw [->] (C) to node [below] {$\rsLan{P}(\id)$} (D);
		\end{tikzpicture}
	\end{center}
	Since $p \coendtensor[v] \id$ is generic, we can apply the same reasoning and obtain that $\bar{z}$ has a left inverse as well and is therefore an isomorphism. By Lemma \ref{lem:IsomorphismSumRep}, there exists $
	\alpha : u \to u'$ in $\gpd1$ such that $\bar{z} = s(\alpha)$. Hence, $p \coendtensor[v] \bar{z} = \alpha^{-1} \cdot p \coendtensor[u] \id$ which implies that $\eta_{b,u}$ is surjective as desired.
\end{proof}

\begin{lem}\label{lem:monoSumRepresentables}
	For a Boolean kit $\kitstr{\gpd1}$ and presheaves $X$ and $Y$ in $\StPSh(\kitstr{\gpd1})$, if $Y$ is a coproduct of representables $\coprod_{i\in I}\yon (a_i)$ and there is a monomorphism $m : X \hookrightarrow Y$, then $X \cong \coprod_{j\in J}\yon (a_j)$ where $J \subseteq I$.
\end{lem}

\begin{proof}
	By Lemma \ref{lem:StShRepresentation}, $X$ is isomorphic to $\coprod_{j\in J}\extyon{b_j}{G_j}$ where each $b_j \in \gpd1$ and $G_j \in \kit1(b_j)$. Since $m$ is monic, $X$ is a free action, i.e. for all $a \in \gpd1$, $x\in X(a)$ and $\alpha: a\to a$ such that $x \cdot \alpha = x$, we must have $\alpha = \id[a]$. Indeed, $  x \cdot\alpha = x$ implies $(m_a(x)) \cdot \alpha = m_a(x)$ which entails that $\alpha=\id[a]$ since $Y$ is a free action. Hence, all the groups $G_j$ are trivial and $X \cong \coprod_{j\in J}\yon (a_j)$. Now, since $\PSh{\gpd1}(X,Y)$ is isomorphic to $\coprod_{\sigma: J \to I} \prod_{j\in J} \gpd1(b_j, a_{\sigma(j)})$, $m: X\to Y$ corresponds to a pair $(\sigma, (\alpha_j)_{j})$ of a function $\sigma : I \to J$ and a family of isomorphisms $\alpha_j : b_j \to a_{\sigma(j)}$. Since $m$ is monic, $\sigma$ is injective and we obtain the desired result.
\end{proof}

\begin{lem}\label{lem:IsomorphismSumRep}
	For sequences $\seq{a_1, \dots, a_n}, \seq{b_1, \dots, b_n}$ in $\Sym\gpd1$, if there is an isomorphism $f:  \coprod_{1\leq i\leq n} \yon(a_i) \cong \coprod_{1\leq i\leq n}  \yon(b_j)$ in $\PSh{\gpd1}$, then there exists a morphism $\alpha \in \Sym \gpd1$ such that $f = s (\alpha)$.
\end{lem}

\begin{proof}
	Immediate corollary of the isomorphism $\PSh{\gpd1}(\seq{a_i},\seq{b_j})\cong \coprod_{\sigma: J \to I} \prod_{j\in J} \gpd1(b_j, a_{\sigma(j)})$.
\end{proof}

\begin{lem}\label{lem:StableGenericFact}
	Let $\kitstr{\gpd1}$ and $\kitstr{\gpd2}$ be Boolean kits and $T : \StPSh(\kitstr{\gpd1}) \to \StPSh(\kitstr{\gpd2})$ be a stable functor. For every $X$ in $\StPSh(\kitstr{\gpd1})$ and $t: b \to T(X)$ in $\StPSh(\kitstr{\gpd2})$, there exists $u\in\Sym \gpd1$, $g : b \to T(s(u))$ generic and $\bar{x} :s(u) \to X$ such that $t = T(\bar{x}) g$.
\end{lem}

\begin{proof}
	Since $T$ is stable, $t$ can be factored as $b \xlongto{g} T(Y)\xto{T(f)} T(X)$ where $g$ is generic. By Corollary \ref{cor:RepresentationQuantitativePresheaves}, $Y$ is isomorphic to $\colim_{i\in I} \coprod_{j\in J_i} \extyon {a_{ij}} {G_{ij}}$ where each $G_{ij}$ is a  finitely generated group in $\kit1(a_{ij})$ and $I$ is filtered. Since $T$ is finitary, we have 
	\[
	T(Y) \cong \colim_{i\in I} T(\coprod_{j\in J_i} \extyon {a_{ij}} {G_{ij}}).
	\]
	Since filtered colimits are computed pointwise in $\StPSh(\kitstr{\gpd2})$, it implies that $g$ can be factored as 
	\[
	b \xlongto{g'} T(\coprod_{j\in J} \extyon {a_{j}} {G_{j}}) \xto{T(\colimin)} \colim_{i\in I} T(\coprod_{j\in J_i} \extyon {a_{ij}} {G_{ij}}).
	\]
	
	For each group $G_j$, the projection morphism $q_j : \yon {a_{j}} \to \extyon {a_{j}} {G_{j}})$ is an epimorphism which implies that $q:=\coprod_{j\in J} q_j: \coprod_{j\in J} \yon {a_{j}} \to \coprod_{j\in J} \extyon {a_{j}} {G_{j}}$ is also an epimorphism. Since $T$ is epi-preserving, we can factor $g'$ as
	\[
	b \xlongto{g''}  T(\coprod_{j\in J} \yon {a_{j}}) \xlongto{q}  T(\coprod_{j\in J} \extyon {a_{j}} {G_{j}}).
	\]
	Since $g$ is generic, there exists a unique $h : Y \to \coprod_{j\in J} \yon {a_{j}}$ such that $\colimin q h=\id$ and $T(h)g=g''$. 
	
	\begin{center}
		\begin{tikzpicture}[scale=0.95]
			\node (A) at (0, 2.3) {$b$};
			\node (B) at (3,2.3) {$T(\coprod_{j\in J} \yon {a_{j}})$};
			\node (C) at (0,0) {$T(Y)$};
			\node (D) at (3,0) {$T(Y)$};
			
			\draw [->] (A) to node [above] {$g''$} (B);
			\draw [->] (A) to node [left] {$g$} (C);
			\draw [->] (B) to node [right] {$T(\colimin q)$} (D);
			\draw [->] (C) to node [below] {$T(\id)$} (D);
			\draw [->, dotted] (C) to node [fill=white] {$T(h)$} (B);
		\end{tikzpicture}
	\end{center}
	Since $h$ is split monic, by Lemma \ref{lem:monoSumRepresentables}, we have $Y \cong \coprod_{k\in K} \yon {a_{k}}$ where $K \subseteq J$.
\end{proof}

\begin{lem}\label{lem:StableCounitIso}
	For a stable functor $T:\StPSh(\kitstr{\gpd1})\to \StPSh(\kitstr{\gpd2})$, the transformation $\varepsilon : \rsLan{\trace{T}} \Rightarrow T$ whose components $\varepsilon_{X,b}$ are given by
	\[
	(t \coendtensor[u] \bar{x} : s(u) \to X ) \mapsto T(\bar{x} ) t : b \to T(X)
	\]
	is a natural isomorphism.
\end{lem}

\begin{proof}
	We first show that the map $\varepsilon_{X,b}$ is well-defined. For elements $t \coendtensor[u] \bar{x}$ and $t' \coendtensor[u'] \bar{x}'$ in $\rsLan{\trace{T}}(X,b)$, if $t \coendtensor[u] \bar{x} = t' \coendtensor[u'] \bar{x}'$, then there exists a unique $\delta : u \to u'$ in $\Sym \gpd1$ such that $\delta \cdot t = T(s(\delta)) t =t'$ and $\bar{x}' (s(\delta)) = \bar{x}$. Hence, $T(\bar{x}) t = T(\bar{x}') T(s(\delta)) t= T(\bar{x}') t'$.
	
	For injectivity, assume that there exists $t \coendtensor[u] \bar{x}$ and $t' \coendtensor[u'] \bar{x}'$ in $\rsLan{\trace{T}}(X,b)$ such that $T(\bar{x}) t = T(\bar{x}') t'$. Since $t$ is generic, there exists a unique $f : s(u) \to s(u')$ such that $\bar{x}' f = \bar{x}$ and $T(f) t =t'$. Likewise, since $t'$ is generic, there exists a unique $h : s(u') \to s(u)$ such that $\bar{x} h = \bar{x}'$ and $T(h) t' =t$. Using again the genericity of $t$ and $t'$, we obtain that $f h =\id $ and $h f =\id$. Hence, by Lemma \ref{lem:IsomorphismSumRep}, there exists $\alpha: u \to u'$ in $\Sym \gpd1$ such that $f \cong s\alpha$ which implies that $t \coendtensor[u] \bar{x} = t \coendtensor[u] \bar{x}' (s\alpha) = \alpha\cdot t \coendtensor[u'] \bar{x}' = t' \coendtensor[u'] \bar{x}'$ as desired.
	
	For surjectivity, let $t : b \to T(X)$ be a morphism in $\StPSh(\kitstr{\gpd2})$. By Lemma \ref{lem:StableGenericFact}, there exists $u \in \Sym \gpd1$, $g : b \to T(s(u))$ generic and $\bar{x} : s(u) \to X$ such that $t = T(\bar{x} ) g$. Hence, $t \coendtensor[u] g \in \rsLan{\trace{T}}(X,b)$ is a preimage for $t$.
\end{proof}

\subsection{Cartesian natural transformations.}

Our purpose now is to extend this result to a bicategorical equivalence, which we achieve by investigating the correspondence at the level of natural transformations.  

The situation at this level has subtle ramifications already in the setting of Berry's domain theory: for cpos $A$ and $B$ the pointwise order on the set of stable functions does not lead to a cartesian closed category. Indeed, to obtain cartesian closure one is forced to restrict the objects to cpos with bounded meets, and the stable function space $A \Rightarrow B$ must be ordered according to
\[
f \sqsubseteq_{\mathrm{st}} g \iff \Forall{ x \leq_A y \in  A}\ g(x) \wedge f(y) = f(x).
\]   
By considering $x \leq_A x$ in $A$, one necessarily has that $f(x) \leq_B g(x)$, so this is a strengthening of the pointwise order. In fact the order $\sqsubseteq_{\mathrm{st}}$ is precisely what is needed to characterize inclusion at the level of traces, and it is via a categorification of this order that we will establish our correspondence result. 

The appropriate notion is that of \defn{cartesian} natural transformations, whose naturality squares are pullbacks. These generalize the order $\sqsubseteq_{\mathrm{st}}$ on stable functions (consider the naturality square for $x \leq_A  y$), and additionally preserve and reflect generic morphisms \cite{Weber2004generic}. Cartesian natural transformations are considered also in the settings of normal functors and polynomial functors \cite{GirardNormal, gambinokock2013}. With the interpretation of polynomial functors as arising from sets of operations with arities, a cartesian natural transformation corresponds to a mapping between operations that preserves arities.

We show in this section that our construction yields a biequivalence between the bicategory of stable species and the $2$-category $\eSt$ defined below:

\begin{defi}
	The $2$-category $\eSt$ has Boolean kits as objects, stable functors $\StPSh(\kitstr{\gpd1}) \to \StPSh(\kitstr{\gpd2})$ as morphisms $\kitstr{\gpd1} \to \kitstr{\gpd2}$, and cartesian natural transformations as 2-cells.
\end{defi}

Natural transformations between stable species in $\SEsp$ yield cartesian transformations between the corresponding stable functors.

\begin{prop}
	\label{prop:NatTransCartesian}
	Let $(\gpd1,\kit1), (\gpd2,\kit2)$ be Boolean kits and let 	
	$f : P \Rightarrow Q$ be a natural transformation between stable species 
	$P, Q: \Sym(\gpd1, \kit1) \profto (\gpd2, \kit2)$. 
	The natural transformation 
	$\lansyon f: \lansyon P\Rightarrow \lansyon Q: \StPSh(\gpd1, \kit1) \to
	\StPSh(\gpd2, \kit2)$, 
	canonically induced by left Kan extension, is cartesian.
\end{prop}

\begin{proof}
	Let $g : X \rightarrow Y$ in $\StPSh(\kitstr{\gpd1})$, we want to show that the square below is a pullback in $\Set$ for all $b \in \gpd2$:  
	\begin{center}
		\begin{tikzpicture}[scale=0.75]
			\node (A) at (0, 2.3) {$\rsLan{P}(X)(b)$};
			\node (B) at (3,2.3) {$\rsLan{Q}(X)(b)$};
			\node (C) at (0,0) {$\rsLan{P}(Y)(b)$};
			\node (D) at (3,0) {$\rsLan{Q}(Y)(b)$};
			
			\draw [->] (A) to node [above] {$\rsLan{f}_X$} (B);
			\draw [->] (A) to node [left] {$\rsLan{P}(g )$} (C);
			\draw [->] (B) to node [right] {$\rsLan{Q}(g )$} (D);
			\draw [->] (C) to node [below] {$\rsLan{f}_Y$} (D);
		\end{tikzpicture}
	\end{center}
	We show that for all $(p\coendtensor[u_1] \bar{y})\in \rsLan{P}(Y)(b)$ and 
	$(q\coendtensor[u_2] \bar{x})\in \rsLan{Q}(X)(b)$ such that 
	\[
	f_{u_1}(p)\coendtensor[u_1]\bar{x} = \rsLan{f}_Y(p\coendtensor[u_1]\bar{x}) = \rsLan{Q}(g )(q\coendtensor[u_2] \bar{x}) = q\coendtensor[u_2] g (\bar{x})
	\] 
	there exists a unique $t\in \rsLan{P}(X)(b)$ such that $\rsLan{P}(g )(t)=p\coendtensor[u_1] \bar{y}$ and  $\rsLan{f}_X(t)= q\coendtensor[u_2] \bar{x}$. The equality $f_{u_1}(p)\coendtensor[u_1] \bar{y} = q\coendtensor[u_2] g (\bar{x})$ implies that there exists $\alpha : u_1 \rightarrow u_2$ in $\Sym \gpd1$ such that $q =  \alpha \cdot f(p)$ and $g(\bar{x}) \cdot \alpha =  g(\bar{x})  \circ s(\alpha)= \bar{y}$. Define $t$ to be ${p\coendtensor[u_1](\bar{x} \cdot \alpha)}\in \rsLan{P}(X)(b)$, we then obtain that 
	\[
	\rsLan{P}(g )(t)
	= p \coendtensor[u_1] g (\bar{x} \cdot \alpha) 
	= p \coendtensor[u_1] g (\bar{x}) \cdot \alpha 
	= p \coendtensor[u_1] \bar{y}
	\]
	and
	\[
	\rsLan{f}_X(t)
	= f(p) \coendtensor[u_1] \bar{x} \cdot \alpha
	= \alpha\cdot f(p) \coendtensor[u_2] \bar{x}
	= q \coendtensor[u_2] \bar{x}.
	\]
	
	Assume now that there exists $v = p_0 \coendtensor[u_0] \bar{x_0} \in \rsLan{P}(X)(b)$ such that $\rsLan{P}(g )(v)=p\coendtensor[u_1] \bar{y}$ and $\rsLan{f}_X(v)= q\coendtensor[u_2] \bar{x}$. We then have that $p_0 \coendtensor[u_0] g (\bar{x_0}) = p \coendtensor[u_1] \bar{y}$ and  $f(p_0)\coendtensor[u_0] \bar{x_0}  = q \coendtensor[u_2] \bar{x}$.  Hence, there exists $\beta : u_0 \rightarrow u_1$ in $\Sym \gpd1$ such that 
	$q = f (p_0) \cdot \beta$ and $\bar{x} \cdot \beta = \bar{x_0}$, and there exists $\gamma: u_1 \rightarrow u_0$ in $\Sym \gpd1$ such that $p_0 = \gamma \cdot p$ and
	$g (\bar{x_0}) \cdot \gamma = \bar{y}$. Therefore,
	\[
	\beta \gamma \cdot f (p)
	= 
	\beta \cdot f (\gamma \cdot p)
	= 
	\beta \cdot f (p_0)
	= 
	q
	\]
	and 
	\[	 
	g (\bar{x}) \cdot \beta \gamma
	=
	g (\bar{x} \cdot \beta) \cdot \gamma
	=
	g (\bar{x_0}) \cdot \gamma
	=
	\bar{y}.
	\]
	Therefore, $\beta \gamma \alpha^{-1}$ is in $\Stab(Y^\Sym(u_1))$ which implies that $\beta \gamma \alpha^{-1}\in \bigcup \oc \kit1(u_1)$ since $Y\in\StPSh(\kitstr{\gpd1})$. We also have $\beta \gamma \alpha^{-1} \cdot q =q$. Since $Q$ is a stable species and $\id[b]$ is in $\bigcup \kit2^\orth(b)$, we obtain that $\beta \gamma \alpha^{-1}$ is in $\bigcup \oc \kit1^\orth(u_1)$ which implies that $\alpha = \beta \gamma$. Hence, we have
	\[
	t
	=
	p\coendtensor[u_1] (\bar{x} \cdot \beta \gamma)
	=
	p\coendtensor[u_1] (\bar{x_0}\cdot \gamma )
	=
	(\gamma\cdot p)\coendtensor[u_0] \bar{x_0}
	=
	p_0\coendtensor[u_0] \bar{x_0}
	=
	v
	\]
	as desired. 
\end{proof}

We therefore have that $\rsLan{(-)}$ defines a functor $\SEsp(\kitstr{\gpd1}, \kitstr{\gpd2}) \to \Stable(\kitstr{\gpd1}, \kitstr{\gpd2})$.

\begin{lem}\label{lem:TraceCartTrans}
	Let $(\gpd1,\kit1), (\gpd2,\kit2)$ be Boolean kits and $T, S :  \StPSh(\kitstr{\gpd1}) \to \StPSh(\kitstr{\gpd2})$ be stable functors. For a cartesian transformation $f : T \Rightarrow S$, the transformation $\trace(f) : \trace(T) \Rightarrow \trace(S):\oc(\gpd1, \kit1) \profto (\gpd2, \kit2)$ whose components $\trace(f)_{(b,u)} : \trace(T)(b,u) \to \trace(S)(b,u)$ are given by:
	\[
	(t : b \to T(s(u)) \mapsto (f_{s(u)}(t) : b \to S(s(u)))
	\]
	is well-defined and natural.
\end{lem}

\begin{proof}
	Since cartesian transformations preserve and reflect generic elements, $f_{s(u)} (t)$ is generic if $t$ is generic. Naturality of $\trace(f)$ follows immediately from the naturality of $f$.
\end{proof}

\begin{thm}
\label{thm:correspondenceStable}
There is a biequivalence $\Stable \simeq \SEsp$.
\end{thm}
\begin{proof}
	We now have an adjoint equivalence $\SEsp(\kitstr{\gpd1}, \kitstr{\gpd2}) \simeq \Stable(\kitstr{\gpd1},\kitstr{\gpd2})$ for Boolean kits $\kitstr{\gpd1}$ and $\kitstr{\gpd2}$. Indeed, the functors $\rsLan{(-)}$ and $\trace$ are well-defined by Propositions \ref{prop:rsLanStable}, \ref{prop:NatTransCartesian}, \ref{prop:StableQProf} and Lemma \ref{lem:TraceCartTrans}. They form an equivalence by Lemmas \ref{lem:StableUnitIso} and \ref{lem:StableCounitIso} and since any equivalence induces an adjoint equivalence, we obtain the desired result.
\end{proof}
As a corollary, we obtain that $\Stable$ is cartesian closed (as a bicategory).

\section{The linear decomposition of stable functors}\label{sec:Linearity}
The linear decomposition of the function space 
\[
A \Rightarrow B  \ \  =\ \  \oc A \multimap B
\]
in Girard's original coherence space model \cite{GirardCoherence} is
what triggered the development of linear logic. With this insight, one can define
cartesian closed models starting from models of linear logic, as we
have done for stable species using stabilized profunctors.

Unlike stable species, the notion of stable functor does not
immediately lead to a linear decomposition as above, but our
correspondence theorem (Theorem~\ref{thm:correspondenceStable})
induces the following chain of equivalences:
\begin{equation}
\label{eq:decomposition}
\Stable(\kitstr{\gpd1}, \kitstr{\gpd2}) \simeq \SEsp(\kitstr{\gpd1}, \kitstr{\gpd2}) = \SProf(\oc \kitstr{\gpd1}, \kitstr{\gpd2}).
\end{equation}
Below we describe a purely extensional decomposition of the form
\[
\Stable(\kitstr{\gpd1}, \kitstr{\gpd2}) \simeq 
\Linear(\oc\kitstr{\gpd1}, \kitstr{\gpd2})
\]
for $\Linear$ a sub-2-category of $\Stable$ which we define. This decomposition is derived from $\eqref{eq:decomposition}$ through a biequivalence $\SProf \simeq \Linear$ analogous to the biequivalence between the bicategory $\Prof$ and the 2-category $\Cocont$ in Theorem \ref{thm:BiequivalenceProfCocont}. Concretely, this biequivalence is given in the direction $\Prof \to \Cocont$ by the mapping $P \mapsto P\lanyon$, and in the reverse direction by pre-composition with the Yoneda embedding: for any functor $\PSh(\gpd1) \to \PSh(\gpd2)$ the composite $\gpd1 \hookrightarrow \PSh(\gpd1) \to \PSh(\gpd2)$ determines a profunctor. 

That this is a biequivalence relies on two facts. First, because $\yon$ is an embedding, the 2-cell fitting in the Kan extension diagram  
\[
	\begin{tikzpicture}[scale=0.9]
	\node (A) at (0,1.25) {$\gpd1$};
	\node (B) at (3,1.25) {$\PSh(\gpd2)$};
	\node (C) at (1.5,0) {$\PSh(\gpd1)$};
	\node (D) at (1.5,0.75) {$\Downarrow$};
	
	\draw [->] (A) to node [above] {$P$} (B);
	\draw [<-, dotted] (B) to node [below right]  {$P\lanyon$} (C);
	\draw [right hook->] (A) to node [below left] {$\yon$} (C);
	\end{tikzpicture}
\]
is an isomorphism. Second, the categories $\PSh(\gpd1)$ and $\PSh(\gpd2)$ are locally small and cocomplete, so that $P\lanyon$ always has a right adjoint $Y \mapsto (a \mapsto \PSh(\gpd2)(P(a) , Y))$, and for functors $\PSh(\gpd1) \to \PSh(\gpd2)$ (with $\gpd1$ and $\gpd2$ small) to be cocontinuous and to have a right adjoint are equivalent properties. 

The first of these facts extends to $\SProf$ without difficulty: for kits $\kitstr{\gpd1}$ and $\kitstr{\gpd2}$ and a stabilized profunctor $P:\kitstr{\gpd1} \profto \kitstr{\gpd2}$, the restriction of $P\lanyon$ to stabilized presheaves
 coincides with the left Kan extension of $P$
along the restricted Yoneda embedding $\yy[\kit1] : \gpd1 \to \StPSh(\kitstr{\gpd1})$  
and indeed we have $P\lanyon \circ \yy[\kit1] \cong P$ (Lemma \ref{lem:RestrictedKanSProf}). 

On the other hand, categories of stabilized presheaves do not have all colimits, and the functors $P\lanyon : \StPSh(\kitstr{\gpd1}) \to \StPSh(\kitstr{\gpd2})$ are not left adjoints in general. They do however have a right adjoint on each slice, and this property suffices to characterize them among the stable functors.
\begin{defi}\label{def:linearfunctor}
For kits $\kitstr{\gpd1}$ and $\kitstr{\gpd2}$, a functor $L : \StPSh(\kitstr{\gpd1}) \to \StPSh(\kitstr{\gpd2})$ is called \defn{linear} if it is stable and a local left adjoint. 
\end{defi}
Local left adjoints do not preserve all colimits, but for a linear functor $L : \StPSh(\kitstr{\gpd1}) \to \StPSh(\kitstr{\gpd2})$ the following weaker property holds: for a diagram $D : J \to \StPSh(\kitstr{\gpd1})$, if $\colim_{j \in J} D_j$ exists in $\StPSh(\kitstr{\gpd1})$, and $\colim_{j \in J} F(D_j)$ exists in $\StPSh(\kitstr{\gpd2})$, then $F(\colim_{j \in J} D_j) = \colim_{j \in J} F(D_j)$. 
Therefore, in particular, $L$ preserves all sums, and for a group $G \in \kit1(a)$, $L(\extyon a G)$ is the colimit of the diagram \[
G \to \gpd1 \hookrightarrow \StPSh(\kitstr{\gpd1}) \xrightarrow{L} \StPSh(\kitstr{\gpd2}).\]

We call $\Linear$ the sub-2-category of $\Stable$ consisting of linear functors and 
cartesian natural transformations between them, and proceed to construct a biequivalence $\SProf \simeq \Linear$ in several steps. 

We first verify that for a stabilized profunctor $P:\kitstr{\gpd1} \profto \kitstr{\gpd2}$, the restriction of $P\lanyon\incl[\kit1]$ is isomorphic to the left Kan extension of $P$ along the restricted Yoneda embedding $\yy[\kit1] : \gpd1 \to \StPSh(\kitstr{\gpd1})$ by using the universal property of Kan extensions in terms of weighted colimits:
\begin{prop}\label{prop:LanWeightedColim}
	For functors $F: \cat{A} \to \cat{C}$ and $H : \cat{A} \to \cat{B}$, there is an isomorphism:
	\[
	\cat{C}\left((\Lan_H F)(b)), c\right) \cong \PSh{\cat{A}} \left(\cat{B}(H(-), b), \cat{C}(F(-), c)\right).
	\]
\end{prop}

\begin{lem}\label{lem:KanExtensions}
	Let $I : \cat{C} \to \cat{D}$ and $J: \cat{B} \to \cat{E}$ be embeddings. For functors $F: \cat{A} \to \cat{C}$, $H : \cat{A} \to \cat{B}$ and $L: \cat{B} \to \cat{C}$, if $IL \cong (\Lan_{JH} (I F)) J$ then $L \cong \Lan_H F$.
	\begin{center}
		\begin{tikzpicture}[scale=1.1]
			\node (A) at (0,1) {$\cat{A}$};
			\node (C) at (2.5,1) {$\cat{C}$};
			\node (B) at (1.25,0) {$\cat{B}$};
			\node (D) at (5,1) {$\cat{D}$};
			\node (E) at (2.5,-1) {$\cat{E}$};
			
			\draw [->] (A) to node [above] {$F$} (C);
			\draw [->] (B) to node [below right]  {$L$} (C);
			\draw [->] (A) to node [below left] {$H$} (B);
			\draw [->] (B) to node [below left] {$J$} (E);
			\draw [->] (C) to node [above] {$I$} (D);
			\draw [->] (E) to node [below right] {$\Lan_{JH}(IF)$} (D);
		\end{tikzpicture}
	\end{center}
	
\end{lem}

\begin{proof}
	For all $b\in \cat{B}$ and $c \in \cat{C}$, we have: 
	\[\begin{aligned}
		\cat{C}\left(L(b), c\right) &\cong \cat{D}(IL(b), I(c))\\
		&\cong \cat{D}( \Lan_{JH} (I F)(J b), I(c))\\
		&\cong \PSh{\cat{A}} \left(\cat{E}(JH(-), J(b)), \cat{D}(IF(-), I(c))\right)\\
		&\cong\PSh{\cat{A}} \left(\cat{B}(H(-), b), \cat{C}(F(-), c)\right)
	\end{aligned}
	\]
	By Proposition \ref{prop:LanWeightedColim}, we obtain the desired result.
\end{proof}

\begin{lem}\label{lem:RestrictedKanSProf}
	Let $\kitstr{\gpd1} = (\gpd1,\kit1)$ and $\kitstr{\gpd2} = (\gpd2,\kit1)$ be Boolean kits. For a profunctor $P: \gpd1\profto\gpd2$, if the functor $P^{\#} \incl[\kit1] : \StPSh(\kitstr{\gpd1}) \to \PSh{\gpd2}$ factors through $\incl[\kit2]$ by a functor $L: \StPSh(\kitstr{\gpd1}) \to \StPSh(\kitstr{\gpd2})$, then there exists a functor $Q : \gpd1 \to \StPSh(\kitstr{\gpd2})$ such that $P$ factors through $\incl[\kit2]$ by $Q$ and $ L \cong \Lan_{\yy[\kit1]}Q$.
	
	\begin{center}
		\begin{tikzpicture}[scale=1]
			\node (A) at (0,1) {$\gpd1$};
			\node (D) at (5,1) {$\PSh{\gpd2}$};
			\node (E) at (2.5,-1) {$\PSh{\gpd1}$};
			\node at (6,0) {$=$};
			\draw [->] (A) to node [above] {$P$} (D);
			
			\draw [->] (A) to node [below left] {$\yon[\gpd1]$} (E);
			\draw [->] (E) to node [below right] {$P^\#$} (D);
			
			\begin{scope}[xshift=7cm]
				\node (A) at (0,1) {$\gpd1$};
				\node (C) at (2.5,1) {$\StPSh(\kitstr{\gpd2})$};
				\node (B) at (1.25,0) {$\StPSh(\kitstr{\gpd1})$};
				\node (D) at (5,1) {$\PSh{\gpd2}$};
				\node (E) at (2.5,-1) {$\PSh{\gpd1}$};
				
				\draw [->] (A) to node [above] {$Q$} (C);
				\draw [dotted, ->] (B) to node [below right]  {$L$} (C);
				\draw [->] (A) to node [below left] {$\yy[\kit1]$} (B);
				\draw [->] (B) to node [below left] {$\incl[\kit1]$} (E);
				\draw [->] (C) to node [above] {$\incl[\kit2]$} (D);
				\draw [->] (E) to node [below right] {$P^\#$} (D);
			\end{scope}
		\end{tikzpicture}
	\end{center}
	
\end{lem}

\begin{proof}
	The restricted functor $L:\StPSh(\kitstr{\gpd1}) \to \StPSh(\kitstr{\gpd2}) $ such that $P^\# \incl[\kit1] = \incl[\kit2] L$ is obtained from Lemma \ref{lem:qprofunctor_characterisation}. Since $\yon[\gpd1]$ is fully faithful, we have $P \cong P^{\#} \yon[\gpd1] = P^{\#} \incl[\kit1] \yy[\kit1] = \incl[\kit2] L \yy[\kit1]$. Let $Q := L \yy[\kit1] : \gpd1 \to \StPSh(\kitstr{\gpd2})$, we obtain that $L \cong \Lan_{\yy[\kit1]} Q$ by Lemma \ref{lem:KanExtensions}.
\end{proof}

For a stabilized profunctor $P: \kitstr{\gpd1} \profto \kitstr{\gpd2}$ we denote by $\rLan{P} : \StPSh(\kitstr{\gpd1}) \to \StPSh(\kitstr{\gpd2})$ the functor $\Lan_{\ryon[{\kit1}]} P \cong P^\# \incl[\kit1]$. We proceed to show that $\rLan{(-)}$ induces a functor from $\SProf(\kitstr{\gpd1}, \kitstr{\gpd2})$ to $\Lin(\kitstr{\gpd1},\kitstr{\gpd2})$.

\begin{lem}\label{lem:FreenessLinear}
	For Boolean kits $\kitstr{\gpd1}, \kitstr{\gpd2}$, a presheaf $X\in\StPSh(\kitstr{\gpd1})$ and a stabilized profunctor $P: \kitstr{\gpd1} \profto \kitstr{\gpd2}$, if two elements $p \coendtensor[a] x$ and $ p' \coendtensor[a'] x'$ in $P^\# X$ are equal, there exists a unique $\alpha : a \to a'$ in $\gpd1$ such that $x' \cdot f = x$ and $f \cdot p = p'$. 
\end{lem}
\begin{proof}
	Similar to Lemma \ref{lem:FreenessStable}.
\end{proof}

\begin{lem}\label{lem:QProfStable}
	For kits $\kitstr{\gpd1} = (\gpd1,\kit1)$ and $\kitstr{\gpd2} = (\gpd2,\kit2)$, if a profunctor $P:\gpd1\profto \gpd2$ is in $\SProf(\kitstr{\gpd1}, \kitstr{\gpd2})$, then $\rLan{P} :\StPSh(\kitstr{\gpd1})\to\StPSh(\kitstr{\gpd2})$ is stable.
\end{lem}

\begin{proof}
	Since $P^{\#} : \PSh{\gpd1} \to \PSh{\gpd2}$ is cocontinuous by Theorem \ref{thm:BiequivalenceProfCocont} and the embeddings $\incl[\kit1]: \StPSh(\kitstr{\gpd1})\hookrightarrow\PSh{\gpd1}$, $\incl[\kit2]: \StPSh(\kitstr{\gpd2})\hookrightarrow\PSh{\gpd2}$ create filtered colimits and epimorphisms by Corollary \ref{cor:QuantitativePresheavesCreation}, the functor $\rLan{P} \cong P^{\#}  \incl[\kit1]: \StPSh(\kitstr{\gpd1}) \to\StPSh(\kitstr{\gpd2})$ is finitary and epi-preserving.
	
	It remains to show that $\rLan{P}$ admits generic factorizations. By Proposition \ref{prop:ColimGeneric}, it suffices to show that $\rLan{P}$ admits generic factorizations relative to representables. Consider a morphism $\yon b \to \rLan{P}(X)$ in $\StPSh(\kitstr{\gpd2})$, it is of the form $p \coendtensor[a]x$ with $p \in P(b,a)$ and $x \in X(a)$ for some $a\in \gpd1$. We show that
	
	\begin{center}
		\begin{tikzpicture}
			\node (A) at (0, 0) {$ \yon b$};
			\node (B) at (2.5, 0) {$\rLan{P} (\yon a)$};
			\node (C) at (5, 0) {$\rLan{P}(X)$};
			
			\draw [->] (A) -- node [above] {$p \coendtensor[a] \id$} (B);
			\draw [->] (B) to node [above] {$\rLan{P}(x)$}  (C);
		\end{tikzpicture}
	\end{center}
	is a generic factorization for $p \coendtensor[a]x$. The equality $p \coendtensor[a]x = \rLan{P}(x) (p \coendtensor[a] \id)$ is immediate, it remains to prove that $p \coendtensor[a] \id$ is generic.
	
	Assume that there exist $Y,Z \in \StPSh(\kitstr{\gpd1})$ and morphisms $p' \coendtensor[a'] y' : \yon b \to \rLan{P}(Y)$ and $f: Y \to Z, z : \yon a \to Z$ such that the following diagram commutes in $\StPSh(\kitstr{\gpd2})$:

	\begin{center}
		\begin{tikzpicture}[scale=1]
			
			\node (A) at (0, 2) {$\yon b$};
			\node (B) at (3,2) {$\rLan{P} (Y)$};
			\node (C) at (0,0) {$\rLan{P} (\yon a)$};
			\node (D) at (3,0) {$\rLan{P} (Z)$};
			
			\draw [->] (A) to node [above] {$p' \coendtensor[a'] y'$} (B);
			\draw [->] (A) to node [left] {$p \coendtensor[a] \id$} (C);
			\draw [->] (B) to node [right] {$\rLan{P}(f)$} (D);
			\draw [->] (C) to node [below] {$\rLan{P}(z)$} (D);
			
		\end{tikzpicture}
	\end{center}
	
	The equality $\rLan{P}(f) (p' \coendtensor[a'] y') = p \coendtensor[a] z$ is equivalent to $p' \coendtensor[a'] f y' = p \coendtensor[a] z$ which implies by Lemma \ref{lem:FreenessLinear} that there exists a unique $\alpha : a \to a'$ in $\gpd1$ such that $(f y') \cdot \alpha = z$ and $\alpha \cdot p = p'$. We then have $\rLan{P}(y' \cdot \alpha) (p \coendtensor[a] \id) = p \coendtensor[a] y' \cdot \alpha = \alpha \cdot p \coendtensor[a']y' = p' \coendtensor[a'] y'$. It remains to show uniqueness, assume that there exists $y : \yon a \to Y$ such that $f y = z$ and $\rLan{P}(y)(p \coendtensor[a] \id) = p' \coendtensor[a'] y'$ i.e. $p \coendtensor[a] y = p' \coendtensor[a'] y'$. It implies that there exists $g : a \to a'$ such that $y' \cdot g = y$ and $g \cdot p = p'$. Hence, $(f y') \cdot g = f y = z$ which implies that $\alpha=g$ so that $y = y' \cdot \alpha$ as desired.
	
	\begin{center}
		\begin{tikzpicture}[scale=1]
			
			\node (A) at (0, 2) {$\yon b$};
			\node (B) at (3,2) {$\rLan{P} (Y)$};
			\node (C) at (0,0) {$\rLan{P} (\yon a)$};
			\node (D) at (3,0) {$\rLan{P} (Z)$};
			
			\draw [->] (A) to node [above] {$p' \coendtensor[a'] y'$} (B);
			\draw [->] (A) to node [left] {$p \coendtensor[a] \id$} (C);
			\draw [->] (B) to node [right] {$\rLan{P}(f)$} (D);
			\draw [->] (C) to node [below] {$\rLan{P}(z)$} (D);
			\draw [->, dotted] (C) to node [fill=white] {$\rLan{P}(y)$} (B);
		\end{tikzpicture}
	\end{center}
\end{proof}

\begin{lem}\label{lem:QProfLinear}
For kits $\kitstr{\gpd1}, \kitstr{\gpd2}$ and a stabilized profunctor $P:\kitstr{\gpd1} \profto \kitstr{\gpd2}$, the functor $\rLan{P} :\StPSh(\kitstr{\gpd1}) \to \StPSh(\kitstr{\gpd2})$ is linear.
\end{lem}

\begin{proof}
	Since $\rLan{P}:\StPSh(\kitstr{\gpd1})\to\StPSh(\kitstr{\gpd2})$ is stable by Lemma \ref{lem:QProfStable}, we need to show that for all $X \in \StPSh(\kitstr{\gpd1})$, the induced functor $\StPSh(\kitstr{\gpd1})/_X \to  \StPSh(\kitstr{\gpd2})/_{\rLan{P}(X)}$ has a right adjoint $R_X$.
	For $f : X' \to X$ in $\StPSh(\kitstr{\gpd1})$ and $g : Y \to \rLan{P}(X)$ in $\StPSh(\kitstr{\gpd2})$, we require the following correspondence:
	\begin{center}
		\begin{tikzpicture}[scale=0.9]
			\node (A) at (0,1.5) {$\rLan{P}(X')$};
			\node (B) at (3,1.5) {$Y$};
			\node (C) at (1.5,0) {$\rLan{P}(X)$};

			\draw [->] (A) to node  {} (B);
			\draw [->] (B) to node [right]  {$g$} (C);
			\draw [->] (A) to node [left] {$\rLan{P}(f)$} (C);
			
			\node at (4.6,0.75) {$\Leftrightarrow$};
			
			\begin{scope}[xshift = 6cm]
				\node (A) at (0,1.5) {$X'$};
				\node (B) at (3,1.5) {$R_X(Y)$};
				\node (C) at (1.5,0) {$X$};
				
				\draw [->] (A) to node  {} (B);
				\draw [->] (B) to node [right]  {$R_X(g)$} (C);
				\draw [->] (A) to node [left] {$f$} (C);
			\end{scope}
		\end{tikzpicture}
	\end{center}
	Note that since $\rLan{P}=\Lan_{\yy[\kit1]}P$, the following sets are isomorphic:
	\[\begin{aligned}
		\StPSh(\kitstr{\gpd2})(\rLan{P}(X'), Y) & \cong \PSh{\gpd1} \left(\PSh{\gpd1}(\yy[\kit1](-), X'), \StPSh(\kitstr{\gpd2})(P(-), Y)\right)\\
		&\cong\PSh{\gpd1}\left( X', \StPSh(\kitstr{\gpd2})(P(-), Y)\right).
	\end{aligned}\]
	Hence, the left triangle commutes in $\StPSh(\kitstr{\gpd2})$ if and only if the following square commutes in $\PSh{\gpd1}$:
	\begin{center}
		\begin{tikzpicture}
			\node at (-1.2,0) {};
			\node (A) at (0, 1.5) {$X'$};
			\node (B) at (3, 1.5) {$\StPSh(\kitstr{\gpd2})(P(-), Y)$};
			\node (C) at (0, 0) {$X$};
			\node (D) at (3, 0) {$\StPSh(\kitstr{\gpd2})(P(-), \rLan{P}(X))$};
			
			\draw [->] (A) -- node  {} (B);
			\draw [->] (B) to node [right] {$\StPSh(\kitstr{\gpd2})(P(-), g)$} (D);
			\draw [->] (A) to node [left] {$f$} (C);
			\draw [->] (C) to node {} (D);
		\end{tikzpicture}
	\end{center}
	Let $R_X(Y)$ be given by the following pullback:
	
	\begin{center}
		\begin{tikzpicture}
			\node (A) at (0, 1.5) {$R_X(Y)$};
			\node (B) at (3, 1.5) {$\StPSh(\kitstr{\gpd2})(P(-), Y)$};
			\node (C) at (0, 0) {$X$};
			\node (D) at (3, 0) {$\StPSh(\kitstr{\gpd2})(P(-), P^{\#}(X))$};
			\node at (0.3, 1.2) {$\usebox\pullback$};
			\draw [->] (A) -- node  {} (B);
			\draw [->] (B) to node [right] {$\StPSh(\kitstr{\gpd2})(P(-), g)$} (D);
			\draw [->] (A) to node [left] {$R_X(g)$} (C);
			\draw [->] (C) to node {} (D);
		\end{tikzpicture}
	\end{center}
	Explicitly, for $a \in \gpd1$, $R_X(Y)(a)$ is given by:
	\[
	\{
	(x,h) \in X(a) \times \StPSh(\kitstr{\gpd2})(P( a), Y) 
	\mid
	\forall b\in \gpd2, \forall p \in P(b,a), p\coendtensor[a]x = g h p\}
	\]
	and $R_X(g)(x,h)=x$. We check that the presheaf $R_X(Y)$ is in $\StPSh(\kitstr{\gpd1})$.  For $(x,h)\in {R_X(Y)}(a)$ and $\alpha\in\Endo(a)$, since $(x,h)\cdot \alpha = (x\cdot \alpha,h \circ P( \alpha))$, if $(x,h)\cdot \alpha = (x,h)$ then $x\cdot \alpha=x$ and therefore that $\alpha\in\bigcup\kit1(a)$.
\end{proof}

At the $2$-cell level, we show that the $\rLan{(-)}$ operator induces cartesian transformations:

\begin{lem}\label{lem:NatTransCartesianLin}
For Boolean kits $\kitstr{\gpd1}, \kitstr{\gpd2}$ and a natural transformation $f$ in $\SProf(P,Q)$, the natural
transformation $\rLan{f}: \rLan{P} \Rightarrow \rLan{Q}:\StPSh(\kitstr{\gpd1})\to\StPSh(\kitstr{\gpd2})$ is cartesian.
\end{lem}

\begin{proof}
	Let $h : X \rightarrow Y$ in $\StPSh(\kitstr{\gpd1})$, we want to show that the square below is a pullback in $\Set$ for all $b \in \gpd2$.  
	\begin{center}
		\begin{tikzpicture}[scale=0.75]
			\node (A) at (0, 2.3) {$\rLan{P}(X)(b)$};
			\node (B) at (3,2.3) {$\rLan{Q}(X)(b)$};
			\node (C) at (0,0) {$\rLan{P}(Y)(b)$};
			\node (D) at (3,0) {$\rLan{Q}(Y)(b)$};
			
			\draw [->] (A) to node [above] {$\rLan{f}_X$} (B);
			\draw [->] (A) to node [left] {$\rLan{P}(h )$} (C);
			\draw [->] (B) to node [right] {$\rLan{Q}(h )$} (D);
			\draw [->] (C) to node [below] {$\rLan{f}_Y$} (D);
		\end{tikzpicture}
	\end{center}
	We show that for all $(p\coendtensor[a_1] y)\in \rLan{P}(Y)(b)$ and 
	$(q\coendtensor[a_2] x)\in \rLan{Q}(X)(b)$ such that 
	\[
	f(p)\coendtensor[a_1]x = \rLan{f}_Y(p\coendtensor[a_1]x) = \rLan{Q}(h )(q\coendtensor[a_2] x) = q\coendtensor[a_2] h (x)
	\] 
	there exists a unique $u\in \rLan{P}(X)(b)$ such that $\rLan{P}(h )(u)=p\coendtensor[a_1] y$ and  $\rLan{f}_X(u)= q\coendtensor[a_2] x$. The equality $f(p)\coendtensor[a_1]y = q\coendtensor[a_2] h (x)$ implies that there exists $\alpha : a_1 \rightarrow a_2$ in $\gpd{A}$ such that $q =  \alpha \cdot f(p)$ and $h(x) \cdot \alpha= y$. Define $u$ to be ${p\coendtensor[a_1](x \cdot \alpha)}\in \rLan{P}(X)(b)$, we then obtain that 
	\[
	\rLan{P}(h )(u)
	= p \coendtensor[a_1] h (x \cdot \alpha) 
	= p \coendtensor[a_1] h (x) \cdot \alpha 
	= p \coendtensor[a_1] y
	\]
	and
	\[
	\rLan{f}_X(u)
	= f(p) \coendtensor[a_1] x \cdot \alpha
	= \alpha \cdot f(p) \coendtensor[a_2] x
	= q \coendtensor[a_2] x.
	\]
	
	Assume now that there exists $v = p_0 \coendtensor[a_0] x_0 \in \rLan{P}(X)(b)$ such that $\rLan{P}(h )(v)=p\coendtensor[a_1] y$ and  $\rLan{f}_X(v)= q\coendtensor[a_2] x$. We then have that $p_0 \coendtensor[a_0] h (x_0) = p \coendtensor[a_1] y$ and  $f(p_0)\coendtensor[a_0] x_0  = q \coendtensor[a_2] x$.  Hence, there exists $\beta : a_0 \rightarrow a_1$ in $\gpd{A}$ such that 
	$q = f (p_0) \cdot \beta$ and $x \cdot \beta = x_0$, and there exists $\gamma : a_1 \rightarrow a_0$ in $\gpd{A}$ such that $p_0 = \gamma \cdot p$ and
	$h (x_0) \cdot \gamma = y$. Therefore,
	\begin{center}
		$
		\beta \gamma \cdot f (p)
		= 
		\beta \cdot f (\gamma \cdot p)
		= 
		\beta \cdot f (p_0)
		= 
		q
		$
		\quad and \quad
		$
		h (x) \cdot \beta \gamma
		=
		h (x \cdot \beta) \cdot \gamma
		=
		h (x_0) \cdot \gamma
		=
		y.
		$
	\end{center}
	Therefore, $\beta \gamma \alpha^{-1}$ is in $\Stab(Y(a_1))$ which implies that $\beta \gamma \alpha^{-1} \in \bigcup \kit1(a_1)$. We also have $\beta \gamma \alpha^{-1} \cdot q =q$. Since $Q$ is a stabilized profunctor and $\id[b]$ is in $\bigcup \kit2^\perp(b)$, we obtain that $\beta \gamma \alpha^{-1}$ is in $\bigcup \kit1^\perp(a)$ which implies that $\alpha = \beta \gamma$. Hence, we have
	\[
	u
	=
	p\coendtensor[a_1] (x \cdot \beta \gamma)
	=
	p\coendtensor[a_1] (x_0\cdot \gamma )
	=
	(\gamma\cdot p)\coendtensor[a_0] x_0
	=
	p_0\coendtensor[a_0] x_0
	=
	v
	\]
	as desired. 
\end{proof}

For the reverse direction, we define a linear trace operator $\tracelin :\Lin(\kitstr{\gpd1}, \kitstr{\gpd2}) \to \SProf(\kitstr{\gpd1}, \kitstr{\gpd2})$ mapping $L \mapsto L \yy[\kit1]$ for a linear functor $L: \StPSh(\kitstr{\gpd1}) \to \StPSh(\kitstr{\gpd2})$. We proceed to show that $\tracelin$ is a well-defined functor.

\begin{lem}\label{lem:LinearGenericFact}
	Let $\kitstr{\gpd1} = (\gpd1,\kit1)$, $\kitstr{\gpd2} = (\gpd2,\kit2)$ be Boolean kits and $L : \StPSh(\kitstr{\gpd1}) \to \StPSh(\kitstr{\gpd2})$ a linear functor. For every $X$ in  $\StPSh(\kitstr{\gpd1})$ and $l: b \to L(X)$ in $\StPSh(\kitstr{\gpd2})$, there exists $a\in \gpd1$, $g : b \to L(\yy[\kit1] (a))$ generic and $x:\yy[\kit1] (a) \to X$ such that $ l = L(x) g$.
\end{lem}

\begin{proof}
	Since $L$ is linear, $l$ can be factored as $b \xrightarrow{g} L(Y)\xrightarrow{L(f)} L(X)$ where $g$ is generic. By Corollary \ref{cor:RepresentationQuantitativePresheaves}, $Y \cong \coprod_{i\in I} \extyon {a_{i}} {G_{i}}$ where each $G_{i}$ is a finitely generated group in $\kit1(a_{i})$. Since $L$ preserves coproducts, we obtain that $L(Y) \cong \coprod_{i\in I} L(\extyon {a_{i}} {G_{i}})$.
	
	Since colimits are computed pointwise in $\StPSh(\kitstr{\gpd2})$, it implies that there exists $a$ and $G$ such that $g$ can be factored as 
	\[
	b \xrightarrow{g'} L(\extyon {a} {G})\xrightarrow{L(\colimin)} \coprod_{i\in I} L(\extyon {a_{i}} {G_{i}}).
	\]
	Since $\extyon {a} {G}$ is the colimit of $(g : \yy[\kit1] (a) \to \yy[\kit1] (a))_{g \in G}$ in $\StPSh(\kitstr{\gpd1})$ and $L$ is linear, $L(\extyon {a} {G})$ is the colimit of $(L(g) : L(\yy[\kit1](a)) \to L(\yy[\kit1](a)))_{g \in G}$ in $\StPSh(\kitstr{\gpd2})$. Hence, there exists $g'' : b \to L( \yy[\kit1](a))$ such that $L(q) g'' = g'$.
	
	Since $g$ is generic, there exists $h : Y \to \yy[\kit1] (a)$ such that $\colimin \circ q \circ h =\id$ and $L(h)g =g''$ where $q : \yy[\kit1] (a) \to \extyon {a} {G}$ is the quotient map. 
	\begin{center}
		\begin{tikzpicture}[scale=0.75]
			\node (A) at (0, 2.3) {$b$};
			\node (B) at (3,2.3) {$L(\yy[\kit1] (a))$};
			\node (C) at (0,0) {$L(Y)$};
			\node (D) at (3,0) {$L(Y)$};
			
			\draw [->] (A) to node [above] {$g''$} (B);
			\draw [->] (A) to node [left] {$g$} (C);
			\draw [->] (B) to node [right] {$L(\colimin \circ q)$} (D);
			\draw [->] (C) to node [below] {$L(\id)$} (D);
			\draw [->, dotted] (C) to node [fill=white] {$L(h)$} (B);
		\end{tikzpicture}
	\end{center}
	Since $h$ is split monic, by Lemma \ref{lem:monoSumRepresentables}, $Y$ is isomorphic to $\yy[\kit1] (a)$.
\end{proof}

\begin{cor}\label{cor:LinearYonGeneric}
	For Boolean kits $\kitstr{\gpd1} = (\gpd1,\kit1)$ and $\kitstr{\gpd2} = (\gpd2,\kit2)$, if a functor $L : \StPSh(\kitstr{\gpd1}) \to \StPSh(\kitstr{\gpd2})$ is linear, then for every $l :b \to L(\yy[\kit1] (a))$ in $\StPSh(\kitstr{\gpd2})$, $l$ is generic.
\end{cor}
\begin{proof}
	By Lemma \ref{lem:LinearGenericFact}, there exists $g : b \to L(\yy[\kit1](a_0))$ generic and $\alpha :\yy[\kit1](a_0)\to  \yy[\kit1](a)$ such that $l = L(\alpha) g$. Since $\alpha$ is an isomorphism, we obtain that $l$ is generic as well.
\end{proof}
\begin{lem}\label{lem:LinearQProf}
	For Boolean kits  $\kitstr{\gpd1} = (\gpd1,\kit1)$ and $\kitstr{\gpd2} = (\gpd2,\kit2)$, if a functor $L : \StPSh(\kitstr{\gpd1}) \to \StPSh(\kitstr{\gpd2})$ is linear then $\tracelin (L):\gpd1 \profto \gpd2$ given by $L \yy[\kit1]$ is a stabilized profunctor.
\end{lem}

\begin{proof}
	Let $l : b \to L(\yy[\kit1] (a))$ be in $\tracelin (L)(b,a)$ and $\alpha \in \gpd1(a,a)$, $\beta \in \gpd2(b,b)$ such that $\alpha \cdot l  \cdot \beta =l$, i.e. the following square commutes in $\StPSh(\kitstr{\gpd2})$:
	\begin{center}
		\begin{tikzpicture}[scale=0.75]
			\node (A) at (0, 2.3) {$b$};
			\node (B) at (3,2.3) {$L(\yy[\kit1] (a))$};
			\node (C) at (0,0) {$b$};
			\node (D) at (3,0) {$L(\yy[\kit1] (a))$};
			
			\draw [->] (A) to node [above] {$l$} (B);
			\draw [->] (C) to node [left] {$\beta$} (A);
			\draw [->] (B) to node [right] {$L(\yy[\kit1] \alpha)$} (D);
			\draw [->] (C) to node [below] {$l$} (D);
		\end{tikzpicture}
	\end{center}
	Assume that $\alpha \in \bigcup \kit1(a)$, then $\extyon a {\langle \alpha \rangle} \in \StPSh(\kitstr{\gpd1})$. Let $q :\yy[\kit1] (a) \to  \extyon a {\langle \alpha \rangle}$ be the quotient map, we then have that $q \circ \yy[\kit1] (\alpha) = q$ in $\StPSh(\kitstr{\gpd1})$ which implies that the following diagram commutes in $\StPSh(\kitstr{\gpd2})$:
	\begin{center}
		\begin{tikzpicture}[scale=0.8]
			\node (A) at (0, 2.2) {$b$};
			\node (B) at (3,2.2) {$L(\yy[\kit1] (a))$};
			\node (C) at (0,0) {$b$};
			\node (D) at (3,0) {$L(\yy[\kit1] (a))$};
			\node (E) at (7,1.1) {$L( \extyon a {\langle \alpha \rangle})$};
			
			\draw [->] (A) to node [above] {$l$} (B);
			\draw [->] (C) to node [left] {$\beta$} (A);
			\draw [->] (B) to node [left] {$L(\yy[\kit1] \alpha)$} (D);
			\draw [->] (C) to node [below] {$l$} (D);
			\draw [->] (B) to node [above] {$L(q)$} (E);
			\draw [->] (D) to node [below] {$L(q)$} (E);
		\end{tikzpicture}
	\end{center}
	Hence, $\beta \in \Stab( L(q) l)$ which implies that $\beta \in \bigcup \kit2(b)$.
	
	Assume now that $\beta \in \bigcup \kit2^\perp(b)$, we want to show that for all $n$, if $\alpha^n$ is in $\bigcup \kit1(a)$, then $\alpha^n = \id[a]$. If $\alpha^n \in \bigcup \kit1(a)$, then $\extyon a {\langle \alpha^n \rangle} \in \StPSh(\kitstr{\gpd1})$ and the following diagram commutes in $\StPSh(\kitstr{\gpd2})$:
	\begin{center}
		\begin{tikzpicture}[scale=0.8]
			\node (A) at (0, 2.2) {$b$};
			\node (B) at (3,2.2) {$L(\yy[\kit1] (a))$};
			\node (C) at (0,0) {$b$};
			\node (D) at (3,0) {$L(\yy[\kit1] (a))$};
			\node (E) at (7,1.1) {$L( \extyon a {\langle \alpha^n \rangle})$};
			
			\draw [->] (A) to node [above] {$l$} (B);
			\draw [->] (C) to node [left] {$\beta^n$} (A);
			\draw [->] (B) to node [left] {$L(\yy[\kit1]\alpha^n)$} (D);
			\draw [->] (C) to node [below] {$l$} (D);
			\draw [->] (B) to node [above] {$L(q')$} (E);
			\draw [->] (D) to node [below] {$L(q')$} (E);
		\end{tikzpicture}
	\end{center}
	where $q' :\yy[\kit1] (a) \to  \extyon a {\langle \alpha^n \rangle}$ is the quotient map in $\StPSh(\kitstr{\gpd1})$. Hence, $\beta^n \in \Stab(L(q')l)$ which implies that $\beta^n =\id[b]$ so that the following diagram commutes:
	\begin{center}
		\begin{tikzpicture}[scale =0.8]
			\node (A) at (0,0) {$L(\yy[\kit1] (a))$};
			\node (B) at (4,0) {$L( \extyon a {\langle \alpha^n \rangle})$};
			\node (C) at (0,2.5) {$b$};
			\node (D) at (4,2.5) {$L(\yy[\kit1] (a))$};
			
			\draw [->] (C) -- node [left] {$l$} (A);
			\draw [->] (A) -- node [below] {$L(q')$} (B);
			\draw [->] (C) -- node [above] {$l$} (D);
			\draw [->] (D) -- node [right] {$L(q')$} (B);
			\draw [->] (A) -- node [fill=white] {$L(\yy[\kit1]\alpha^n)$} (D);
		\end{tikzpicture}
	\end{center}
	Since $l$ is generic by Corollary \ref{cor:LinearYonGeneric}, we obtain that $\alpha^n = \id[a]$ as desired.
\end{proof}

Since $\SProf(\kitstr{\gpd1}, \kitstr{\gpd2})$ is a full subcategory of $\Prof(\gpd1, \gpd2)$, for a cartesian transformation $f : L \Rightarrow L'$ in $\Lin (\kitstr{\gpd1}, \kitstr{\gpd2})$, we automatically have $f  \yy[\kit1]$ in $\SProf(\kitstr{\gpd1}, \kitstr{\gpd2}) (L\yy[\kit1] , L', \yy[\kit1] )$.

\begin{lem}\label{lem:LinearUnitIso}
	For Boolean kits $\kitstr{\gpd1} = (\gpd1,\kit1)$, $\kitstr{\gpd2} = (\gpd2,\kit2)$ and a stabilized profunctor $P : \kitstr{\gpd1} \profto \kitstr{\gpd2}$, there is an isomorphism $\eta: P  \cong \tracelin{\rLan{P}}$.
\end{lem}

\begin{proof}
	Since we are taking a left Kan extension along a fully faithful functor $\yy[\kit1]$, $P$ is isomorphic to $(\Lan_{\yy[\kit1]}P)\yy[\kit1] = \tracelin{\rLan{P}}$.
\end{proof}

\begin{lem}\label{lem:LinearCounitIso}
	For Boolean kits $\kitstr{\gpd1}$, $\kitstr{\gpd2}$ and a linear functor $L : \StPSh(\kitstr{\gpd1}) \to \StPSh(\kitstr{\gpd2})$, there is an isomorphism $\varepsilon : \rLan{\tracelin (L)} \cong L $.
\end{lem}

\begin{proof}
	Let $X$ be in $\StPSh(\kitstr{\gpd1})$ and $b \in \gpd2$. An element of $\rLan{\tracelin(L)}(X)(b)$ is of the form $l \coendtensor[a] x$ where $l : b \to L(\yy[\kit1] (a))$ and $x \in X(a)$, we define $\varepsilon_{X,b} : \rLan{\tracelin (L)}(X,b) \to L(X,b)$ by $l \coendtensor[a] x \mapsto L(x) l$. For $l \coendtensor[a] x$ and $l' \coendtensor[a'] x'$ in $\rLan{\tracelin (L)}(X)(b)$,  $l \coendtensor[a] x=l' \coendtensor[a'] x'$ is equivalent to the existence of a unique $\alpha : a \to a'$ such that $x' \cdot \alpha =x$ and $L(\alpha) l =l'$ which implies that $\varepsilon_{X,b}$ is well-defined and injective. Let $l : b \to L(X)$ be in $\StPSh(\kitstr{\gpd2})$, by Lemma \ref{lem:LinearGenericFact},  there exists $a\in \gpd1$, $g : b \to L(\yy[\kit1] (a))$ generic and $x:\yy[\kit1] (a) \to X$ such that $ l = L(x) g$. Now, $\varepsilon_{X,b} (g \coendtensor[a] x) = L(x)g$ which implies that $\varepsilon_{X,b}$ is surjective as desired. Naturality of $\varepsilon$ follows immediately from the functoriality of $L$.
\end{proof}
\begin{thm}
	For Boolean kits $\kitstr{\gpd1} = (\gpd1,\kit1)$ and $\kitstr{\gpd2} = (\gpd2,\kit2)$, there is an
	adjoint equivalence as follows
	\begin{center}
		\begin{tikzpicture}[line join=round,yscale=0.5]
			\node (A) at (0,0) {$\SProf(\kitstr{\gpd1},\kitstr{\gpd2})$} ;
			\node (B) at (4, 0) {$\Lin(\kitstr{\gpd1},\kitstr{\gpd2})$} ;
			
			\draw [->] (A) to [bend left =30]  node [above] {$\rLan{(-)}$} (B);
			\draw [->] (B) to [bend left =30]  node [below] 
			{$\tracelin$} (A);
			\node (C) at (2,0) {$\bot\simeq$};
		\end{tikzpicture}
	\end{center}
\end{thm}
\begin{proof}
	The functors $\rLan{(-)}$ and $\tracelin$ are well-defined by Lemmas \ref{lem:QProfLinear}, \ref{lem:NatTransCartesianLin} and \ref{lem:LinearQProf}. They form an equivalence by Lemmas \ref{lem:LinearUnitIso} and \ref{lem:LinearCounitIso} and since any equivalence induces an adjoint equivalence, we obtain the desired result.
\end{proof}

We complete the picture by observing that the  intensional/extensional biequivalences that we have shown at the linear and stable levels are compatible, with the overall situation summarized in the following diagram
	\begin{center}
	\begin{tikzpicture}[line join=round,yscale=0.5]
		\node (A) at (0,0) {$\SProf(\kitstr{\gpd1},\kitstr{\gpd2})$} ;
		\node (A1) at (0,-6) {$\SEsp(\kitstr{\gpd1},\kitstr{\gpd2})$} ;
		\node (B) at (4, 0) {$\Lin(\kitstr{\gpd1},\kitstr{\gpd2})$} ;
		\node (B1) at (4, -6) {$\Stable(\kitstr{\gpd1},\kitstr{\gpd2})$} ;
		
		\draw [->] (A) to [bend left =30]  node [above] {$\rLan{(-)}$} (B);
		\draw [->] (B) to [bend left =30]  node [below] 
		{$\tracelin$} (A);
		\node (C) at (2,0) {$\simeq$};
		\draw [right hook-latex] (B) to   node  {} (B1);
		\draw [->] (A) to node  [left] {$- \circ \der[\gpd1]$} (A1);
		\draw [->] (A1) to [bend left =30]  node [above] {$\widetilde{(-)}$} (B1);
		\draw [->] (B1) to [bend left =30]  node [below] 
		{$\trace$} (A1);
		\node (C) at (2,-6) {$\simeq$};
	\end{tikzpicture}
\end{center}
which is seen to commute up to natural isomorphism. 

It is compelling that standard computational intuitions for linearity, involving resource usage, remain accessible in this setting. One has for instance that for a linear functor $L : \StPSh(\kitstr{\gpd1}) \to \StPSh(\kitstr{\gpd2})$, the trace $\trace(L) \in \SEsp(\kitstr{\gpd1}, \kitstr{\gpd2})$ has components 
\[
\trace(L)(b, \seq{a_i}_{i \in [n]}) = 
\begin{cases}
\emptyset & \text{ if } n \neq 1 \\
L(\yon(a))(b) & \text{ if } \seq{a_i}_{i \in [n]} = \seq{a} 
\end{cases}
\]
i.e. the result of the computation is determined by inspecting the `argument' $\gpd1$ exactly once, in accordance with the original ideas of Girard \cite{GirardCoherence}.

\section{Conclusion}
\label{sec:conclusion}

We have introduced a new bicategorical model of classical linear logic
based on groupoids with kits and stabilized profunctors.  We have
also presented this model in extensional form, in terms of stable
functors between subcategories of presheaves. Restricted to discrete
groupoids (\emph{i.e.}~sets), stable functors coincide with finitary polynomial
functors $\Set^I \to \Set^J$ between categories of indexed sets,
called \emph{normal functors} in the influential work of
Girard~\cite{GirardNormal}. The construction of a cartesian closed
bicategory that embeds finitary polynomial functors is a key
contribution of this work.

Some aspects of our construction should be understood more
abstractly. The bicategory of stable species of structures refines the
bicategory of generalized species \cite{FioreCartesian2008} by means
of an orthogonality construction. Refinements of this kind are
well-studied in the $1$-categorical setting and provide a rich family
of models for linear logic \cite{GlueingHylandShalk}.  Our
construction provides a new example of orthogonality and
double-glueing in a two-dimensional setting. The full theory will be
developed using the present work as a guiding example.

The notion of kit also deserves further exploration. In particular,
various families of non-Boolean kits with additional closure
conditions provide models of intuitionistic linear logic. It seems
less clear how to understand these models in extensional form.

Finally we mention some related work.
Generalized analytic and polynomial functors are considered in the
work of Garner and Hirschowitz \cite{GarnerHirschowitz}, including a
method for controlling actions in terms of stabilizers. We will
investigate the connections to our work. In a separate line of work, Finster, Lucas, Mimram
and Seiller~\cite{FinsterPolynomials} present another groupoid model,
described in the language of homotopy type theory.  The relationship
with our model should be explored, also in connection to the work of
Kock et al.~\cite{PolynomialGpdKock,KockInfinityOperads}.





\bibliographystyle{alphaurl}
\bibliography{biblio}

\end{document}